\documentclass[12pt]{book}
\usepackage[lmargin=3.cm,rmargin=4cm]{geometry}
\usepackage[utf8]{inputenc}
\usepackage[english]{babel} 

\usepackage{fouriernc}
\usepackage{afterpage}
\usepackage{authblk}
\usepackage{microtype}
\usepackage{hyperref}
\usepackage[compress]{cite}

\usepackage{amsmath,mathtools}
\usepackage{amsfonts}
\usepackage{amssymb}
\usepackage{amsthm}
\usepackage{mathrsfs}
\usepackage{bbm}
\usepackage{eufrak}
\usepackage{yfonts}

\usepackage{bmpsize-base}
\usepackage{setspace}
\usepackage{booktabs}
\usepackage{eso-pic}
\usepackage{shapepar}
\usepackage{enumerate}
\usepackage{enumitem}
\usepackage{braket}
\usepackage{graphicx}
\usepackage{epstopdf}
\usepackage{caption}
\usepackage{float}
\usepackage{subfigure}
\usepackage{bm}
\usepackage{makeidx}
\usepackage{placeins}
\usepackage{listings} 
\usepackage{color} 
\usepackage{titlepic}
	
\usepackage{soul}
\usepackage{scalerel}
\usepackage{framed,lipsum}
\usepackage{pdfpages} 

\usepackage {emptypage}
\usepackage {fancyhdr}

\fancypagestyle{My}{
\fancyhf{}
\fancyhead[LE,RO]{\small\bfseries\thepage}
\fancyhead[RE]{\small\bfseries\nouppercase{\leftmark}}
\fancyhead[LO]{\small\bfseries\nouppercase{\rightmark}}
}
\renewcommand{\chaptermark}[1]{%
  \markboth{\ifnum\value{chapter}=0 \else\thechapter.\ \fi #1}{}%
}

\makeatletter
\renewcommand{\part}{%
\clearpage
\thispagestyle{empty}
\null\vfil
\secdef\@part\@spart
}
\makeatother

\newtheorem{definition}{Definition}[section]
\newtheorem{theorem}{Theorem}[section]
\newtheorem{corollary}{Corollary}[theorem]
\newtheorem{lemma}[theorem]{Lemma}
\newtheorem{proposition}[theorem]{Proposition}
\definecolor{gray}{rgb}{92.,92.,92.}
\newcommand {\Cdot}{\vcenter{\hbox{\scalebox{1.6}{.}}}}

\newcommand {\Max}{\text{a.s.}-\max}

\title{The full replica symmetry breaking in the Ising spin glass on random regular graph}

\date{}
\author{Francesco Concetti}

\begin{document}
\frontmatter
\includepdf[page=1]{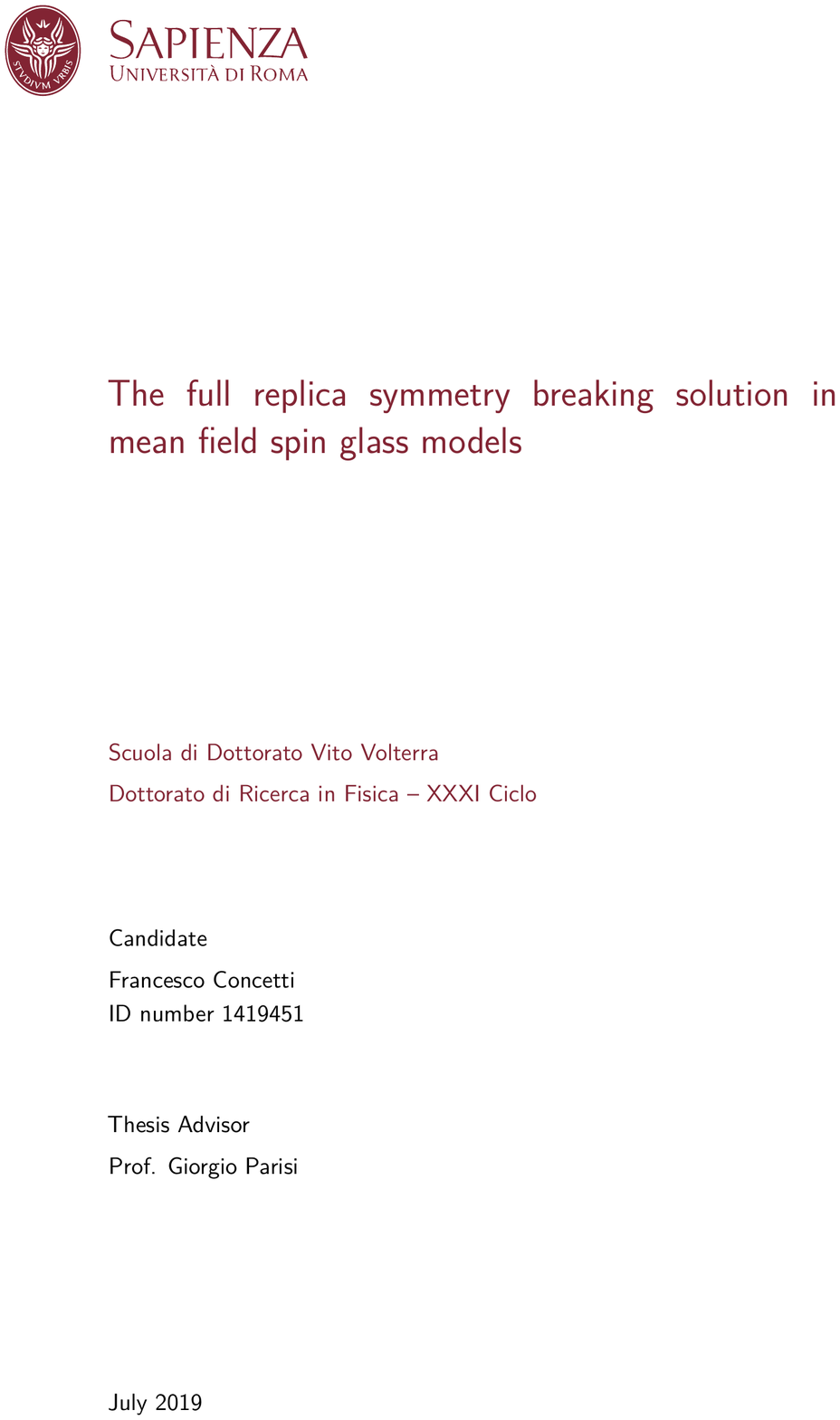}
\includepdf[page=2]{title}
\chapter*{Abstract}
\thispagestyle{empty}
\setcounter{page}{3}
This thesis focus on the extension of the Parisi full replica symmetry breaking solution to the Ising spin glass on a random regular graph. We propose a new martingale approach, that overcomes the limits of the Parisi-Mézard cavity method, providing a well-defined formulation of the full replica symmetry breaking problem in random regular graphs. 

We obtain a variational free energy functional, defined by the sum of two variational functionals (auxiliary variational functionals), that are an extension of the Parisi functional of the Sherrington-Kirkpatrick model. We study the properties of the  two variational functionals in detailed, providing a representation through the solution of a proper backward stochastic differential equation, that generalize the Parisi partial differential equation.

Finally, we define the order parameters of the system and get a set of self-consistency equations for the order parameters and the free energy.

\tableofcontents
\addtocontents{toc}{\protect\thispagestyle{empty}}

\chapter{Introduction}
\thispagestyle{empty}
\markright{Introduction}

Starting form the pioneering work of Samuel Edward and Philip Anderson \cite{EA} in 1975, spin glasses \cite{VPM,FH} acquired a dominant role in the theory of disordered systems, attracting a wide-ranging multidisciplinary interest in condensed matter, combinatorial optimization \cite{Mez85a, ParMezRRG2, Krz07}, computer science \cite{Nishimori,MonaMez} finance \cite{Bou03} and so on.

Spin glass theory, however, is still far to be completely understood. Only fully connected models have been exactly solved \cite{VPM,pSpinCri,Talagrand_Spinglass}. 

The interest in fully connected spin glass models was initially pointed out by Sherrington and Kirkpatrick (SK) with the introduction of the Sherrington-Kirkpatrick (SK) model \cite{SK1,SK2}. The SK model was solved, through the \emph{replica trick}, by G.Parisi, who introduced the concept \emph{replica symmetry breaking} (RSB) \cite{Par1_0,Par1_1,Par1_2} to describes the spin glass phase of the model.

Parisi proposed a sequence of approximated solutions \cite{Par1_1}, each of them depending on a variational order parameter with increasing dimension. Such solutions are called \emph{discrete} (or \emph{finite step})$-$RSB approximation. The order parameter of the $r-$step RSB (or simply $r-$RSB), is given by two sequences of $r$ numbers. The complete version of the Parisi solution, the so-called full$-$RSB solution, is formally obtained by the $r\to \infty$ limit of the $r-$RSB solution and the order parameter is a function.
 
The physical meaning of the RSB was further clarified as related to the decomposition of the Gibbs state in a mixture of a large number (infinite in the thermodynamic limit) of pure states\footnote{for a definition of pure state, see section 2.2 of\cite{ParisiSFT}}, that can be identified as the minima of a suitable free energy functional, depending on the local magnetizations: the so-called Thouless-Anderson-Palmer (TAP) free energy\cite{VPM,FH}.
The TAP free energy function describes a rough landscape on the space of the magnetization. The valleys of such landscape are the states of the system. 
 
It is now clear that, in the mean field approximation, the existence of multiple equilibrium states is a distinctive feature of glassiness.

If and how the RSB scheme also applies in non-fully connected systems is still debated, in spite of recent results. 

The first efforts to extend the Parisi RSB scheme to spin glass models defined on a sparse graph took place in the eighties
\cite{VB,mottishow,G2}.

Sparse graph models represent a more realistic class of mean-field models, including the notion of neighborhood which is absent in the infinite range case. They attract a significant interest also in computer science since many random optimization problems turn out to have a finite connectivity structure \cite{MonaMez}.

The $1$$-$RSB scheme was successfully extended to the Ising spin glass on sparse graphs by Parisi and Mézard (PM) with the cavity method \cite{ParMezRRG1,ParMezRRG2,ParMezRRG4,MonaMez}, improving the Bethe–Peierls method in order to deal with many equilibrium states. The approach can be easily generalized to the case with $r$ steps of RSB, by imposing the Parisi ultrametricity ansatz \cite{VPM} as in the fully connected case \cite{Panchenko2015,Panchenko2016}. 

It was proved, via interpolation arguments, that this approach provides a rigorous lower bound of the free energy \cite{FranzLeone,FranzLeone2}.

The $1$$-$RSB PM cavity method has been very successful even now, since this approach provides an algorithmic solution of certain random satisfiability problems with finite connectivity\cite{ParMezRRG2, MonaMez}. 

Within the PM cavity method formalism, the $r$$-$RSB order parameter is a distribution of $(r-1)$$-$RSB order parameters. The replica symmetric order parameter is the local fields distribution, so the $1$$-$RSB order parameter is a probability distribution over the space of distributions, the $2$$-$RSB order parameter is a distribution of distributions of distributions \cite{ParMezRRG1} and so on. As a consequence, the cavity method turns out to be inadequate to achieve a full$-$RSB theory for RRG spin glasses, indeed the $r\to \infty$ limit of the order parameter has no mathematical meaning.

The most dramatic consequence is that, actually, a complete theory that takes into account all the discrete-RSB solutions does not exist. Moreover, the high levels of replica symmetry breaking are actually numerically intractable.

The $1$$-$RSB cavity method is commonly used also to achieve an approximated solution. This approximation, however, cannot grab the presence of marginal states and then completely misses the right evaluation of such quantities that have very different properties in the marginally stable phase, as the spectrum of small oscillations, nonlinear susceptibilities and so on \cite{ParisiMarginal}. This limitedness entails that the cavity method cannot describe a glassy phase with many marginal equilibrium states. 

In this thesis we extend the full$-$RSB scheme to the computation of the free energy of the Ising spin glass on a random regular graph \cite{Bellobas}. We obtained a proper well-defined variational functional and an order parameter that can describe all the discrete-RSB solutions as a special subclass of solutions. This is the first result of this kind in diluted spin glasses.

The main contribution of this manuscript consists of the introduction of a new martingale approach \cite{YoRev} that allows us to describe several steps of replica symmetry breaking in a more compact form. This method overcomes the limits of the Parisi-Mézard cavity method, providing a well-defined formulation of the full$-$RSB problem in random regular graphs. 

We manage the progressive branching of the clusters using a martingale representation of the cavity magnetizations, improving the idea suggested in \cite{ParisiMarginal}. We reduce the computation to a composition of variational problems, where the variational parameters are martingales.
The order parameter, then, is not a deterministic distribution, as in the cavity method, but it is a stochastic process. 

We deal with non-Markovian martingales. Non-Markovianity is the source of many mathematical issues, that do not emerge in the fully connected theory of spin glasses. In particular, the free energy functional cannot be represented as the solution of a proper partial differential equations \cite{VPM}.

We obtain the analog of the Parisi functional using stochastic control theory; the Parisi equation (equation III.56 of \cite{VPM}) is replaced by a backward stochastic differential equation \cite{PaPeng}. We rigorously study this problem and provide a method to compute the derivatives and the series expansion of such a functional with respect to the parameters of the Hamiltonian.

Part \ref{P1} is devoted to provide a short review in spin glass theory, with particular attention in the problem of the computation of the free energy. In Chapter \ref{C1}, we provide a very basic introduction in the spin glass world. In Chapter \ref{C2} we describe the Parisi solution and the physical properties of fully connected spin glass.

The second part explores the problem of the Ising spin glass on Random Regular Graph. We present the model and provide a review of the cavity method ( Chapter \ref{C3}).

In Chapter \ref{C4}, we present our martingale approach. The martingale approach allows to obtain the full-RSB free energy functional. Suitable full$-$RSB order parameters are defined and the variational free energy functional is finally derived as a generalization of the Chen and Auffinger representation of the Parisi free energy for SK model \cite{ChenAuf}. The total free energy is given by the sum of two functional, that we call \emph{RSB expectations}.

In Chapter \ref{C5}, we study the mathematical properties of the RSB expectation. We show that such functional is related to the solution of a proper stochastic control problem. We prove that the solution of such problem exists, it is unique and verifies some stability conditions.

In Chapter \ref{C6} we develop some mathematical tools in order to compute the derivative and the series expansion of the RSB expectation. This chapter has remarkable importance, since the RSB$-$ is a quite "cumbersome" object.

The self-consistency equations are derived in Chapter \ref{C6}, by deriving the free energy functional with respect the order parameter.

The aim of the thesis is to provide a clear formal definition and a non-ambiguous mathematical setting for the full$-$RSB problem for spin glass models in random regular graphs. We will tackle this problem only from a technical point of view. The physical interpretation of our approach will be discussed further in next works.

We guess that the computations presented in this manuscript are essential tools to deepen the actual physical and mathematical knowledge about the problem of spin glass on diluted graphs. In particular, it is worth noting that, within our approach, the similarities between the Ising model on Bethe Lattice and the Sherrington-Kirkpatrick model \cite{SK1} are more evident. 

We are still far to provide a numerical solution of the Ising spin glass model on RRG. However, we remind that Giorgio Parisi derives the full$-$RSB equation of the SK model in 1979 \cite{Par1_0}, but Rizzo and Crisanti obtained the complete numerical solution only in 2002 \cite{RiCri}. 
 
\mainmatter
\part{Preliminaries}
\label{P1}
\thispagestyle{empty}
\chapter{Spin glass theory: a brief introduction}
\label{C1}
\thispagestyle{empty}

Spin glasses are a fascinating and interdisciplinary research topic that in the last forty years has inspired a vast scientific literature in the framework of theoretical and experimental physics, mathematics, computer science, finance and so on. Providing a worth introduction in a few pages is actually impossible. For this reason, we concentrate only on those topics that we consider to be more relevant to the aim of this thesis.

In the first section we present an overview of the original reason that motivated the study of glasses in the framework of condensed matter physics. In the second section, we provide a theoretical definition of a spin glass and discuss what we aim to "solve" when we deal with such class of models.
\section{What is a Spin Glass}
The term \emph{Spin Glass} has been introduced in \cite{Anderson_70}, referring to a particular class of dilute solutions of magnetic transition metal impurities in noble metal hosts, such as manganese (Mn) on copper (Cu) and iron (Fe) in gold (Au). Experimentally, such materials present a peculiar non-ergodic magnetic state that is neither ferromagnetic nor anti-ferromagnetic \cite{Ford82}. This different kind of magnetism is avoided in ordered systems, being a consequence of the disordered structure of such class of magnetic alloys.

In these systems, impurity moments polarizes the surrounding Fermi sea of conduction electrons of the host metal, and the induced polarization produces an effective interaction potential between impurities \cite{Ash76}. Typical behavior of such effective interaction is described by the famous Ruderman-Kittel-Kasuya-Yosida (RKKY) interaction \cite{Kittel54,Kasuya56,Yosida57}
\begin{equation}
V(r)=\frac{\cos(2\,k_F r)}{r^3}
\end{equation}
where $r$ is the distance between two impurities and $k_F$ is the Fermi wavevector of the conduction electron. Because of the random position of impurities, interactions are random and can be both positive or negative. In the low-temperature regime, impurity magnetizations "freeze" in a spatially random (or amorphous) pattern of directions. The absence of a long-range periodicity is the peculiar difference between ordinary ferromagnets or anti-ferromagnet and spin glasses \cite{FH}. 

Spin glass behavior arises also in various chemically different compounds \cite{HandMagMat24}. 

The universality of spin glass phenomena motivated the introduction of coarse-grained models that capture the essential properties of such kind of materials.

It appears that the presence of a spin glass state depends essentially on two ingredients:
\begin{itemize}
\item \emph{frustration} \cite{Van}, namely no single microscopical configuration of the system satisfies the minimum energy condition in all the interactions and the ground state is degenerate.
\item \emph{randomness} of the interaction, with competition between ferromagnetic and anti-ferromagnetic interaction.
\end{itemize} 
In 1975 \cite{EA}, Edwards and Anderson proposed an archetypal model, the so-called Edward-Anderson model (EA), that inspired all the theoretical spin glass models that has been developed until now. They considered a spin glass version of the Ising model, constituted by a set of $N\gg 1$ Ising spins $\bm{\sigma}=(\sigma_1,\,\sigma_2,\,\cdots,\sigma_N)\in\{-1,1\}^N$, with Hamiltonian
\begin{equation}
H_J[\bm{\omega}]=\sum_{(i,j)}J_{ij}\sigma_i\sigma_j
\end{equation}
where the sum $\sum_{(i,j)}\,\cdot$ runs over the links of a $d-$dimensional hypercubic lattice and the \emph{couplings} $J_{ij}$ are identical normal distributed random variables. It is generally believed that such model reproduce the main feature of a real spin glass.

After more than forty years, the EA model is still unsolved. For this reason, the theoretical physics research mainly focused on a simpler, but still highly nontrivial, class of models, the so-called \emph{mean field models}. Such models are generalizations of EA model, obtained by changing the graph where the spins lye \cite{SK1,VB}, the Hamiltonian or the kind of spins \cite{pSpinCri}.

It was early clear that such models are representative of a wider class of disordered systems and the interest in spin glass theory now goes far beyond the original motivation\cite{VPM}.
\section{The spin glass problem}
\label{SG_problem}
In this section, we describe the main issues that we aim to address when we deal with spin glass models.
 
Generally speaking, a spin glass model is a graphical model, constituted by a collection of $N\gg1$ variables $\bm{\sigma}=(\sigma_1,\cdots,\sigma_N)$ and a lower bounded random function $\bm{\sigma}\mapsto H_{N,J}(\bm{\sigma})$ that is called \emph{Hamiltonian}. The subscript $J$ denotes that the Hamiltonian depends on some external random control variables $J$, that are independent on the spins configuration $\bm{\sigma}$. 

For any given choice of $J$, the Hamiltonian $H_{N,J}$ assigns to each configuration $\bm{\sigma}$ the so-called Gibbs probability measure \cite{Huang}
\begin{equation}
\label{Gibbs-measure}
dP[\bm{\sigma}]=d\mu(\sigma_1)\cdots d\mu(\sigma_N)\frac{e^{-\beta H_{N,J}(\bm{\sigma})} }{Z_{N,J}(\beta)}
\end{equation}
The symbol $\beta$ denotes the so-called \emph{Boltzmann factor}, that in physical literature represent the inverse of the temperature. The measure $d\mu(\sigma)$ is a given probability measure that defines the nature of the spins and $Z_{N,J}(\beta)$ is the so-called \emph{partition-function}
\begin{equation}
\label{Gibbs-measure}
Z_{N,J}(\beta)= \int \text{d}\mu(\sigma_1)\cdots d\mu(\sigma_N)e^{-\beta H_{N,J}(\bm{\sigma})}.
\end{equation}
Let us also introduce the free-energy density of the system:
\begin{equation}
\label{FreeEmIntro}
f_{N,J}(\beta)=-\frac{1}{\beta N}\log\,Z_J(\beta).
\end{equation}
The computation of the partition function \eqref{Gibbs-measure} does not involve the average of the $J$s, that are "frozen" in a given configuration. For this reason the $J$s are called \emph{quenched variables}. The randomness of the $J$s constitutes a source of disorder that may induce the spin glass phase, as we discussed in the previous section.

It is worth noting that, if the spins are discrete variables, in the zero temperature limit, i.e. $\beta\to \infty$, the Gibbs measure \eqref{Gibbs-measure} concentrates around the configurations that minimize the Hamiltonian $H_{N,J}$, and the computation of the free energy \eqref{FreeEmIntro} provides the global minimum of the function $H_{N,J}$ over all the allowable configurations of the spins. This observation is the basis of the deep connection between statistical mechanics and optimization problems.

Standard statistical mechanics deals with the computation of thermodynamic limit of the free energy density, defined as follows:
\begin{equation}
\label{thermo_limit}
f_J(\beta)=\lim_{N\to\infty}f_{N,J}(\beta)\,.
\end{equation}
Because of the dependence on the quenched variables, the free energy density is a random variable. However, it appears that, under some regularity conditions on the Hamiltonian, in the $N\to \infty$ limit, the distribution of $f_J$ sharply concentrates around the mean:
\begin{equation}
\label{selfAv}
 \overline{f_J^2(\beta)}-\overline{f_J(\beta)}^2\sim O(1/N)
\end{equation}
where the symbol $\overline{\Cdot}$ denotes the \emph{quenched average}, namely the average over the random variables $J$. In physical jargon, the quantities that verify the above relation are "self-averaging”, which means that they are essentially independent of the disorder induced by the quenched variables.

We say that a spin glass model is \emph{solved} if we are able to compute the \emph{quenched free energy}, i.e. the average of the free energy density \eqref{thermo_limit} with respect the quenched variables:
\begin{equation}
f=\overline{f_J(\beta)}= \int \text{d}\mathbb{P}[J] f_J(\beta)\,.
\end{equation}
where $\mathbb{P}$ is the probability measure of the quenched variables. 

The computation of the quenched free energy is a very hard task. A rigorous treatment of this problem is a formidable challenge for mathematicians \cite{Talagrand_challange}.
\chapter{Fully connected model}
\label{C2}
\thispagestyle{empty}
The mean-field theory of spin glasses has attracted considerable interest in the last forty years, as a promising theory to describe the statistical mechanics of glassiness and disorder systems. In particular, fully connected models have represented a fruitful source of insights in this field \cite{VPM,pSpinCri}.

In the first section, we describe the Sherrington-Kirkpatrick model and the Parisi solution. In the second section, we provide a brief overview of the actual understanding of the physical properties of this kind of models. 
\section{The Sherrington-Kirkpatrick model}
The Sherrington and Kirkpatrick (SK) model \cite{SK1,SK2} is the first studied fully connected model (and maybe the most successful) that was completely solved.

The solution was derived by Parisi \cite{Par1_0,Par1_1,Par1_2}, through the introduction of a powerful mathematical tool, the so-called replica method. The physical meaning underlying the Parisi solution was further clarified \cite{VPM}. 
\subsection{The model}
The SK model is given by a set of $N$ Ising spin $\bm{\sigma}=(\sigma_1,\cdots \sigma_N)\in\{-1,1\}^N$, interacting according to the following Hamiltonian:
\begin{equation}
\label{SKHam}
H_J\left(\,\bm{\sigma}\,\right)=- \sum^N_{i=1}\sum^N_{j=i+1}J_{ij} \sigma_i \sigma_j -h \sum^N_{i=1}\sigma_i\,,
\end{equation}
where the couplings $J_{ij}$ are $N(N-1)/2$ independent and identically distributed Gaussian random variables such as:
\begin{equation}
\overline{J^2_{ij}}=\frac{1}{N}\,,
\end{equation}
where the symbol $\overline{\Cdot}$ denotes the Gaussian average of the coupling. The distribution of the couplings implies that
\begin{equation}
R_{\bm{\sigma},\bm{\tau}}=\overline{H_J\left(\,\bm{\sigma}\right)H_J\left(\,\bm{\tau}\,\right)}=\frac{1}{N}\left(\sum^N_i \sigma_i\tau_i\right)^2\leq N
\end{equation}
for any pairs $\bm{\sigma}$ and $\bm{\tau}$ in $\{-1,1\}^N$. It is proved that the above condition implies the free energy is self-averaging, according to the definition \eqref{selfAv}.

Despite his apparent simplicity, the computation of the quenched free-energy provides a formidable mathematics challenge that was completely overcome only after thirty years, requiring the introduction new physical intuitions \cite{VPM} and mathematical tools \cite{Talagrand_Spinglass}.

Sherrington and Kirkpatrick proposed a solution based on the so-called replica trick \cite{EA}. The Sherrington and Kirkpatrick approach, later referred to as the \emph{replica symmetric} (RS) solution, is wrong in the low-temperature regime, below the so-called \emph{dAR line} (de Almeida and Touless \cite{dAT}). Later on, Parisi improved the replica trick in the so-called replica method.
\subsection{The replica trick}
The replica trick, originally introduced by Edward and Anderson \cite{EA}, relies on the identity
\begin{equation}
\label{replica method}
-N\beta f=\lim_{n\to \infty} \frac{\overline{Z^n_J}-1}{n}\,.
\end{equation}
The key point of replica method consists in computing $\overline{Z^n_J}$ for an integer $n$. In such a way, $\overline{Z^n_J}$ can be consider as a partition function of a system of $n N$ particles (in replica theory jargon $n$ is the number of replicas). If $h=0$, the computation of $\overline{Z^n_J}$ gives
\begin{equation}
\overline{Z^n_J}=\sum_{\bm{\sigma}_1}\sum_{\bm{\sigma}_2}\cdots\sum_{\bm{\sigma}_n}\exp\left(N\frac{\beta^2}{2}+\frac{\beta^2}{N}\sum^n_{a=1}\sum^n_{b=a+1}\left(\sum^N_{i=1}\sigma_i^{(a)}\sigma_i^{(b)}\right)^2\right)
\end{equation}
where the spin $\sigma_i^{(a)}$ is the $i-$th spin of the $a-$th replica. After some manipulations, one gets
\begin{equation}
\overline{Z^n_J}= \int \text{d}[Q] \exp\left(N\frac{\beta^2}{2}-N\frac{\beta^2}{4}\sum_{a<b}Q^2_{a,b}+N \,\log\left(\sum_{\sigma^{(1)}\cdots \sigma^{(n)}}\exp\left(\beta^2\sum_{a<b} Q_{a,b}\sigma^{(a)}\sigma^{(b)}\right)\right) \right)
\end{equation}
where $ \int \text{d}[Q]\cdot$ is the integral in the space of $n\times n$ matrix. By the above approach, we are able to decouple the spins of different site, whilst, the average over the random couplings introduces an effective interaction between the spins of different replicas, driven by the replica matrix $Q$.

The idea is that, in the $N\to \infty$ limit, the integrand concentrate along a given configuration of the matrix $Q$, that verifies the stationary condition
\begin{equation}
\label{self_cons_rep}
Q_{p,r}=\frac{\sum_{\sigma^{(1)}\cdots \sigma^{(n)}}\sigma^{(p)}\sigma^{(r)}\exp\left(\sum_{a<b} Q_{a,b}\sigma^{(a)}\sigma^{(b)}\right) }{\sum_{\sigma^{(1)}\cdots \sigma^{(n)}}\exp\left(\sum_{a<b} Q_{a,b}\sigma^{(a)}\sigma^{(b)}\right) }
\end{equation}
The limit $n\to 0$ is achieved by imposing a particular a priori ansatz. From a mathematical point of view, such a method is non rigorous, indeed the wrong ansatz leads to a wrong solution.
\subsection{The Parisi solution}
\label{Parisi_sol}
 The solution was derived by Parisi through the introduction of a clever Replica-Symmetry-Breaking ansatz (RSB) \cite{VPM}. For an integer number of replicas $n$, he proposed to consider the case where the entries $Q_{a,b}$ of the replica matrix can take only $K$ positive values
\begin{equation}
1>q_0>q_1>q_2>\cdots >q_K.
\end{equation}
and (formula (4) in \cite{Par1_1})
\begin{equation}
\label{Parisi_integer_rep}
Q_{a,b}=q_i\quad \text{if:}\,\,\,\,\left[\frac{a}{m_i}\right]\neq\left[\frac{b}{m_i}\right]\,\,\,\,\text{and}\,\,\,\,\left[\frac{a}{m_{i+1}}\right]=\left[\frac{b}{m_{i+1}}\right]\,,\quad 0\leq i\leq K
\end{equation}
for a given sequence of integer numbers
\begin{equation}
1=m_0<m_1\cdots m_{K+1}=n
\end{equation}
such as $m_{i+1}/m_i$ is an integer number, for each $0\leq i\leq K$. Moreover he suggested that the limit $n\to 0$ may be achieved by replacing the above sequence with any sequence of $K+2$ real number, such as:
\begin{equation}
\label{Parisi_0_rep}
0=m_{K+1}<m_{K}<\cdots <m_0=1.
\end{equation}
He realized that the sequence of $m_i$ and $q_i$ can be associate the increasing function $x:[0,1]\to[0,1]$ such defined:
\begin{equation}
\label{POSK}
x(q)=\sum^K_{i=1} m_{i+1}\mathbb{1}_{(q_{i+1},q_i]}(q)
\end{equation}
In the complete formulation of the Parisi ansatz \cite{Par1_1}, namely full$-$RSB ansatz, the function \eqref{POSK} is replaced by a generic increasing function $x:[0,1]\to[0,1]$ and the free energy is given by
\begin{equation}
-f=-\frac{\beta }{2}\int _0^1 \text{d}q\,qx(q)+\phi (0,0)\,,
\end{equation}
where the function $\phi:[0,1]\times \mathbb{R}\to \infty$ is given by the solution of the following partial differential equation:
\begin{equation}
\label{start}
\begin{gathered}
\phi (1,y)=\Psi(1,y)=\frac{1}{\beta}\log(\,2\cosh\,\beta y\,)\\
\frac{\partial \phi (q,y)}{\partial q}=-\frac{1}{2} \left(\frac{\partial ^2\phi (q,y)}{\partial y^2}+\beta \,x(q)\left(\frac{\partial \phi (q,y)}{\partial y}\right)^2\right)
\end{gathered}
\end{equation}

In the next subsection we will show that the Parisi RSB ansatz is related to the decomposition of the Gibbs state in a mixture of a large number of pure equilibrium states. In particular, Parisi recognized the function $x$ as the cumulative distribution function of the mutual overlap of the magnetizations of two equilibrium states on the phase space \cite{Par2}. 

Moreover, Parisi argues \cite{Par2} that any spin glass model can be characterized by a function $x:[0,1]\to[0,1]$ related to the probability distribution of the overlaps of the magnetizations of two different pure states. Such a function is the so-called \emph{Parisi Order Parameter} (POP).

It is proved \cite{Guerra} that the function $\phi(0,0)$ is a continuous functional of the starting condition \eqref{start} and the POP $x$. The functional $(\Psi,x)\mapsto \phi(0,0)$ is the \emph{Parisi functional}.
\subsection{Overlap distribution and ultrametricity}
\label{ultrametricity}
In this subsection we describe the physical meaning of the Parisi solution, reporting the results of \cite{VPM}.

Before the Parisi solution, Thouless, Anderson and Palmer (TAP) attempted to solve the SK model by writing a variational free energy as a function of the local magnetizations $\bm{m}=(m_1,m_2,\cdots,m_N)$, associated to each site. They guessed that the equilibrium free energy, for a fixed configuration of the quenched couplings, is given by deriving the variational free energy with respect the local magnetizations and put each derivative to $0$, obtaining the so-called \emph{TAP equations}. The TAP equations may present a large number of solution ( infinite in the $N\to \infty$ limit \cite{BM}). Moreover, two different solutions may lead to two different value of free energy.

The TAP approach suggests that, in the spin glass phase, Gibbs state decomposes in the convex combination of many pure states (Chapter III of\cite{VPM}), associated to the solutions $\bm{m}^{\alpha}$ of the TAP equations. Formally we write
\begin{equation}
\braket{\Cdot}=\sum_{\alpha}w_{\alpha}\braket{\,\cdot \,}_{\alpha}
\end{equation}
where $\braket{\,\cdot \,}$ is the average with respect the Gibbs measure induced by the Hamiltonian \eqref{SKHam}, the weights $w_{\alpha}$ are the probabilities of each state $\alpha$, and $\braket{\,\cdot \,}_{\alpha}$ is the average in the pure state $\alpha$, defined in such a way
\begin{equation}
\braket{\sigma_{i_1}\sigma_{i_2}\cdots \sigma_{i_k}}_{\alpha}=m^{\alpha}_{i_1}m^{\alpha}_{i_2}\cdots m^{\alpha}_{i_k}\,.
\end{equation}
for each $k-$uple $i_1,\cdots,i_k$ of different sites. Obviously the solutions $\bm{m}^{\alpha}$ and the weight $w^{\alpha}$ depends on the random couplings $J$.

An interesting question is how the pure states differ from each other. Let us introduce the Euclidean distance between two states:
\begin{equation}
\label{distances}
d_{\alpha,\beta}=\frac{1}{N}\sum^N_{i=1}(m_i^{\alpha}-m_i^{\beta})^2\,
\end{equation}
and the overlap
\begin{equation}
q_{\alpha,\beta}=\frac{1}{N}\sum^N_{i=1}m_i^{\alpha} m_i^{\beta}.
\end{equation}
It was also argued, and recently proved \cite{PanchenkoUltra2}, that so-called \emph{Edward-Anderson overlap} $q_{\alpha,\alpha}$ is a constant that does not depend on $\alpha$. A straightforward computation yields the following relation
\begin{equation}
\label{P_J}
\sum_{\alpha}w_{\alpha}w_{\beta}\,q^k_{\alpha,\beta}=\frac{1}{N^k}\sum^N_{i_1=1}\cdots \sum^N_{i_k=1}\braket{\sigma_{i_1}\cdots \sigma_{i_k}}^2\,.
\end{equation}
wher the summation is over different $k-$uple of sites. Let $P_J(q)$ be the overlap distribution for a given choice of the random couplings $J$, formally defined as
\begin{equation}
 P_J(q)=\sum_{\alpha}\sum_{\beta}w_{\alpha}w_{\beta}\delta(q-q_{\alpha,\beta})
\end{equation}
where the symbol $\delta$ is the commonly used notation for the Dirac delta distribution (see example 3 in chapter V of \cite{Berry}).

Combining equation \eqref{P_J} with the equation \eqref{self_cons_rep}, we obtain the relation
\begin{equation}
g(y)=\int^1_0\text{d}q\,\overline{ P_J(q)}e^{y q}=\lim_{n\to 0}\frac{1}{n(n-1)}\sum^n_{a=1}\sum^n_{b=a+1} e^{y\, Q_{a,b}}
\end{equation}
where the matrix $Q$ on the right-hand side is the matrix defined in \eqref{Parisi_integer_rep} and the limit $n\to 0$ is obtained according to the Parisi ansatz. The function $g$ is the characteristic function of the averaged overlap distribution:
\begin{equation}
P(q)=\overline{P_J(q)}.
\end{equation}
In a similar manner, we can compute the joint probability distributions for several overlaps. In particular, given three pure states, denoted by the indexes $1$, $2$ and $3$, the averaged joint distribution of the overlaps $q_{1,2}$, $q_{1,3}$, $q_{2,3}$ verifies the following celebrated relation
\begin{equation}
\begin{multlined}
P(q_{1,2},q_{1,3},q_{2,3}) =\frac{1}{2}x(q_{1,2}) P(q_{1,2})\delta(q_{1,2}-q_{1,3})\delta(q_{1,2}-q_{2,3})\\+\frac{1}{2}\left\{\,P(q_{1,2})P(q_{2,3})\theta(q_{1,2}-q_{2,3})\delta(q_{1,2}-q_{2,3})+ \text{permutations}\,\right\}.
\end{multlined}
\end{equation}
This formula means that, given any three states $1$, $2$, $3$ and the distances $d_{1,2}$ and $d_{2,3}$, defined in \eqref{distances}, the distance between the state $1$ and the state $3$ verifies the \emph{ultrametric property}:
\begin{equation}
d_{1,3}\leq \max \left\{d_{1,2},d_{2,3}\right\}\,,\quad \text{with probability 1.}
\end{equation}
The discovered ultrametricity yielded several intuition about the microscopical feature of spin glasses in fully connected graphs, as we will discuss in the next section. 

From the joint distribution of the overlaps $q_{1,2}$ and $q_{3,4}$ of two distinct couples of states, we also get:
\begin{equation}
\begin{multlined}
\overline{P_J(q_{1,2})P_J(q_{3,4})}-\overline{P_J(q_{1,2})}\,\,\overline{P_J(q_{3,4})}\\=
P(q_{1,2},q_{3,4})-P(q_{1,2})P(q_{3,4})=\delta(q_{1,2}-q_{3,4})P(q_{1,2})\,.
\end{multlined}
\end{equation}
This implies that the overlap distribution $P_J(q)$ is not self-averaging. It has been proved that the above relation is a consequence of the so-called \emph{Ghirlanda-Guerra} identities \cite{GhirGuerra}. 

\section{General results fully connected model}
\label{overview}
The replica method was successfully used in many problems on fully connected networks \cite{VPM,Nishimori}.

Various models show different patterns of RSB, depending on the way the states are "distant” to each other. 
\begin{itemize}
\item The overlaps between different states can take (almost surely) only two different values. In this case, we speak about "one-step replica symmetry breaking”($1$$-$RSB) solution. The states are scattered randomly in the phase space and correspond to stable well-defined minima (genuine minima) of the free energy landscape \cite{REM, pSpinCri}. 
\item The overlaps can take a discrete number $r+1 \in \mathbb{N}$ of values in the interval $[q_m : q_M]$. In this case, we speak about "$r$-step replica symmetry breaking”($r$$-$RSB) solution. The equilibrium states exhibit a hierarchal structure, where clusters of states with a given mutual overlap are grouped in a progressively wider level of clusters with a progressively lower overlap, for $r$ levels \cite{ViraMezUltra,VPM}. Each state enters in the Gibbs decomposition with a random weight, which is generated according to the Derrida’s REM and GREM calculations (that we will explain in Part \ref{P2})\cite{DerridaGREM, ParisiMezardGREM}. Such a solution is an iterative composition of $1$$-$RSB solutions. Far as we know, this situation is very uncommon.
\item The overlaps among states can take all possible values in the interval $[q_m : q_M]$. In this case, we speak about "full replica symmetry breaking”(full$-$RSB), and it can be considered as a $r\to \infty$ of the preceding case. The equilibrium states exhibit a continuous fractal clustering, and the random weights are configurations of a Ruelle random probability cascade \cite{Ruelle,ASS,Arguin}, that provides a continuous extension of GREM. It is worth stressing that, unlike the preceding case, the equilibrium states can be arbitrarily close, and the barriers between the states may be arbitrarily small. The minima are marginal, with many flat directions ( infinite in the thermodynamic limit) \cite{DeDoKondor}.
\end{itemize}

Most of these results have been reproduced using a probabilistic iterative approach, the \emph{cavity method}, which avoids the mathematical weirdnesses of the replica method \cite{VPM, ASS}. The replica jargon, however, is used in spin glass theory, regardless the approach considered: actually, we speak about RSB if the system exhibits many pure states, organized according to one of the schemes described before.

Almost thirty years later, Talagrand rigorously proved the Parisi solution of the SK model \cite{Tala}, using the Guerra's interpolation scheme \cite{Guerra}. Soon after, Panchenko proved the ultrametricity of the states \cite{PanchenkoUltra2}, in relation to the so-called Ghirlanda-Guerra identities\cite{GhirGuerra}.

\part{The full replica symmetry breaking in the Ising spin glass on random regular graph}
\label{P2}

\chapter{The cavity method for diluted model}
\label{C3}
\thispagestyle{empty}

As we discussed in the previous chapters, after forty years of efforts, a deep understanding of the fully connected spin glasses has been achieved, both from the mathematical and physical point of view. 

If the main hallmarks of the fully connected theory appear also in real spin glasses, as the Edward-Anderson models, is still debated. The fully connected models, indeed, seem to be quite unrealistic, since each spins interacts with a diverging number of other spins and all the spins are at the same distance.

A more realistic theory of spin-glasses is given by considering models defined on a random graph with finite connectivity \cite{Bellobas}. These of random graphs are usually called sparse graphs. 

Initially invented to deal with the Sherrington-Kirkpatrick model of spin glasses (chapter V of \cite{VPM}), the cavity method is a powerful method to compute the properties of many systems with a local tree-like structure \cite{ParMezRRG1,ParMezRRG2}.

The cavity method is equivalent to the replica approach, but it turns out to have a much clearer and more direct physical interpretation.

This approach allows us to exploit the locally tree-like structure of a typical random sparse graph, reducing the problem to solving a set of recursive equations for a given set of \emph{cavity variables}. 

In Section \ref{The_model} we introduce the Ising spin glass model on RRG and describes the basic iterative techniques that motive the use of the cavity method for such kind of models. In section \ref{RS} we present the so-called Bethe-Peierles \cite{Bethe} solution and we discuss the instability of such solution. In the last section we derive the discrete-RSB Parsi-Mézard (PM) cavity method. The chapter is basically a review of the results presented in \cite{ParMezRRG1}.

\section{The model}
\label{The_model}
In this section we describe the model that we aim to deal with. We start by providing a general presentation of the diluted spin glass model, then we focus on the Ising spin glass on a Random Regular Graph. Finally we explain the general idea behind the cavity method.
\subsection{Diluted model}
A sparse graph is a random graph, where each vertex is involved in a finite number of connections, with probability $1$. Spin glass models on sparse graphs are named diluted spin glasses. 

The ensembles of sparse graphs which are usually considered are:
\begin{itemize}
\item The Cayley tree: it is a tree-like graph with no loops. It is generated starting from a central site $i = 0$ and inserting a first shell of $c$ neighbors. Then, each vertex of the first shell is connected to $c-1$ new vertices in the second shell etc. The last shell constitutes the boundary of the graph.

\item The Erd\"os-Rényi random graphs ensemble, where the graphs are generated by connecting each pair of vertices $(i j)$ with probability $z/N$. Each vertex is involved in a random number of connections, with a Poisson distribution with mean $z$.

\item The random regular graphs (RRG) ensemble given by the space of random graph, where each vertex is involved in $c$ different connections. It is usually assumed that every graph has the same probability.
\end{itemize}

A Cayley tree has a finite fraction of the total number of spins lie on the boundary. For this reason, diluted spin systems on Cayley tree presents a quite trivial behavior and they do not represents interesting model for spin glasses.

In the other two ensembles, for large $N$ and $c> 2$, or $z>1$, typical graphs presents loops with typical length of order $\sim \log N$. The probability to have finite loops vanishes in the $N\to \infty$ limit \cite{Bellobas}. As consequence, sparse graphs are locally isomorphic to a tree graphs (i.e. graphs without loop). 

As we will see later, in the high-temperature regime systems the contribution of large loops is negligible, and such kind of graph may be considered as the interior of a large tree-like graph.

 For low temperature regime, the presence of large loops may induce frustration and the system may exhibit a spin glass phase at low temperature \cite{ParMezRRG1,ParMezRRG2,MonaMez}.

Diluted spin glass models have attracted a large interest also mathematics and computer science, since they are intimately related to sparse graph codes and to random satisfiability problems, among others. 

It is generally believed that the Parisi replica symmetry breaking ansatz for fully connected spin glasses holds (in a certain way) also for such class of models. For the rest of the Chapter we will assume that it is the case.

It is worth remarking that, in spite of recent mathematical breakthroughs on this field\cite{FranzLeone,FranzLeone2,Panchenko2015,Panchenko2016}, only few exact results has been achieved so far, and only a fraction of these have been proved rigorously.
\subsection{The Ising spin glass on a random regular graph}
We consider a system of $N$ Ising spins $\bm{\sigma}:=(\sigma_1,\sigma_2,\cdots,\sigma_N)\in\{-1,1\}^N$ with Hamiltonian:
\begin{equation}
\label{hamiltonian}
H[J,\mathcal{G}_{N,c},\bm{\sigma}]=\sum_{\braket{ij}_{\mathcal{G}_{N,c}}}J_{i,j}\sigma_i\sigma_j,
\end{equation}
where the sum $\sum_{\braket{ij}_{\mathcal{G}_{N,c}}}$ is restricted only to the edges $\braket{ij}$ of a random regular graph $\mathcal{G}_{N,c}$ with connectivity $c$ ( $c-$RRG)\cite{Bellobas}.

The couplings $\{J_{i,j}\}$ are independent identically distributed quenched random variables, with symmetric distribution and finite variation, defined on the edges of the graph. Gaussian distributions with zero average or a bimodal distribution ($J=\pm c$ with equal probability) are commonly considered.

As we have already explained in section \ref{SG_problem}, we are interested in the computation of the thermodynamic limit of the free energy density:
\begin{equation}
\label{aim}
f_{J,\mathcal{G}_{N,c}}=\lim_{N\to \infty}-\frac{\log Z_{N,J,\mathcal{G}_{N,c}}}{\beta N}=\lim_{N\to \infty}-\frac{1}{\beta N}\,\log\sum_{\bm{\sigma}\in \{-1,1\}^N}e^{-\beta H[J,\mathcal{G}_{N,c},\bm{\sigma}]}\,,
\end{equation}
where $Z_{N,J,\mathcal{G}_{N,c}}$ is the partition function of the $N$ spins system. The couplings and the adjacency matrix associated to the graph $\mathcal{G}_{N,c}$ constitute the quenched disorder of the system.

The finite variance assumption assures the thermodynamic limit of the free energy exists and does not depend on the realization of the quenched disorder with probability one. For this reason, as in the fully connected case, we will concentrate on the computation of the quenched average of the free energy:
\begin{equation}
\label{aim}
f=\lim_{N\to \infty}-\frac{1}{\beta N}\overline{\log \,Z_{N,J,\mathcal{G}_{N,c}}\,}
\end{equation}
where the the overline $\overline{\Cdot}$ denotes the average with respect the random couplings distribution and all the realizations of random regular graphs, generated with uniform distribution \cite{Bellobas}.

In the limit of infinite connectivity $c \to \infty$, keeping $c\Braket{J^2} = 1$, the free energy becomes independent of the probability distribution of the $J$ and it is equivalent to the free energy of the SK model.

\subsection{Computation of the free energy}
\label{graph_operations}
In this subsection, we describe the derivation of the cavity method that was originally presented in \cite{ParMezRRG4}. In the following we denote by $c$ the connectivity of the random graph:
\begin{equation}
c=\text{connectivity of the graph.}
\end{equation}
We start by introducing an intermediate structure. Let $\mathcal{G}_{N,c,q}$ be a random graph, where $q$ randomly selected \emph{cavity vertices} have only $c-1$ neighbors, while the other $N-q$ have $c$ neighbors. The graph $\mathcal{G}_{N,c,q}$ is called \emph{cavity graph}.

We consider a system of $N$ spins, with Hamiltonian given by \eqref{hamiltonian}, defined on the random lattice $\mathcal{G}_{N,c,q}$. The spins corresponding to the cavity vertices are the \emph{cavity spins}. Note that, if the number of cavity spins $q$ is a multiple of the connectivity $c$, we can look at the cavity graph as a c-RRG where some vertices have been removed. Using this notation, a $c-$RRG can be considered as a cavity graph $\mathcal{G}_{N,c,0}$ with zero cavity spins.

The values $\sigma_1, \sigma_2,\cdots, \sigma_{q}$ of the cavity spins are kept fixed. 

We can pass from a cavity graph to another by performing one of the following graph operations:
\begin{enumerate}
\label{Iteration}
\item \label{Iteration} Iteration: by connecting a new spin $\sigma_0$ of fixed value to $c-1$ cavity spins $\sigma_1, \sigma_2,\cdots, \sigma_{c-1}$ via a new set of random couplings $J_{0,1}\cdots J_{0,c-1}$ and averaging over these $c-1$ cavity spins, one changes a $\mathcal{G}_{N,c,q}$ cavity graph into a $\mathcal{G}_{N+1,c,q-c+2}$ cavity graph.

\item Link addition: by adding a new random interaction $J_{ij}$ between two randomly chosen cavity spins $\sigma_i$ and $\sigma_j$ and averaging away over the values of these $2$ spins one changes a $\mathcal{G}_{N,c,q}$ cavity graph into a $\mathcal{G}_{N+1,c,q-c+2}$ cavity graph.

\item Site (or vertex) addition: by connecting a new spin $\sigma_0$ to $c$ cavity spins $\sigma_1, \sigma_2,\cdots, \sigma_{c}$ via a new set of random couplings $J_{0,1},\cdots,J_{0,c}$ and averaging over the values of the spin $\sigma_0$ and the cavity spins $\sigma_1, \sigma_2,\cdots, \sigma_{c}$, one changes a $\mathcal{G}_{N,c,q}$ cavity graph to a $\mathcal{G}_{N+1c,q-c}$ cavity graph.
\end{enumerate}

Spin glasses on random regular graphs can be obtained from such intermediate models on cavity graphs by performing the graph operations described above.
\begin{figure}
 \includegraphics[width=\linewidth]{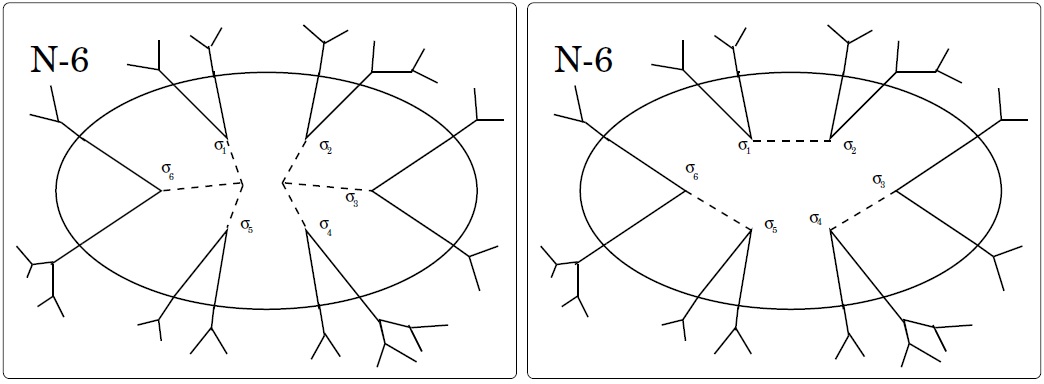}
 \caption{\small{The above figure is reprinted from \cite{ParMezRRG2}. Starting from the $\mathcal{G}_{N,3,6}$ cavity graph, one can either add two sites (left figure) and create a $\mathcal{G}_{N+2,3,0}$ graph, or add three links (right figure) and create a $\mathcal{G}_{N,3,0}$ graph.}}
 \label{cavityfigure}
\end{figure}
If one starts from a graph $\mathcal{G}_{N,c,2c}$ and performs $c$ link additions, one get a spin glass with $N$ spins on $\mathcal{G}_{N,c,0}$, that is actually the model described by the Hamiltonian \eqref{hamiltonian}, and let $F_{N}$ be the free energy of this system. On the other hand, if one starts from the same cavity graph $\mathcal{G}_{N,c,2c}$ and performs two site additions, one get a system of $N+2$ spins on the random regular graph $\mathcal{G}_{N+2,c,0}$, that is actually the model described by the Hamiltonian \eqref{hamiltonian}, and let $F_{N+2}$ be the free energy of this system (see \ref{cavityfigure}). Let $\Delta F^{(1)}$ and $\Delta F^{(2)}$ be the free energy shifts due to a site addition (\emph{vertex contribution}) and a link addition (\emph{edge contribution}), averaged over all the possible choice of cavity spins and the random couplings, then we have:
\begin{equation}
\label{shift}
F_{N+2}-F_{N}=2\Delta F^{(1)}-c\Delta F^{(2)}.
\end{equation}
If the thermodynamic limit exists, the total free energy $F_{N}$ is asymptotically linear in $N$, so we get
\begin{equation}
\label{cavFree}
f=\lim_{N\to \infty}\frac{1}{N} F_N=\lim_{N\to \infty}(F_{N+2}-F_{N})=\Delta F^{(1)}-\frac{c}{2}\Delta F^{(2)}.
\end{equation}

The computation of the free energy shifts $\Delta F^{(1)}$ and $\Delta F^{(2)}$ is the crucial point of such approach. 

The underlying intuition of the PM cavity method is given by a particular hypothesis on the cavity spins marginal distributions that allows to compute the two free energy shifts from quantities that does not depends on the whole system, but only on the cavity spins involved in the two graph operations, as we will discuss in the following sections. Such hypothesis is equivalent to the discrete-RSB ansatz in the fully connected systems. 

\section{The Bethe approximation: replica symmetric solution}
\label{RS}
The Bethe (cavity) approximation was originally proposed as a mean field theory for the ferromagnetic Ising spin model \cite{Huang,Bethe}.

The basic assumption is that, when a spin in a vertex $i$ is removed, forming a cavity in $i$, the cavity spins that were connected to the spin $i$ become uncorrelated.

This hypothesis is obviously correct for spin systems defined on a Cayley tree, indeed if we remove a vertex $i$, forming a cavity, the Cayley tree decomposes in disconnected tree-like components originating from each cavity vertices.

As we already discussed, random regular graphs converge locally to a tree in the thermodynamic limit, since the typical size of a loop diverges as $N\to \infty$. As a consequence, if we remove a vertex $i$, the distance on the lattice between two generic cavity spins is large for large $N$ ($\sim \log N$, \cite{Bellobas}). 

If there is a single pure state, then correlations in the Gibbs measure decay quickly with the distance and the Bethe approximation is asymptotically correct for $N\to\infty$. 

 Assuming the existence of a unique pure state is equivalent to impose the RS ansatz \cite{ParMezRRG1}.

In the first subsection, we derive the Bethe-Peierls equation.In the section we reformulate the Bethe-Peierls approach in a variational representation. In the last subsection we discuss the stability of this solution and the replica symmetry breaking.
\subsection{The Bethe iterative approach}
Let us consider a cavity graph, as defined in the previous subsection. For each cavity vertex $i$, let $\eta_i^{\text{(cav)}}(\sigma_i)$ be the marginal probability distribution that the value of the cavity spin in the vertex $i$ is equal to $\sigma_i$. The distribution $\eta_i^{\text{(cav)}}(\sigma_i)$ is usually called \emph{cavity distibution}.

Since an Ising spin is a $\{-1,1\}-$valued random variable, the marginal probability $\eta_i^{\text{(cav)}}(\sigma_i)$ can be expressed as:

\begin{equation}
\label{cavprob}
\eta_i^{\text{cav}}(\sigma_i)=1+\sigma_i \tanh\left(\,\beta h_i\,\right)
\end{equation}

where we have introduced the effective \emph{cavity field} $h_i$, that encodes the action onto the cavity spin $\sigma_i$ of all the other spins of the cavity graph.

It is worth noting that the cavity distribution $\eta_i^{\text{(cav)}}(\sigma_i)$ is not the marginal probability distribution of the true system, since we are considering a cavity graph. Analogously, the cavity field $h_i$ is not the true local field of the vertex $i$.

By iteration, we merge $c-1$ cavity vertices into the new cavity vertex $\sigma_0$, as in Figure \eqref{fig:merging}, and average over the merged spins. 
\begin{figure}
\center
 \includegraphics[scale=0.4]{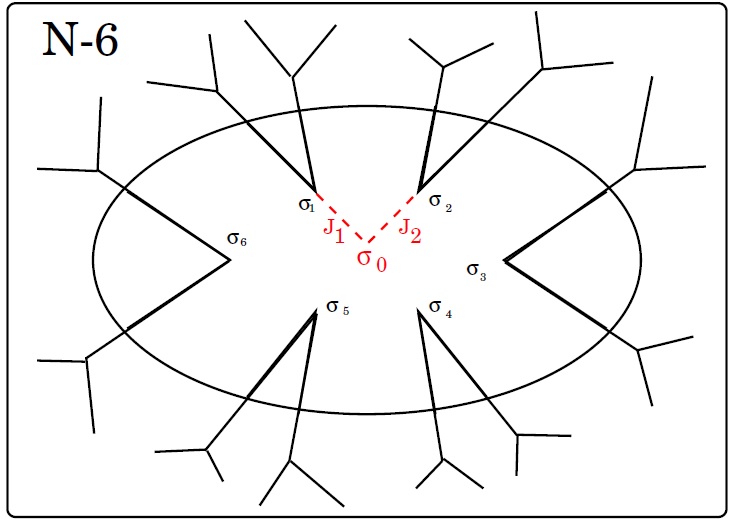}
 \caption{\small{The above figure is reprinted from \cite{ParMezRRG2}. The cavity spins $\sigma_1$ and $\sigma_2$ merges in the spin $\sigma_0$. The spin $\sigma_0$ is a new cavity spin.}}
 \label{fig:merging}
\end{figure}
Let $J_{0,1},\,J_{0,2},\,\cdots,\,J_{0,c}$ be the random couplings that connect the new cavity spin $\sigma_0$ to the old cavity spins $\sigma_1,\,\sigma_2,\,\cdots,\,\sigma_c$.

The cavity distribution of the new cavity spin can be computed from the cavity fields of the merged spins and the random couplings in such a way:
 
\begin{multline}
\label{recursion}
\eta_0^{\text{cav}}(\sigma_0)=\frac{1}{z_0}\prod^{c-1}_{i=1}\left(\,\sum_{\sigma=\pm 1}e^{\beta J_{0,i}\sigma_i\sigma_0}\left(\,1+\sigma_i \tanh\left(\,\beta h_i\,\right)\,\right)\,\right)\\=1+\sigma_0\,\tanh\left(\,\,\beta \sum^{c-1}_{i=1} u(J_{0,i},h_i)\,\,\right),
\end{multline}
where $z_0$ is the normalization constant such as 
\begin{equation}
\eta_0^{\text{cav}}(+1)+\eta_0^{\text{cav}}(-1)=1
\end{equation}
and
\begin{equation}
u(J_{0,i},h_i)=\frac{1}{\beta} \text{atanh}\left(\,\tanh\left(\beta h_i\right) \tanh\left(\beta J_{0,i}\right)\,\right).
\end{equation}
Let $h_0$ be the cavity field corresponding to the spin $\sigma_0$, then we have:
\begin{equation}
\label{BPEq}
h_0= \sum^{c-1}_{i=1} u(J_{0,i},h_i)
\end{equation} 
Because of the randomness of the couplings, the cavity fields are also random quantities. In principle, one can solve numerically the equation \eqref{BPEq}, iteratively for all the vertex of the graph, by taking a given realization of the disorder and for a finite (but large) number of spins N.

If we are interested in computing the free energy averaged over the disorder, it is more convenient to define a non random order parameter $p(h)$, i.e. the probability distribution of cavity field, defined formally by:
\begin{equation}
p(h)=\overline{\delta(h-h_i)},
\end{equation}

where the overline $\overline{\Cdot}$ is the usual average over the disorder and $h_i$ is the cavity fields of the site $i$ obtained form the equation \eqref{BPEq} for a given realization of the disorder. This distribution must be the same for all sites, since, after averaging over the disorder, all the sites are statistically equivalent. For this reason we can obtain a recursion formula for the probability distribution $p(h)$, by using the equation \eqref{BPEq}:
\begin{equation}
\label{recursion_ph}
p(h)=\overline{ \int \prod^{c-1}_{i=1}[dh_i p(h_i)] \,\, \delta \left(\,h-\sum^{c-1}_{i=1} u(J_{0,i},h_i)\,\right)}.
\end{equation}

The two averaged free energy shift contributions in \eqref{shift} can now be computed from the cavity field distribution in such a way

\begin{equation}
\label{RS1}
-\beta\,\Delta F^{(1)}=\overline{\int \prod^{c}_{i=1}[dh_i p(h_i)] \log \, \sum_{\sigma_0=\pm 1}\prod^c_{i=1}\left(\,\sum_{\sigma_i=\pm 1}e^{\beta J_{0,i}\sigma_i\sigma_0}\left(\,1+\sigma_i \tanh\left(\,\beta h_i\,\right)\,\right)\,\,\right)}
\end{equation}
and
\begin{multline}
\label{RS2}
-\beta\,\Delta F^{(2)}=\\\overline{\int \prod^{2}_{i=1}[dh_i p(h_i)]\log \sum_{\substack{\sigma_1=\pm 1\\ \sigma_2=\pm 1}}e^{\beta J_{1,2}\sigma_1\sigma_2}\left(\,1+\sigma_1 \tanh\left(\,\beta h_i\,\right)\,\right)\left(\,1+\sigma_2\tanh\left(\,\beta h_2\,\right)\right)}
\end{multline}
and the RS free energy is given by these two quantities through \eqref{cavFree}.

The true local field $H_i$ of the vertex $i$ can be computed from the cavity fields of the nearest neighbors on to the cavity graph where the vertex $i$ has been removed, in such a way:
\begin{equation}
H_i=\sum_{j\in \partial i} u(J_{i,j},h_j)\,,
\end{equation} 
where the symbol $\sum_{j\in \partial i}\Cdot$ denotes the sum over the nearest neighbors of the vertex $i$ and the quantities $J_{i,j}$ and $h_j$ are, respectively, the coupling between the spin $\sigma_i$ and $\sigma_j$ and the cavity field of the vertex $j$.

The distribution of the true local fields is then given by
\begin{equation}
P(H)=\overline{ \int \prod^{c}_{i=1}[dh_i p(h_i)] \,\, \delta \left(\,H-\sum^{c}_{i=1} u(J_{i},h_i)\,\right)}.
\end{equation}
From the true local field distribution, we can compute the Edward-Anderson order parameter $q_{EA}=\overline{\sum^N_{i=1}\Braket{\sigma_i}^2}$, by
\begin{equation}
\label{qEARS}
q_{EA}= \int \text{d}H\,P(H) \left(\,\tanh\left(\,\beta H\,\right)\,\,\right)^2=\overline{ \int \prod^{c}_{i=1}[dh_i p(h_i)] \,\, \left(\,\tanh \left(\,\beta \sum^c_{i=1}u(J_{i},h_{i})\,\right)\,\,\right)^2}
\end{equation}

Unlike the fully connected systems, the overlap does not give a complete quantitative characterization of the state of the system, since we cannot derive the local field distribution $P(H)$ directly from it, but we need to consider the cavity field distribution $p(h)$. The replica symmetric solution, therefor, already involves an order parameter which is a whole function.

\subsection{Variational formulation}

For the aim of this thesis, it is useful to reformulate the problem in order to derive the equation \eqref{recursion} from a variational principle.

If we consider the two free energy contributions \eqref{RS1} and \eqref{RS2} as functionals of the cavity field distribution $p(h)$, the resulting functional $F[p]$, obtained by substituting the relation \eqref{RS1} and \eqref{RS2} in \eqref{cavFree}
\begin{multline}
\label{freeRSvar1}
F[p]=\overline{\int \prod^{c}_{i=1}[dh_i p(h_i)] \log \, \sum_{\sigma_0=\pm 1}\prod^c_{i=1}\left(\,\sum_{\sigma_i=\pm 1}e^{\beta J_{0,i}\sigma_i\sigma_0}\left(\,1+\sigma_i \tanh\left(\,\beta h_i\,\right)\,\right)\,\,\right)}\\
-\frac{c}{2}\overline{\int \prod^{2}_{i=1}[dh_i p(h_i)]\log \sum_{\substack{\sigma_1=\pm 1\\ \sigma_2=\pm 1}}e^{\beta J_{1,2}\sigma_1\sigma_2}\left(\,1+\sigma_1 \tanh\left(\,\beta h_i\,\right)\,\right)\left(\,1+\sigma_2\tanh\left(\,\beta h_2\,\right)\right)}
\end{multline}
plays the role of a variational free energy functional and the equilibrium free energy is given by:
\begin{equation}
\label{freeRSvar2}
-\beta\,f=\underset{p}{\min}\,\,F[p],
\end{equation}
where, as usual in such kind of systems, we have a $"\max"$ variational principle, instead of a $"\min"$ variational principle as in the Gibbs principle.

 The self-consistency equation \eqref{recursion} can be obtained by imposing the stationary condition
\begin{equation}
\frac{\delta F[p]}{\delta p(h)}=0
\end{equation}
under the constrained
\begin{equation}
\label{varConstr}
 \int \text{d}h \,p(h)=1.
\end{equation}
Let us introduce the following notation:
\begin{equation}
\label{important_phi}
\phi(h_1,h_2)=\log\left( \sum_{\substack{\sigma_1=\pm 1\\ \sigma_2=\pm 1}}e^{\beta J_{1,2}\sigma_1\sigma_2}\left(\,1+\sigma_1 \tanh\left(\,\beta h_i\,\right)\,\right)\left(\,1+\sigma_2\tanh\left(\,\beta h_2\,\right)\right)\right),
\end{equation}
then we have:
\begin{multline}
\log\left( \, \sum_{\sigma_0=\pm 1}\prod^c_{i=1}\left(\,\sum_{\sigma_i=\pm 1}e^{\beta J_{0,i}\sigma_i\sigma_0}\left(\,1+\sigma_i \tanh\left(\,\beta h_i\,\right)\,\right)\,\,\right)\,\right)=\\
\phi\left(\,h_1,\sum^{c}_{i=2} u(J_{0,i},h_i)\,\right)+\sum^c_{i=1}\log\left(\, \frac{\cosh \left(\,\beta J_{0,i}\,\right)}{ \cosh\left(\,\beta u(J_{0,i},h_i)\,\right)}\,\right).
\end{multline}
By introducing a Lagrange multiplier $\mu$ for the constrained \eqref{varConstr}, we get
\begin{multline}
\frac{\delta}{\delta p(h)}\left(\,F[p]+\mu\left(\, \int \text{d}h \,p(h)-1\,\right) \,\right)=\\
c \overline{\int \prod^{c}_{i=2}[dh_i p(h_i)]\phi\left(\,h_1,\sum^{c}_{i=2} u(J_{0,i},h_i)\,\right)}-c\overline{ \int \text{d}h_2 p(h_2)\, \phi(h_1,h_2)}\\
+c \overline{ \int \text{d}h p(h)\,\log\left(\, \frac{\cosh \left(\,\beta J_{0,i}\,\right)}{ \cosh\left(\,\beta u(J_{0,i},h)\,\right)}\,\right)}+\mu=0.
\end{multline}
Such equation is obviously solved by the solution of the recursion equation \eqref{recursion} and
\begin{equation}
\mu=-c \overline{ \int \text{d}h p(h)\,\log\left(\, \frac{\cosh \left(\,\beta J_{0,i}\,\right)}{ \cosh\left(\,\beta u(J_{0,i},h)\,\right)}\,\right)}.
\end{equation}
This prove that the variational principle on the functional \eqref{freeRSvar1} is equivalent with to BP approach presented in the previous subsection.
\subsection{Instability with respect replica symmetry breaking}
In the high-temperature phase, the solution of the self consistent equation \eqref{recursion} is a simple $\delta$ Dirac function at the origin. This implies that the Edward-Anderson order parameter $q_{EA}$ \eqref{qEARS} vanishes, describing a paramagnetic phase.

The paramagnetic solution is stable for $\beta \leq \beta_c$, where the critical inverse temperature $\beta_c$ fulfills the following equation \cite{Touless}
\begin{equation}
\overline{\left(\,\tanh\left(\,\beta_c J\,\right)\,\right)^2}=\frac{1}{c-1}.
\end{equation}

In the low-temperature regime, the Bethe approximation is no more correct and the BP free energy \eqref{freeRSvar2} misses the true free energy of the system. 

An indication that the above procedure gives a wrong result is the fact that, in the limit of infinite connectivity $c \to \infty$ and $c\Braket{J^2} = 1$, the free energy \eqref{freeRSvar2} converges to the RS solution of the SK model \cite{SK1,SK2}, that is known to be unstable below the dAT critical temperature \cite{dAT}. 

Since the early years after the Parisi solution, many derivations of a replica theory for diluted spin glass have been proposed by applying the Bethe approximation to the $n-$times replicated system \cite{VB,G2,mottishow,MyArticle,ParMezRRG1}. The free energy of the replicated system is then given by
\begin{equation}
\label{freeeRep}
-\beta n f_n=\frac{c}{2}\log\left(\sum_{\bm{\sigma}}\sum_{\bm{\tau}}\rho(\,\bm{\sigma}\,)^{c-1}\rho(\,\bm{\tau}\,)^{c-1}\exp\left(\sum^n_{a=1}\beta \sigma_{a}\tau_{a}\right)\right)
\end{equation}
where $\bm{\sigma}$ and $\bm{\tau}$ denote two set of $n$ spins $\sigma_1,\cdots \sigma_n$ and $\tau_1,\cdots \tau_n$. The symbols $\sum_{\bm{\sigma}} \Cdot$ and $\sum_{\bm{\sigma}} \Cdot$ denote the sums over all the $2^n$ configurations of the two sets of $n$ spins $\bm{\sigma}$ and $\bm{\tau}$ and $\rho(\Cdot)$ is the replica order parameter, consisting in a function of the $n$ spins satisfying the following self-consistency equation:
\begin{equation}
\label{repEq}
\rho(\bm{\sigma})=\frac{\sum_{\bm{\tau}}\rho(\,\bm{\tau}\,)^{c-1}\exp\left(\sum^n_{a=1}\beta \sigma_{a}\tau_{a}\right)}{\sum_{\bm{\tau}}\rho(\,\bm{\tau}\,)^{c-1}}.
\end{equation}
The quenched average of the free energy is formally given by the replica limit
\begin{equation}
\label{repLim}
f=\underset{n\to 0}{\lim} f_n.
\end{equation}
As in the SK model, the limit $n\to 0$ cannot be computed in a rigorous way, but one have to consider a proper ansatz about the dependence of the function $\rho$ to the set of $n$ spins.

The RS ansatz is given by imposing that the function $\rho$ depends on the spins through the sum $\Sigma=\sum^n_{a=1} \sigma_a$:
\begin{equation}
\label{RSMotti}
\rho(\bm{\sigma})=\rho_{\text{RS}}(\Sigma)= \int \text{d}h p(h) e^{\beta h\Sigma}.
\end{equation}
After some straightforward calculations, it can be shown that the RS solution \eqref{RSMotti} is equivalent to the BP solution \cite{ParMezRRG1}. The limit \eqref{repLim} of the free energy \eqref{freeeRep} gives the free energy \eqref{freeRSvar1} and the equation \eqref{repEq} is equivalent to the equation \eqref{BPEq}.

The stability of the Bethe assumption, presented in the previous section, can be investigated, in a non-rigorous way, using replica method, by adding to the RS solution \eqref{RSMotti} a "small" perturbation that breaks the replica symmetry, in such a way:
\begin{equation}
\rho(\bm{\sigma})=\rho_{\text{RS}}(\Sigma)+\delta \rho(\bm{\sigma}),
\end{equation}
where
\begin{equation}
\delta \rho(\bm{\sigma})=\sum^n_{k=1}\sum_{a_1,\cdots,a_k} q_{a_1,\cdots,a_k} \sigma_{a_1}\cdots \sigma_{a_k}.
\end{equation}
and computing the second variation of the free energy \eqref{freeeRep} with respect the perturbation $\delta \rho$ around the RS solution.

Mottishow proved that for $\beta>\beta_c$, the Ising spin glass on RRG undergoes a transition toward a RSB phase. The Mottishow stability condition is acltually a generalization of the dAT stability line for the fully connected models.

Unfortunately, the Parisi solution cannot be extended to diluted models in a simple way. The non-Gaussianity of the cavity fields implies that the problem involves an infinity of order parameters which are the multi-spin overlaps \cite{VB,G2,mottishow,MyArticle}. This is the main issue when we deal with diluted models. 

In the replica theory framework, approximated solutions have been obtained near the critical temperature\cite{mottishow} or in the limit of large connectivities \cite{G2}, where the contribution of multi-spin overlap is negligible, and the solution can be derived by applying the Parisi ansatz to the two spins overlap matrix.

\section{The $1-$RSB cavity method}
\label{1RSBcav}
The spin glass dAT transition from the RS to the RSB phase is characterized by the divergence of the spin-glass susceptibility $\chi_{\text{SG}}$, defined as:
\begin{equation}
\chi_{\text{SG}}=\frac{1}{N}\sum^N_{i=1}\sum^N_{j=1}\overline{\left(\Braket{\sigma_i \sigma_j}-\Braket{\sigma_i}\Braket{\sigma_j}\right)^2},
\end{equation}
where the angular brakets means that we take the thermal average with respect the Boltzmann-Gibbs distribution \cite{VPM}.

In diluted models, divergence of spin glass susceptibility is equivalent to the condition that the following quantity vanishes 
\begin{equation}
\lambda=\underset{l\to \infty}{\lim}\frac{1}{l} \log\left(\frac{\Big(\Braket{\sigma_i \sigma_j}-\Braket{\sigma_i}\Braket{\sigma_j}\Big)^2}{(c-1)^l}\right)
\end{equation}
where $\sigma_0$ is a reference starting spin and $\sigma_l$ is a spin at distance $l$ from $\sigma_0$. If $\lambda<0$, the correlation between spins decays with the distance, by contrast, if $\lambda\geq 0$ the system is characterized by long range correlations, invalidating the hypothesis of the Bethe approximation.

In the low temperature regime, the Gibbs state is given by a statistical mixture of pure states. As a consequence, the connected correlation functions do not vanish with the distance \cite{ParMezRRG1,ParisiMarginal}. As in fully connected systems (see section \ref{ultrametricity}), this is the basic mechanism underlying the RSB phenomenon.

In this section we present the $1-$RSB solution for diluted spin glasses using the cavity method. Such approach considers the presence of multiple pure states and defines an iterative approach on the graph, under some assumptions. The results presented in this section were originally derived in \cite{ParMezRRG1}. 

The key assumption of the RSB cavity method is that there is a one to one correspondence among the pure states before and after graph operations described in Subsection \ref{graph_operations}, at least among the pure states with lowest free-energy. Under this hypothesis, the iteration \eqref{BPEq} is fulfilled within each given pure state.

However, the free-energy shifts due to graph operations may differ from state to state, so after iteration the pure states with the lowest free energy (i.e. the equilibrium states) may be different than the ones before. In this case, we cannot map the old equilibrium states to the new ones by simply applying the iterative rule \eqref{BPEq} for the new cavity field.

Here we present a detailed description of the $1-$RSB cavity method. This kind of RSB pattern is actually the only one that has been implemented numerically in an efficient way \cite{ParMezRRG1,ParMezRRG2,ParMezRRG4,SP}.

\subsection{The $1-$RSB hypothesis}
Here we state the basic assumptions of the $1-$RSB solutions to spin glass in the cavity method, derived from the $1-$RSB theory of fully connected spin glass\cite{VPM}.
\begin{enumerate} 
\item \label{firstHyp} The cavity spins are uncorrelated within each pure state. Given a pure state labeled by the index $\alpha$, after merging $c-1$ cavity vertices in a single new vertex $0$, the corresponding cavity field $h_i^{\alpha}$ is given by:
\begin{equation}
\label{IterOneState}
h^{\alpha}_0=\sum^c_{i=1} u(J_{0,i},h_j)\,,
\end{equation}
where $h^{\alpha}_1,\cdots,h^{\alpha}_{c-1}$ are the cavity field of the merged vertices.
\item \label{secondHyp} The free energies $\{f^{\alpha},\alpha\in \mathbb{N}\}$ of the pure states are independent and identically distributed random variables, with an exponential probability distribution given by 
\begin{equation}
\label{expoDist}
\text{d}\rho(f)= \text{d}f\,\beta x\, \exp\left(\,x\beta (f-f_R)\,\right).
\end{equation}
where $f_R$ is a reference free energy and $x$ is the Parisi 1 RSB order parameter. In the thermodynamic limit, two different pure states, with the same free energy per particle, may have a finite random difference in the total free energy \cite{ParMezRRG1,REM}, so each equilibrium state $\alpha$ has a random probability given by:
\begin{equation} 
\label{randomweight}
w_{\alpha}=\frac{e^{-\beta f^{\alpha}}}{\sum_{\alpha}e^{-\beta f^{\alpha}}}.
\end{equation}
The family $\{ w_{\alpha},\alpha \in\mathbb{N}\}$ is a point process. Note that probabilities $w_{\alpha}$ and $w_{\gamma}$ of two different states are not independent, since the sum over all $(w_{\alpha})$ is normalized:
\begin{equation}
\sum_{\alpha}w_{\alpha}=1.
\end{equation}
The hypothesis of exponential distribution of the free energy is simply derived by analogy with the Parisi 1RSB solution for fully connected spin glass \cite{VPM,REM,DerridaGREM,ASS} and it can be justified by considering the pure states as extremes of the free energy, and with the Gumbel universality class for extremes \cite{26ParMezRRG3,27ParMezRRG3}.
\item \label{thirdHyp} On a given vertex $i$, the population of cavity fields on various states $\{h^{\alpha}_i,\,\alpha\in \mathbb{N}\}$ is a population of independent random variables generated according to the same distribution $\pi_i(h)$. The fact that the cavity fields of different states are independent from each other is the basic hallmark of the PM 1 RSB ansatz (as we have discussed in section \ref{overview}).

Note that, in the RS solution, the probability distribution $p(h)$ describes how the cavity fields are distributed over the realization of the quenched disorder, whilst $\pi_i(h)$ describes the distribution of the cavity fields $h_i$ over the pure states. We will call $\pi_i(h)$ the \emph{one site cavity field distribution}.

 The distribution $\pi_i(h)$ depends on the vertex label $i$. The order parameter of the system is a site-independent functional $\mathbb{P}[\pi(h)]$, that represents the probability that the cavity fields $h^{\alpha}_i$, associated to a randomly picked cavity vertex $i$, is generated with probability distribution $\pi_i(\cdot)=\pi(\cdot)$ (a probability distribution of distributions).
\end{enumerate}

The 1 RSB ansatz, obviously, does not cover the possibility that the cavity fields of different pure states may be correlated. This situation appears in higher orders of RSB.

As we expect, the RS solution can be considered as a special case of the 1$-$RSB solution, where the order parameter distribution $\mathbb{P}$ is a functional Dirac delta around the solution of the BP recursion equation \eqref{recursion}.
\subsection{$1-$RSB equations}
 The aim of this section is obtaining a recursive equation, for the 1$-$RSB order parameter $\mathbb{P}$.

By hypothesis \ref{firstHyp} the BP recursion \eqref{BPEq} is valid within a pure state, so we start by the RS iterative equation for the state-dependent and site-dependent variables, then we get a new iterative equation for the state-independent and site-dependent variables and finally we get the equation for the state-independent and site-independent order parameter.

 Let us denote with $h_1^{\alpha},h_2^{\alpha},\cdots,h_{c-1}^{\alpha}$ the cavity fields corresponding to a given pure state $\alpha$; by hypothesis \ref{firstHyp} they are uncorrelated.

For each state, we repeat the same iteration operation that we have described for the RS solution. We merge the $c-1$ cavity vertices to a new vertex with a new cavity field $h_0^{\alpha}$, given by equation \eqref{IterOneState}. 

The free energy of the cavity system changes after iteration by a quantity $\Delta \phi_{\text{iter}}^{\alpha}$ depending on the cavity fields $\bm{h}=(h_1^{\alpha},h_2^{\alpha},\cdots,h_{c-1}^{\alpha})$ and the couping $\bm{J}$
\begin{multline}
\label{Fiter}
-\beta \Delta \phi^{\alpha}_{\text{iter}}=-\beta \Delta \phi_{\text{iter}}(\bm{h}^{\alpha})\\=\log\left(\,2\cosh\left(\,\sum^{c-1}_{i=1} u(J_{0,i},h^{\alpha}_i)\,\right)\,\right)+\sum^{c-1}_{i=1}\log\left(\, \frac{\cosh \left(\,\beta J_{0,i}\,\right)}{ \cosh\left(\,\beta u(J_{0,i},h^{\alpha}_i)\,\right)}\,\right)\,.
\end{multline}
Note that the new cavity field $h^{\alpha}_0$ and the free energy shift $\Delta F^{\alpha}_{\text{iter}}$ are correlated. 

By the hypothesis \ref{thirdHyp}, the family of pairs $\{\,(h^{\alpha}_0,\Delta \phi_{\text{iter}}^{\alpha}),\alpha\in \mathbb{N}\,\}$ is a family of independent and identically distributed random pairs with distribution $P_0(h_0,\Delta \phi)$, given by:
\begin{equation}
\label{recursion11}
 P_0(h_0,\Delta \phi)=\int \prod^{c-1}_{i=1} [\text{d}h_i \pi_i(h_i)] \delta\left(\Delta \phi-\Delta \phi_{\text{iter}}(\bm{h})\right) \delta\left(\,h_0-\sum^{c-1}_{i=1}u(J_{0,i},h_i)\,\right)
\end{equation}
Let us call $f^{\alpha}$ the free energy of the state $\alpha$ on the cavity graph before the addition of the new spin $\sigma_0$. By hypothesis \ref{secondHyp}, the free energies $f^{\alpha}$ and $f^{\gamma}$ of two different states are independent random variables with the exponential distribution \eqref{expoDist}. 

For each state $\alpha$ the free energy $f^{\alpha}$ and the free energy shift $\Delta f^{\alpha}_{\text{iter}}$ are independent random variables, since the free energy shift depends on the new couplings $J_{0,1},\cdots J_{0,c-1}$ of the added links between the old cavity spins $\sigma_{0,1},\cdots \sigma_{0,c-1}$ and the new cavity spin $\sigma_0$, while the free energy $f^{\alpha}$ depends only on the old couplings that were already present on the cavity graph before iteration.

Let $g^{\alpha}$ be the new free energy in the state $\alpha$, after adding the spin $\sigma_0$ in the iteration
\begin{equation}
g^{\alpha}=f^{\alpha}+\Delta f^{\alpha}.
\end{equation}
The family of the new free energies $\{\,g^{\alpha},\alpha \in \mathbb{N}\,\}$ is obviously a family of independent and identically distributed random variables. 

A standard argument of the cavity method \cite{VPM,ParMezRRG1} proves that the new free energy $g^{\alpha}$ is uncorrelated with the local field $h_0^{\alpha}$, with an exponential distribution as in \ref{secondHyp} (equations (44) \cite{ParMezRRG1}). In particular, the computation of the joint distribution $R_0(h_0,g)$ of the local field and the new free energy yields:
\begin{multline}
\label{independence}
R_0(h_0,g)=\\
\beta x \int \text{d}(\Delta \phi)\,df \,\,e^{\beta x (f-f_{\text{ref}} )} \theta(f_{\text{ref}}-f) P(h_0,\Delta f)\delta(g-f-\Delta \phi)\\
\propto 
\beta x\, e^{\beta x (g-f'_{\text{ref}}) } \pi_0(h_0)
\end{multline}
 where $\pi_0(h_0)$ is the one site cavity field distribution of of the cavity vertex $0$, given by (equation (45) in \cite{ParMezRRG1}):
\begin{equation}
\label{cavityDist1}
\pi_0(h_0)\propto \int \text{d}(\Delta \phi) P(h_0,\Delta F) e^{-\beta x \Delta \phi }\,.
\end{equation}
The symbol $\propto$ means that the left-hand member is proportional to the right-hand member.

The exponential distribution \eqref{expoDist}, then, is stable under iteration. The only effect of the iteration is a shift of the reference free energy.

Note that, the fact that the new free energy and the new cavity field are uncorrelated is a consequence of the exponential distribution \eqref{expoDist} of the free energy.

Substituting the equation \eqref{recursion11} in \eqref{cavityDist1}, one get the iterative equations for the one site probability vertex distributions.
\begin{equation}
\label{recursion12}
\pi_0(h_0)=\frac{1}{Z_{0}}\int \prod^{c-1}_{i=1} [dh_i \pi_i(h_i)] e^{-\beta x \Delta \phi_{\text{iter}}(h_1,\cdots,h_{c-1})} \delta\left(\,h_0-\sum^{c-1}_{i=1}u(J_{0,i},h_i)\,\right)\,,
\end{equation}
where $Z_{0}$ is the constant such as $\int dh\pi_0(h)=1$.

As in the $RS$ solution, we have obtained a set of iterative equations over a population of one site variables, that in this case are the one site probabilities $\pi_i(h)$, that depends on the realization of the disorder and on the vertex.

 The above iterative relation induces a recursion equation on the global probability distribution $\mathbb{P}$ defined in the hypothesis \ref{thirdHyp} stated in the previous subsection.

Let us define the following functional of the old cavity field distribution:
\begin{equation}
\mathfrak{F}^{(1)}[\pi_1,\cdots,\pi_{c-1}]=\frac{1}{Z_{0}}\int \prod^{c-1}_{i=1} [dh_i \pi(h_i)] e^{-\beta\,x \Delta \phi_{\text{iter}}(h_1,\cdots,h_{c-1})} \delta\left(\,h_0-\sum^{c-1}_{i=1}u(J_{0,i},h_i)\,\right)
\end{equation} 
then we formally have:
\begin{equation}
\label{recursion13}
\mathbb{P}^{(1)}[\pi_0]= \overline{\int \prod^{c-1}_{i=1}\big [\,\text{d}\pi_i \mathbb{P}^{(1)}[\pi_i]\,\big]\,\delta\big[\pi_0- \mathfrak{F}^{(1)}[\pi_1,\cdots,\pi_{c-1}]\,\big]}
\end{equation}
where $\delta$ is a Dirac delta function. The above equality means the functional $\mathfrak{F}^{(1)}[\pi_1,\cdots,\pi_{c-1}]$ and the random single site probability $\pi$ are equivalent in distribution. 

The equation \eqref{recursion12} has been deeply studied for all the past decade. The solution has been obtained by population dynamic algorithms on populations of cavity fields $\{h_i^{\alpha},\alpha\in\mathbb{N}\}$ defined for each vertex of a generic sparse graph \cite{ParMezRRG1,ParMezRRG2}. The population dynamic algorithm has been improved in a propagation algorithm, the so-called Survey Propagation \cite{ParMezRRG2}, in order to deal with a given realization of the disorder, without needing to consider the quenched average. The discussion of the population dynamic algorithms and Survey Propagation is beyond the aim of this thesis. Note that if $x=0$ and the one site distributions do not fluctuate from site to site, the iterative 1 RSB equations \eqref{recursion12} recovers the RS recursion equation \eqref{recursion}. 
 \subsection{Variational formulation and free energy}
\label{Var_1RSB}
As in the RS case, we can derive the self-consistency equation \eqref{recursion13} and the free energy by a variational principle on a proper $1-$RSB free energy functional $\Phi_{\text{1RSB}}\left[\mathbb{P},x_1\,\right]$, depending on the Parisi 1 RSB parameter $x_1$ and on the order parameter $\mathbb{P}$. A detailed discussion about these topics is in \cite{ParMezRRG4}

The $1$RSB free energy functional is a generalization of the RS one. 

Within a given pure state $\alpha$, the free energy shift due to vertex addition and ling addition are respectively given by:
\begin{gather}
\begin{multlined}
-\beta\Delta \phi_{\text{vertex}}^{\alpha}=-\beta\Delta \phi_{\text{vertex}}(h_1,\cdots,h_{c})= \phi \left(\,h^{\alpha}_1,\sum^{c}_{i=2} u(J_{0,i},h^{\alpha}_i)\,\right)\\+\sum^c_{i=1}\log\left(\, \frac{\cosh \left(\,\beta J_{0,i}\,\right)}{ \cosh\left(\,\beta u(J_{0,i},h^{\alpha}_i)\,\right)}\,\right)\,,
\end{multlined}\\
-\beta\Delta \phi_{\text{edge}}^{\alpha}=-\beta\Delta \phi_{\text{edge}}(h^{\alpha}_1,h^{\alpha}_2)=\phi(\,h^{\alpha}_1,h^{\alpha}_2)\,.
\end{gather}
where the function $\phi$ on the right.hand side is defined in \eqref{important_phi}.

The total free energy shifts are given by averaging over all the states. Each state $\alpha$ has a random probability $w_{\alpha}$, defined in \eqref{randomweight}, that does not depends on the cavity fields. Thus we get
\begin{equation}
\label{free11}
\Delta F^{(1)}=
\overline{\mathbb{E}_{\mathbb{P},\pi}\left[\,\mathbb{E}_w\left[\sum_{\alpha}w_{\alpha}\Delta \phi_{\text{vertex}}^{\alpha}\right]\right]}
\end{equation}
and
\begin{equation}
\label{fee12}
\Delta F^{(2)}=
\overline{\mathbb{E}_{\mathbb{P},\pi}\left[\,\mathbb{E}_w\left[\sum_{\alpha}w_{\alpha}\Delta \phi_{\text{edge}}^{\alpha}\right]\right]}
\end{equation}
where $\mathbb{E}_w$ is the average with respect the random probabilities, $\overline{\Cdot}$ denotes, as usual, the average over the couplings, and $\mathbb{E}_{\mathbb{P},\pi}$ is the average with respect the field and the single site distribution $\pi$.

The most remarkable property of the point process $\{w_{\alpha},\alpha \in \mathbb{N}\}$ is the \emph{quasi-stationarity} \cite{ASS,ParMezRRG1}. Let $\{a_{\alpha},\alpha\in \mathbb{N}\}$ be a family of independent and identically distributed positive random variables, that are independent of the random weights $w_{\alpha}$. We have the following identity
\begin{equation}
\mathbb{E}_a\left[\,\mathbb{E}_w\left[\sum_{\alpha\in \mathbb{N}}w_{\alpha}a_{\alpha}\right]\,\right]=\left(\mathbb{E}_a\left[\,a^x\,\right]\,\right)^{\frac{1}{x}},
\end{equation}
where $\mathbb{E}_a$ denotes the average over the variables $a_{\alpha}$.

Using the above identity, we can rewrite the two free energy shift as explicit functional of the order parameter $\mathbb{P}$ and the parameter $x$:
\begin{equation}
-\beta \Delta F^{(1)}\\=\frac{1}{x} \overline{\int \prod^{c}_{i=1}\big [\,\text{d}\pi_i \mathbb{P}[\pi_i]\,\big]\,\log \left(\int \prod^{c}_{i=1} [\text{d}h_i \pi_i(h_i)]\,e^{-\beta x \Delta \phi_{\text{vertex}}(h_1,\cdots,h_{c})\,}\right)}
\end{equation}
\begin{equation}
-\beta\Delta F^{(2)}=\frac{1}{x} \overline{\int \prod^{2}_{i=1}\big [\,\text{d}\pi_i \mathbb{P}[\pi_i]\,\big]\,\log \left(\int \prod^{2}_{i=1} [\text{d}h_i \pi_i(h_i)]e^{-\beta x \Delta \phi_{\text{edge}}(h_1,h_2)\,}\right)}\,.
\end{equation}
The 1 RSB variational free energy functional is finally given by:
\begin{equation}
\label{1RSB_varFun}
\Phi_{\text{1RSB}}\left[\,\mathbb{P},x\,\right]=\Delta F^{(1)}-\frac{c}{2}\Delta F^{(1)}
\end{equation}
As in the RS case, the equilibrium free energy is given by
\begin{equation}
f=\underset{\mathbb{P},x}{\max}\,\, \Phi_{\text{1RSB}}\left[\,\mathbb{P},x\,\right].
\end{equation}
For a fixed value of the Parisi 1$-$RSB parameter $x$, the maximum is attained by the solution of the equation \eqref{recursion13}, or by imposing the stationary condition over the order parameter $\mathbb{P}$, under the constrained $ \int \text{d}\mathbb{P}[\pi]=1$:
\begin{gather}
\frac{\delta}{\delta \mathbb{P}[\pi]}\left(\Phi_{\text{1RSB}}\left[\,\mathbb{P},x\,\right]+\mu \left( \int \text{d}\mathbb{P}[\pi]-1\right)\,\right)=0,\\
 \int \text{d}\mathbb{P}[\pi]=1.
\end{gather}
where $\mu$ is a Lagrange multiplier.
\section{Extension of the PM cavity method to many steps RSB}
The PM cavity method can be formally generalized to higher number of steps of Replica Symmetry Breaking.
We start by considering the 2$-$RSB case and then we will present the general case.

As in the 1$-$RSB case, all the manipulations below are based on the assumption that there is a one to one correspondence among the pure states before and after graph operations presented in section \ref{graph_operations}. 

In the 2$-$RSB ansatz, the pure states are assumed to be grouped in \emph{clusters}. A cluster $\alpha$ is a random event that associates, at each cavity vertex $i$, a family of random cavity fields $\{h_i^{(\alpha,\alpha')},\alpha'\in \mathbb{N}\}$ with the property to be exchangeable, i.e. the distribution of the whole family is invariant under permutation of the fields in the family. Two families of cavity fields $\{h_i^{(\alpha,\alpha')},\alpha'\in \mathbb{N}\}$ and $\{h_i^{(\gamma,\gamma')},\gamma'\in \mathbb{N}\}$ are assumed to be independent and identically distributed.

Within a given cluster $\alpha$, a cavity field $h_i^{(\alpha,\alpha')}$ is generated with probability distribution $\pi_i^{(\alpha)}\big(h_i^{(\alpha,\alpha')}\big)$.

The family of marginal distributions associated with each cluster $\{\pi_i^{(\alpha)},\alpha\in\mathbb{N}\}$ is a family of independent and identically distributed random probability distributions. Let $\mathbb{P}_i$ be the distribution of the random probabilities $\pi_i^{\alpha}$.

We associate to each state $(\alpha,\alpha')$ a random free energy $f_{\alpha,\alpha'}=f^{(0)}_{\alpha}+f^{(1)}_{\alpha,\alpha'}$, where $f^{(0)}_{\alpha}$ and $f^{(1)}_{\alpha,\alpha'}$ are independent random variables with exponential distributions, respectively 
\begin{equation}
\label{expoDist21}
\text{d}\rho^{(0)}\big(f^{(0)}\big)= \text{d}f\,\beta x_1 \, e^{\beta x_1\big(f^{(0)}-f^{(0)}_{\text{ref}}\big)}
\end{equation}
and
\begin{equation}
\label{expoDis22}
\text{d}\rho^{(1)}\big(f^{(1)}\big)= \text{d}f\,\beta x_2 \,e^{\beta x_2\big(f^{(1)}-f^{(1)}_{\text{ref}}\big)}
\end{equation}
with $x_2<x_1$. Te free energy contributions $f^{(0)}_{\alpha}$ and $f^{(1)}_{\alpha,\alpha'}$ are independent for different label $\alpha$ and $\alpha'$.

In the 1$-$RSB ansatz, the RS iterative equations \eqref{BPEq} are assumed to be valid within a pure state. In the 2$-$RSB ansatz we assume that the iterative 1$-$RSB iterative equation \eqref{recursion12} holds within a cluster, so we get an equation for the $\alpha-$dependent one site cavity field distribution:
\begin{equation}
\label{recursion21}
\pi^{\alpha}_0(h_0)=\frac{1}{Z_0^{\alpha}}\int \prod^{c-1}_{i=1} [dh_i \pi^{\alpha}(h_i)] e^{-\beta x_1 \Delta \phi_{\text{iter}}(h_1,\cdots,h_{c-1})} \delta\left(\,h_0-\sum^{c-1}_{i=1}u(J_{0,i},h_i)\,\right).
\end{equation}
Where $Z^{\alpha}_0$ is the normalization constant, that depends on the site label $0$ and the cluster $\alpha$:
\begin{equation}
Z_0^{\alpha}=Z_0[\pi^{\alpha}_1,\cdots,\pi^{\alpha}_{c-1}]=\int \prod^{c-1}_{i=1} [dh_i \pi_i^{\alpha}(h_i)] e^{-\beta x_1 \Delta \phi_{\text{iter}}(h_1,\cdots,h_{c-1})}
\end{equation}
Let us also introduce the functional
\begin{equation}
\Delta \Phi^{\alpha}_{\text{iter}}=\Delta \Phi_{\text{iter}}[\pi^{\alpha}_1,\cdots,\pi^{\alpha}_{c-1}]=-\frac{1}{\beta x_1}\log\,Z_0[\pi^{\alpha}_1,\cdots,\pi^{\alpha}_{c-1}]
\end{equation}
The above quantity plays the role in the second step of RSB as the free energy shift \eqref{Fiter} in the 1-RSB.

By proceeding in a similar manner to the 1$-$RSB case, we may obtain an iterative equation for the $\alpha-$ independent one site distribution $\Pi_i[\pi]$:
\begin{equation}
\label{recursion23}
\Pi_0[\pi_0]=\mathfrak{F}^{(2)}[\Pi_1,\cdots,\Pi_{c-1}]
\end{equation}
where
\begin{multline}
\label{recursionF}
\mathfrak{F}^{(2)}[\Pi_1,\cdots,\Pi_{c-1}]\\=\frac{1}{\mathfrak{Z}_0} \int \prod^{c-1}_{i=1}\big [\,\text{d}\pi_i \Pi_i[\pi_i]\,\big] e^{-\beta x_2 \Delta \Phi_{\text{iter}}[\pi_1,\cdots,\pi_{c-1}]}\,\delta\big[\pi_0- \mathfrak{F}_{1}[\pi_1,\cdots,\pi_{c-1}]\,\big]
\end{multline}
and $\mathfrak{Z}_0$ is the normalization constant.

We finally get a recursion equation for the non-random 2$-$RSB order parameter
\begin{equation}
\mathbb{P}^{(2)}[\Pi_0]=\overline{\int \prod^{c-1}_{i=1}\big [\,\text{d}\Pi_i \mathbb{P}^{(2)}[\Pi_i]\,\big] \delta\big[\Pi_0- \mathfrak{F}^{(2)}[\Pi_1,\cdots,\Pi_{c-1}]\,\big]}
\end{equation}
Such procedure can now be easily generalized to more step of RSB.

In the case of $k$ step of RSB ($k-$RSB), the states are assumed to be grouped in a hierarchical structure of clusters \cite{VPM}.

We label the bigger clusters with an index $\alpha_0\in \mathbb{N}$; ech sub-cluster inside a cluster $\alpha_0$ is labeled by two indices $(\alpha_0,\alpha_1)$, with $\alpha_1\in \mathbb{N}$; sub-sub-cluster inside a sub-cluster $(\alpha_0,\alpha_1)$ are labeled by $(\alpha_0,\alpha_1,\alpha_2)$, with $\alpha_2\in \mathbb{N}$ and so on.

The $n-$RSB iterative equations, for $n\leq k$, are valid within each cluster of level $n$ of clustering (the bigger clusters are the highest levels and the pure states are the lower level).

Given the functional $\Delta\Phi^{(k-1)}$, at some step $k-1$ of RSB, then: 
\begin{multline}
\label{recursion33}
-\beta\Delta \Phi^{(k)}_{\text{iter}}\left[\Pi^{(k)}_1,\cdots,\Pi^{(k)}_{c-1}\right]\\=\log\,\int \prod^{c-1}_{i=1} \left[\,\,\text{d}\Pi^{(k-1)}_i \Pi_i^{(k)}\left[\Pi^{(k-1)}_i\right]\,\,\right] e^{-\beta x_{k-1} \Delta \Phi^{(k-1)}_{\text{iter}}\left[\Pi^{(k-1)}_1,\cdots,\Pi^{(k-1)}_{c-1}\right]} \,.
\end{multline}
and
\begin{multline}
\mathfrak{F}^{(k)}\left[\Pi^{(k)}_1,\cdots,\Pi^{(k)}_{c-1}\right]\\=e^{\beta x_{k-1}\Delta \Phi^{(k)}_{\text{iter}}\left[\Pi^{(k)}_1,\cdots,\Pi^{(k)}_{c-1}\right]}\int \prod^{c-1}_{i=1}\left[\,\,\text{d}\Pi^{(k-1)}_i \Pi_i^{(k)}\left[\Pi^{(k-1)}_i\right]\,\,\right] e^{-\beta x_{k-1} \Delta \Phi^{(k-1)}_{\text{iter}}\left[\Pi^{(k-1)}_1,\cdots,\Pi^{(k-1)}_{c-1}\right]}\,.
\end{multline}
The equation of the $k-$RSB order parameter is finally given by:
\begin{equation}
\mathbb{P}^{(k)}\left[\Pi_0^{(k)}\right]=\overline{\int \prod^{c-1}_{i=1}\left [\,\,\text{d}\Pi^{(k)}_i \mathbb{P}^{(k)}\left[\Pi^{(k)}_i\right]\,\,\right] \delta\left[\Pi^{(k)}_0- \mathfrak{F}^{(k)}\left[\Pi^{(k)}_1,\cdots,\Pi^{(k)}_{c-1}\right]\,\right]}
\end{equation}
As we have already discussed, the $k-$RSB theory involves an order parameter that is a distribution of distributions of distributions...

A population dynamic algorithm \cite{ParMezRRG1,ParMezRRG2,MonaMez} is actually intractable already at the $2-$RSB level. Moreover the $k\to \infty$ limit of the $k-$RSB equations cannot be obtained with the cavity method and the order parameter $\mathbb{P}^{(k)}$ is not well-defined in this limit. 

In the next chapter we will provide a different approach that allows to obtain full$-$RSB theory

\chapter{The full Replica Symmetry Breaking free energy}
\label{C4}
\thispagestyle{empty}
In this chapter we present the main results of this thesis: the full-RSB formula for Ising spin glass on random regular graph.

We start by reformulating the discrete-RSB scheme of the previous chapter in a martingale formalism \cite{YoRev}. The power of the martingale approach becomes clear in section \ref{Cont_extension}, where we obtain the full-RSB free energy functional, by using a variational representation principle à la Boué-Dupuis \cite{BoueDepuis}.

In the last section we reduce the problem by considering only a certain class of martingale. 

Under some restriction on the parameters of the theory, the full$-$RSB formula of the free energy may recover the $k$-RSB formula, for any number of step of $k$. This implies that 
This chapter is reprinted from \cite{MyArty}.

\section{Martingale formulation of the discrete-RSB theory}
In this section, we describe the discrete-RSB scheme for this model. 

In the first subsection, we define the $r$$-$RSB cavity free energy functional for sparse graphs. We provide an accurate description of the Parisi RSB ansatz for diluted models from the point of view of pure states probabilities and the cavity field distributions \cite{Panchenko2015,Panchenko2016}; we recall the notion of discrete Ruelle random probability cascade, or GREM \cite{DerridaGREM, Ruelle, ASS, Arguin}. 

In the second subsection, we recast the progressive steps of replica symmetry breaking in a discrete time recursive map, that generalizes the Parisi replica computation for the SK models \cite{Par1_0,Par1_1,Par1_2, VPM}. 

In the third section, we prove that the free energy obtained by the recursive map is equivalent to the one obtained with the cavity method.

In the last subsection, we derive a new variational representation of the $r$$-$RSB free energy, using a progressive iteration of the Gibbs variational principle: the iterated Gibbs variational principle.

The iterated Gibbs variational principle is the basic tool in the derivation of the full$-$RSB theory.
\subsection{Pure states distributions}
\label{finitersbpara}
Let us assume that the system has many equilibrium states, that are labeled by an index $\bm{\alpha}$. The cavity spins are uncorrelated within a given state $\bm{\alpha}$, leading to a factorized cavity spins distribution, that depends on the label $\bm{\alpha}$. Since each spin $\sigma_i$, with $1\leq i \leq N$, can take only two values, the cavity probability distribution, for a given state $\bm{\alpha}$, depends only on the cavity magnetization $m_{i|\bm{\alpha}}$ or, equivalently, on the cavity field $h_{i|\bm{\alpha}}=1/\beta\,\,\text{atanh}\, m_{i|\bm{\alpha}}$. The cavity fields depend on the random couplings, so they are also random quantities and their distribution is not known a priori. The equilibrium free energy is finally given by the Gibbs state, that is a statistical mixture of the states $\bm{\alpha}$.

The cavity free energy functional is given by \cite{ParisiMarginal}
\begin{equation}
\label{cavityFE}
\Phi= \overline{\mathbb{E}\log\left(\frac{\sum_{\bm{\alpha}}\xi_{\bm{\alpha}}\Delta^{\text{(v)}}(\bm{J},\bm{h}_{\bm{\alpha}})}{\sum_{\bm{\alpha}}\xi_{\bm{\alpha}}}\right)}-\frac{c}{2}\overline{\mathbb{E} \log\left(\frac{\sum_{\bm{\alpha}}\xi_{\bm{\alpha}}\Delta^{\text{(e)}}(J_{1,2},h_{1|\bm{\alpha}},h_{2|\bm{\alpha}})}{\sum_{\bm{\alpha}}\xi_{\bm{\alpha}}}\right)}\,.
\end{equation}
Here $\bm{J}=(J_{0,1},J_{0,2},\cdots,J_{0,c})$ and $\bm{h}_{\bm{\alpha}}=(h_{1|\bm{\alpha}},h_{2|\bm{\alpha}},\cdots,h_{c|\bm{\alpha}})$ and the variables $\{\xi_{\bm{\alpha}}\}_{\bm{\alpha}}$ are the (non-normalized ) statistical weights of the states. All the $c+2$ couplings in the functional \eqref{cavityFE} are independent.

The functions $\Delta^{\text{(v)}}$ and $\Delta^{\text{(e)}}$ are defined as:

\begin{gather}
\label{deltaterms}
\Delta^{\text{(v)}}(\bm{J},\bm{h}_{\bm{\alpha}})=\cosh\big(\beta U_c(\bm{J},\bm{h}_{\bm{\alpha}})\,\big)\prod^c_{i=1}\frac{\cosh(\beta J_{0,i})}{\cosh\big(\,\beta u(J_{0,i},h_{i|\bm{\alpha}})\,\big)}\,,\\
\Delta^{\text{(e)}}(J_{1,2},h_{1|\bm{\alpha}},h_{2|\bm{\alpha}})=\cosh(\beta J_{1,2})\big(1+\tanh(\beta J_{1,2})\tanh(\beta h_{1|\bm{\alpha}})\tanh(\beta h_{2|\bm{\alpha}})\,\big)\,,
\end{gather}

with
\begin{gather}
\label{U}
u(J_{1,2},h_{2|\bm{\alpha}})=\frac{1}{\beta}\text{atanh}\big(\tanh(\beta J_{1,2})\tanh(\beta h_{2|\bm{\alpha}})\,\big)\,,\\
U_c(\bm{J},\bm{h}_{\bm{\alpha}})=\sum^c_{i=1}u(J_{0,i},h_{i|\bm{\alpha}})\,.
\end{gather}
The overline $\overline{\Cdot}$ stands for the average over the quenched disorder and the expectation value $\mathbb{E}$ is over all the cavity fields and the random weights $\{\xi_{\bm{\alpha}}\}_{\bm{\alpha}}$.

The contribution to the free energy \eqref{cavityFE} that depends on the function $\Delta^{\text{(v)}}$ is usually called \textit{vertex contribution}, whilst the contribution depending on $\Delta^{\text{(e)}}$ is the \textit{edge contribution}.

The equilibrium free energy is, formally given by \cite{ParisiMarginal}
\begin{equation}
\label{EqFreeEnRPC}
-\beta F=\min_{\mathbb{P}(\{\xi _{\bm{\alpha}}\}_{\bm{\alpha}},\{h_{i|\bm{\alpha}}\}_{i,\bm{\alpha}})}\Phi\,,
\end{equation}
where the supremum must be take over the set of all the possible probability distributions of the cavity fields $\{h_{i|\bm{\alpha}}\}_{i,\bm{\alpha}}$ and the random weight of the state $\{\xi _{\bm{\alpha}}\}_{\bm{\alpha}}$.
This set is huge and too general, then further assumptions are needed to face up the problem.

In the Parisi-Mézard RSB ansatz, the sum $\sum_{\bm{\alpha}}\Cdot$ runs over the leaves of an infinitary rooted taxonomic tree and $\bm{\xi}:=\{\xi_{\bm{\alpha}}\}_{\bm{\alpha}}$ is a collection of positive random variables generated by a Ruelle random probability cascade defined along the tree; for each site $i$, the set $\{h_{i|\bm{\alpha}}\}_{\bm{\alpha}}$ is a random hierarchical population of fields generated along the same tree. Such hierarchical populations are independent for different site index $i$ and identically distributed. 

More specifically, the $r$$-$RSB ansatz, for a finite integer $r$, is defined as a generalization of the Aizenman-Sims-Starr (ASS) \cite{ASS}
 construction of the hierarchal Random Overlap Structucture (ROSt) for the SK model \cite{Panchenko2016,PanchenkoExchange}. 

Let $X$ be a non-decreasing sequence of $r+2$ numbers, for some $r\in\mathbb{N}$:
\begin{equation}
\label{X_sequence}
0=x_0 \leq x_1\leq \cdots\leq x_r \leq x_{r+1}=1\,.
\end{equation}
We first define a Poisson point process $\bm{\xi}^{(1)}:=\{\xi_{\alpha_1}^{(1)};\,\alpha_1 \in \mathbb{N}\}$ on $[0,\infty)$, with density given by $\rho(d\xi)=x_1 \xi^{-x_1-1} d\xi$; such a process is usually referred to as $REM_{x_1}$.

Next, for each $\alpha_1$, a $REM_{x_2}$ process $\bm{\xi}_{\alpha_1}^{(2)}:=\{\xi_{(\alpha_1,\alpha_2)}^{(2)};\,\alpha_2 \in \mathbb{N}\}$ is generated, independently for different values of $\alpha_1$. We then iterate the procedure: at the $n-$th level, up to $n=r$, independent realizations of the $REM_{x_n}$ process $\bm{\xi}_{(\alpha_1,\cdots,\alpha_{n-1})}^{(n)}:=\{\xi_{(\alpha_1,\cdots,\alpha_{n-1},\alpha_{n})}^{(n)};\,\alpha_{n} \in \mathbb{N}\}$ are generated for each of the distinct values of the multi-index $(\alpha_1,\alpha_2,\cdots,\alpha_{n-1})$ of the previous iteration.
Let also introduce the quantity $\xi_{\star}=1$.

Such structure defines an infinitary rooted taxonomic tree of depth $r$, with the vertex set given by
\begin{equation}
\mathcal{A}=\{\star\}\cup\mathbb{N}^1\cup\mathbb{N}^2\cup\cdots\cup\mathbb{N}^r
\end{equation}
with each vertex $(\alpha_1,\cdots,\alpha_{n-1})$ branching to the vertices $(\alpha_1,\cdots,\alpha_{n-1},\alpha_n)$, for all $\alpha_n\in \mathbb{N}$. We denote by $|\bm{\alpha}|$ the level, i.e. the lenght, of the multi-index $\bm{\alpha}\in\mathcal{A}$, with $|\star|=0$.

Each $\bm{\alpha}=(\alpha_1,\alpha_2,\cdots,\alpha_r)\in\mathbb{N}^r$, at the boundary, identifies a path along the tree, defined as:
\begin{equation}
\bm{\alpha}\mapsto p(\bm{\alpha})=\left\{\,\star,(\alpha_1),(\alpha_1,\alpha_2),\cdots,(\alpha_1,\alpha_2,\cdots,\alpha_r)\,\right\}.
\end{equation}
The vertex $\star$ is the starting point of all the paths. 

The $r-$step Ruelle random probability cascade, for the sequence $X$, or GREM$_X$, is then defined as the point process $\{\xi_{\bm{\alpha},r}\}_{\bm{\alpha}\in\mathbb{N}^r}$ such that:
\begin{equation}
\xi_{\bm{\alpha},r}=\prod_{\bm{\beta}\in p(\bm{\alpha})}\xi_{\bm{\beta}}^{(|\bm{\beta}|)}=\xi_{\star}^{(0)}\xi_{(\alpha_1)}^{(1)}\xi_{(\alpha_1,\alpha_2)}^{(2)}\cdots \xi_{(\alpha_1,\alpha_2,\cdots,\alpha_r)}^{(r)}.
\end{equation}

Note that a rigorous definition of the Ruelle probability cascade point process requires the reordering, for each level $0\leq k\leq r+1$, of the random variables, generated in $REM_{x_k}$, in a decreasing order \cite{DerridaGREM, Ruelle, ASS, Arguin, Panchenko2015}.

For any given site $i$, the population of cavity fields $\{h_{i|\bm{\alpha},r}\}_{\bm{\alpha}\in\mathbb{N}^r}$ is a random array, that is assumed to be independent of the random weights $\{\xi_{\bm{\alpha},r}\}_{\bm{\alpha}\in\mathbb{N}^r}$
 and \emph{hierarchical exchangeable}, i.e. the distribution is invariant under permutations that preserve the tree structure; such assumption is the key of the Parisi-Mézard ansatz \cite{ParMezRRG1,ParMezRRG2,ParMezRRG4,MonaMez,Panchenko2015,Panchenko2016,FranzLeone,FranzLeone2,ParisiMarginal} and it turns out to be exact, assuming the validity of the Ghirlanda Guerra identities \cite{PanchenkoExchange}. 

Furthermore, by general argument, we can safely argue that all the cavity fields have zero mean and are almost surely bounded:

\begin{equation}
\mathbb{E}[h_{i|\bm{\alpha},r}]=0
\end{equation}
\begin{equation}
\mathbb{E}[|h_{i|\bm{\alpha},r}|]<\infty\quad \text{for all allowed $i$ and $\bm{\alpha}$}\,.
\end{equation}

In the ASS hierarchal ROSt \cite{ASS}, the populations of cavity fields are generated by defining, independently for each site $i$, a set of independent Gaussian variables, labelled by the vertices of the taxonomic tree $\mathcal{A}$, and representing each cavity fields $h_{i|\bm{\alpha},r}$ by the sum over the Gaussian variables corresponding to the vertices of the path $p(\bm{\alpha})\subset\mathcal{A}$.

Gaussianity is too restrictive for the actual model, and a more general distribution must be considered.

A more generic hierarchical exchangeable random array can always be represented by the hierarchical version of the Aldous-Hoover theorem, presented in \cite{Austin,AustinPanchenko}.

As in the ASS work, for any given index $i$, let $\{W^{(|\bm{\alpha}|)}_{i|\bm{\alpha}}\}_{\bm{\alpha}\in \mathcal{A}}$ be a collection of independent and identical normal distributed random variables\footnote{In the original works by Austin and Panchenko \cite{Austin,AustinPanchenko}, a random array is generated by a function of unifom random variables on $[0,1]$. A uniform random variable, however, can be generated in distribution as a function of a Gaussian variable, than the representation presented here is equivalent to Austin and Panchenko representation. }, labelled by the vertices of the taxonomic tree $\mathcal{A}$, and consider a measurable function $h:\mathbb{R}^{r+1}\to \mathbb{R}$, which we will refer to as \emph{cavity field functional}. The cavity field population, at the site $i$, can be generated by presenting each cavity fields $h_{i|\bm{\alpha},r}$ as follow:
\begin{equation}
\label{AustinRep}
h_{i|\bm{\alpha},r}=h_r\big(\,\big\{\,W^{(|\bm{\beta}|)}_{i|\bm{\beta}}\big\}_{\bm{\beta}\in p(\bm{\alpha})}\,\,\big)=h_r\big(\,W^{(0)}_{i|\star},W^{(1)}_{i|(\alpha_1)},W^{(2)}_{i|(\alpha_1,\alpha_2)},\cdots,W^{(r)}_{i|(\alpha_1,\alpha_2,\cdots,\alpha_r)}\,\big)\,.
\end{equation}
 
The variable $W^{(0)}_{i|\star}$ is the root random variable of the site $i$ and it is shared amongst all the $\bm{\alpha}$s. The collections $\big\{\,W^{(|\bm{\beta}|)}_{i|\bm{\beta}}\big\}_{\bm{\beta}\in p(\bm{\alpha})}$ are independent for different sites $i$.
 
Taking the average over all the random quantities, the functional $\Phi$ will depends only on the sequence $X$ and on the cavity field functional $h$. The equilibrium free energy is given by the extremizing the functional $\Phi$ with respect to such two parameters. 

The cavity field functional is the actual order parameter of the model and encodes entirely the Parisi-Mézard ansatz for the cavity fields distributions inside the pure states \cite{Panchenko2016}, as shown for the $1-$RSB case in subsection \ref{appendix}.

The cavity field functional turns out to be a handier order parameter than the recursive tower of distributions on the set of distributions presented in the Parisi-Mézard original works \cite{ParMezRRG1,ParMezRRG2,ParMezRRG4,MonaMez} and can be easily extended to the full$-$RSB case. It is worth noting, however, that the representation \eqref{AustinRep} is redundant; indeed there are many choices of the function $h$ that will produce the same array in distribution \cite{Panchenko2016}.

If the cavity field functional is linear, the representation \eqref{AustinRep} recovers the ASS hierarchal ROSt scheme. As discussed in the next subsection, in case of linearity, or additive separability at least, the free energy can be represented as the solution of a proper partial (integro-)differential equation, like the Parisi solution of the SK model \cite{Par1_1,VPM}.
Additive separability, however, fails to fit the results emerging at $1-$RSB levels \cite{ParMezRRG1,ParMezRRG2,ParMezRRG4,MonaMez}. 

As we shall see below, the martingale approach to the cavity method allows dealing with a generic cavity field functional $h$, leading to a well-defined full$-$RSB theory, with an explicit definition of the order parameter, an explicit representation of the functional \eqref{cavityFE} and a proper self-consistency mean-field equation.

Note that, for a fixed state $\bm{\alpha}$ and $i$, the distribution of the $r+1$ random variables $\{W^{(|\bm{\beta}|)}_{i|\bm{\beta}}\}_{\bm{\beta}\in p(\bm{\alpha})}$ does not depend explicitly on the multi-index $\bm{\alpha}$, and, in the following, we will drop it away without ambiguities:
\begin{equation}
\label{realizationfixed}
\{W^{(|\bm{\beta}|)}_{i|\bm{\beta}}\}_{\bm{\beta}\in p(\bm{\alpha})}\longrightarrow \{W^{(m)}_i\}_{m\leq r}:= \,\,\{W_i^{(0)},W_i^{(1)},W_i^{(2)},\cdots,W_i^{(r)}\}\,,\\
\end{equation}

We will also indicate with $\{W^{(m)}_i\}_{m\leq n}$ the set of all the variables $W_i^{(m)}$, along a given path, from the level $m=0$, to the level $m=n\leq r$:
\begin{equation}
\begin{gathered}
\label{realizationfixed2}
 \{W^{(m)}_i\}_{m\leq n}:= \,\,\{W_i^{(0)},W_i^{(1)},W_i^{(2)},\cdots,W_i^{(n)}\}\,.
\end{gathered}
\end{equation}

We have not defined the probability setting which the cavity field functional is defined on. Let us consider the space $\Omega_r:=\mathbb{R}^{r+1}$ as the sample space of the random variables $\{W^{(m)}\}_{m\leq r}$, endowed with the Borel $\sigma -$algebra $\mathcal{B}_r$ and with the filtration $\{\mathcal{B}^{W}_n\}_{n\leq r}$ such as, for each $0\leq n\leq r$, $\mathcal{B}^{W}_n$ is the $\sigma-$algebra generated by the random variables $\{W^{(m)}\}_{m\leq n}$ (natural filtration)\cite{Billingsley}. Let $\nu$ denote the one dimensional normal distribution and $\mathbb{W}_r:=\nu^{\otimes r}$ be the product probability measure of $r$ normal distribution on $(\Omega_r,\mathcal{B}_r)$. In this formalism, the cavity field functional is a real-valued $\mathcal{B}_r-$measurable function $h:\Omega_r\to \mathbb{R}$.

 We also define the probability spaces $(\Omega_r,\mathcal{B}_r)^{\otimes 2}$ and $(\Omega_r,\mathcal{B}_r)^{\otimes c}$, given, respectively, by the $2-$fold and $c-$fold Cartesian product of the probability space $(\Omega_r,\mathcal{B}_r)$, together with the respectively filtrations $\{(\mathcal{B}^{W}_n)^{\otimes 2}\}_{n\leq r} $ and $\{(\mathcal{B}^{W}_n)^{\otimes c}\}_{n\leq r}$.

\subsection{Composition of non-linear expectation values}

The average in the vertex and the edge contributions of the free energy \eqref{cavityFE} over the random weights $\{\xi_{\bm{\alpha}}\}_{\bm{\alpha}}$ can be evaluated by exploiting the quasi-stationarity property of Ruelle RPC under the cavity dynamics \cite{ASS}. In particular, the average of a population of hierarchical random variables, weighted by the Ruelle random probability cascade configurations $\{\xi_{\bm{\alpha},r}\}_{\bm{\alpha}\in\mathbb{N}^r}$, is equivalent to a recursive composition of non-linear expectation values over only such variables, so we can get rid of the cumbersome random weights.

By this property, we can compute the average over all the states by considering only one path on the taxonomic tree $\mathcal{A}$, so that we can omit the state label $\alpha$ in our computation. 

The edge and the vertex free energy contributions have a quite similar form; then, for simplicity, we will use a unique notation representing both the cases. 

In the rest of the chapter, the edge/vertex superscript $\Cdot^{\text{(e/v)}}$ will denote that a given result must be considered both for the two contributions. The symbol $(2/c)$ will denote that one has to consider $2$ or $c$ variables respectively for the edge and vertex contributions. The boldface symbols $\bm{J}$, $\bm{h}_{\bm{\alpha},r}$ and $\bm{W}^{(m)}$represent arrays of $2$ or $c$ independent random variables, one for each site which the considered function depends on, according to the definitions \eqref{deltaterms}. The quantities without the edge/vertex superscript have the same probability law in both the contributions.

The vertex and the edge contributions satisfies the following identity:
\begin{equation}
\mathbb{E}\log\left(\frac{\sum_{\bm{\alpha}}\xi_{\bm{\alpha}}\Delta^{\text{(e/v)}}(\bm{J},\bm{h}_{\bm{\alpha},r})}{\sum_{\bm{\alpha}}\xi_{\bm{\alpha}}}\right)= \int_{\mathbb{R}^{(2/c)}}\left(\prod^{(2/c)}_{i=1} d\nu\big(W^{(0)}_{i}\big)\,\right) \, \phi^{\text{(e/v)}}_0\big(\bm{J},\bm{W}^{(0)}\big)\,.
\end{equation}
The function $\phi^{\text{(e/v)}}_0\big(\bm{J},\bm{W}^{(0)}\big)$ is obtained from the backward map, with starting condition given by
\begin{equation}
\label{start}
\phi^{\text{(e/v)}}_{r}\big(\,\bm{J},\big\{\bm{W}^{(m)}\big\}_{m\leq r}\,\big)\\=\log\left(\,\Delta^{\text{(e/v)}}\big(\,\bm{J},\bm{h}_r\big(\,\big\{\bm{W}^{(m)}\big\}_{m\leq r}\big)\,\big)\,\right)\,,
\end{equation}
and the recursion
\begin{multline}
\label{start1}
\phi^{\text{(e/v)}}_{n}\big(\,\bm{J},\big\{\bm{W}^{(m)}\big\}_{m\leq n}\,\big)\\
=\frac{1}{x_{n+1}}\log \mathbb{E}_{\mathbb{W}_r^{\otimes (2/c)}}\left[\,\exp\left(\,x_{n+1} \phi^{\text{(e/v)}}_{n+1}\big(\,\bm{J},\big\{\bm{W}^{(m)}\big\}_{m\leq n+1}\,\big)\,\right)\,\big|\big\{\bm{W}^{(m)}\big\}_{m\leq n}\right]\\\text{for} 1\leq n\leq r-1
\end{multline}
and finally
\begin{equation}
\label{start2}
\phi^{\text{(e/v)}}_{0}(\bm{J},\bm{W}^{(0)}\,)=\frac{1}{x_{1}}\log \mathbb{E}_{\mathbb{W}_r^{\otimes (2/c)}}\left[\,\exp\left(\,x_{1} \phi^{\text{(e/v)}}_{1}\big(\bm{J},\bm{W}^{(0)},\,\bm{W}^{(1)}\,\big)\,\right)\,\big|\bm{W}^{(0)}\right]\,,
\end{equation}
where
\begin{equation}
\bm{h}_r\big(\,\big\{\bm{W}^{(m)}\big\}_{m\leq r}\big)=\left(\,h_1\big(\,\big\{W_1^{(m)}\big\}_{m\leq r}\big),\,\cdots\,,h_{c}\big(\,\big\{W_c^{(m)}\big\}_{m\leq r}\big)\,\right)\,,
\end{equation}
and
\begin{equation}
\bm{h}_r\big(\,\big\{\bm{W}^{(m)}\big\}_{m\leq r}\big)=\left(\,h_1\big(\,\big\{W_1^{(m)}\big\}_{m\leq r}\big),\,h_{2}\big(\,\big\{W_2^{(m)}\big\}_{m\leq r}\big)\,\right)\,,
\end{equation}
respectively for the vertex and adge contribution.

The symbol $\mathbb{E}_{\mathbb{W}_r^{\otimes (2/c)}}[\cdot|\{\bm{W}^{(m)}\}_{m\leq n}]$ is the expectation over the random variables corresponding to the last $r-n$ steps $\bm{W}^{(n+1)},\bm{W}^{(n+2)},\cdots,\bm{W}^{(r)}$, taking the values of the random variables $\bm{W}^{(0)},\bm{W}^{(1)},\cdots,\bm{W}^{(n)}$ fixed. The subscript $\mathbb{W}_r^{\otimes (2/c)}$ means that the expectation value is with respect the probability measure given by the $2-$fold or $c-$fold (according to the case) product of the probability measure $\mathbb{W}_r$. 

The functional $\phi^{\text{(e/v)}}_{n}$, for each level $0\leq n\leq r$, depends only on the first $n$ random variables $\{\bm{W}^{(m)}\}_{m\leq n}$, so the process $\phi^{\text{(e/v)}}:=\{\phi^{\text{(e/v)}}_{n}(\bm{J},\{\bm{W}^{(m)}\}_{m\leq n}\,)\}_{1\leq n\leq r}$ is adapted (or non-anticipative) to the natural filtration $\{(\mathcal{B}^{W}_n)^{\otimes (2/c)}\}_{n\leq r}$. 

The process $\phi^{\text{(e/v)}}$ will be called \emph{edge/vertex free energy process}. It is easy to show that the free energy process is a supermartingale \cite{YoRev}. 

For each level n, with $0\leq n\leq r$, we calls the the first $n$ random variables $\{\bm{W}^{(m)}\}_{m\leq n}$ as \emph{past random variables}.

Let us define also the $r-$steps free energy stochastic process $\phi$ such that:
\begin{equation}
\phi=\phi^{\text{(v)}}-\frac{c}{2}\phi^{\text{(e)}}
\end{equation}

Note that the process $\phi$ depends on the variables $\{\bm{W}^{(m)}\}_{m\leq r}$ through the cavity field functional $h_r$.

This notation can be applied to the Parisi solution of the SK model \cite{VPM} by using the ASS construction. The cavity field functional in the Parisi solution is linear, so the free energy process is Markovian and the functional $\phi_n$ depends on the random variables $\{\bm{W}^{(m)}\}_{m\leq n}$ only through the linear combination
\begin{equation}
h_n=\sum^n_{m=1}\sqrt{q_{m}-q_{m-1}} \,\bm{W}^{(m)}\,,
\end{equation}
where $q_1,q_2,\cdots, q_r$ are the overlaps \cite{VPM,ASS}. 

In the case of Markovianity, the functional $\phi_n$ is actually a function of one variable, for all the levels $0\leq n\leq r$ and for any number of RSB steps $r\in\mathbb{N}$. The expectation $\mathbb{E}_{\mathbb{W}_r^{\otimes (2/c)}}[\cdot|\{\bm{W}^{(m)}\}_{m\leq n}]$ in \eqref{start1} can be replaced by the expectation over $h_{n+1}$, conditionally to $h_n$. As a consequence, for each level $0\leq n\leq r$, the expectation value of $\phi_{n+1}$ can be evaluated by the Kolmogorov backward equation, which is a deterministic (non-random) partial differential equation (PDE) \cite{YoRev}. 

The function $\phi_n$, thus, can be represented as the solution at time $q_{n}$ of a proper PDE, starting from $\phi_{n+1}$ at time $q_{n+1}$. The juxtaposition of such PDEs gives a "continuous version" of the iteration \eqref{start1}, that is the Parisi antiparabolic PDE in the $r\to \infty$ limit\cite{VPM}.

A similar construction can be generalized to a wider class of cavity field functional $h_r:\Omega_r\to \mathbb{R}$, provided that, at least in the $r\to\infty$ limit, the process $H:=\{H_n(\{W^{(m)}\}_{m\leq n})\}_{l\leq r}$, defined as
\begin{equation} 
H_r(\{W^{(m)}\}_{m\leq n})=h_r(\{W^{(m)}\}_{m\leq n})
\end{equation}
and
\begin{equation} 
H_n(\{W^{(m)}\}_{m\leq n})=\mathbb{E}_{\mathbb{W}_r^{\otimes (2/c)}}[h_r(\{W^{(m)}\}_{m\leq n})\,|\{W^{(m)}\}_{m\leq n}]\,,
\end{equation}
is a Markov martingale. In this case, a Parisi-like equation can be achieved from the master equation of the process $H$ \cite{YoRev}.

In a more general case, the free energy process is not Markovian, i.e., for each level $n$, the functional $\phi_n$ depends on the specific values of all the past variables of the list $\{\bm{W}^{(m)}\}_{m\leq n}$. Non-Markovianity is the basic difference with the replica symmetry breaking scheme in the fully connected model \cite{VPM,ParisiMarginal}. 

Because non-Markovianity, we cannot get rid of the randomness represented by the past variables and the free energy process cannot be evaluated by a deterministic PDE. Moreover, in the full$-$RSB limit, i.e. in the $r\to \infty$ limit, the free energy process is a functional, depending on an infinite number of variables.

 One may consider a functional extension of the Parisi PDE, using the recent results about the functional It\^o calculus and functional Kolmogorov equations \cite{Dupire,ContFur}. However, a continuous version of the iteration \eqref{start1} has a complicated dependence on the cavity field functional, so computing the first variation of the functionals $\phi^{\text{(e/v)}}_{0}$ with respect the cavity field functional is a quite tricky task with this approach.

In the next section, we introduce an auxiliary variational approach that allows evaluating the map $\phi_r\mapsto \phi_0$, without any iteration procedure as in \eqref{start1}.
 
In the following, the dependence of the process $\phi^{\text{(e)}}$ and $\phi^{\text{(v)}}$ on the random couplings will be omitted for convenience; in all the equations below one has to consider a particular realization of the random couplings. 
\subsection{Equivalence with the cavity method}
\label{appendix}
In this subsection, we consider the case of the 1$-$RSB solution. We show that the variational problem on the functional \eqref{cavityFE} with respect the cavity field functional $h$ defined in\eqref{AustinRep}, together with the ansatzes described in the previous subsections, is equivalent to the $1-$RSB cavity method described in the subsection \ref{1RSBcav}

Let us define the functions
\begin{equation}
F^{\text{(v)}}(h_1,\cdots,h_c\,)=\log\Delta^{\text{(v)}}(J_{0,1},\cdots,J_{0,c},h_1,\cdots,h_c\,),\\
\end{equation}
and
\begin{equation}
F^{\text{(e)}}(h_1,h_2\,)=\log\,\Delta^{\text{(e)}}(J_{1,2},h_{1},h_{2})\,.
\end{equation}
For convenience, we omit the dependence on the coupling.

Putting $r=1$, the cavity field functional \eqref{AustinRep} is a measurable function of two independent normal random variables $W^{(0)}$ and $W^{(1)}$:
\begin{equation}
\{W^{(0)},W^{(1)}\}\mapsto h\big(\,W^{(0)},W^{(1)}\,\big).
\end{equation}
The edge and vertex contributions to the free energy are given by a single iteration of the iterative rule \eqref{start1}:
\begin{multline}
\label{start1RSB}
\phi^{\text{(e/v)}}_{0}\big(\bm{J},\,\bm{W}^{(0)}\,\big)=\frac{1}{x_{1}}\log \mathbb{E}_{\mathbb{W}_1^{\otimes (2/c)}}\left[\,\exp\left(\,x_{1} \phi^{\text{(e/v)}}_{1}\big(\bm{J},\bm{W}^{(0)},\,\bm{W}^{(1)}\,\big)\,\right)\,\big|\bm{W}^{(0)}\right]=\\
\frac{1}{x_{1}}\log\left(\int \left(\prod^{(2/c)}_{i=1} d\nu\big(W^{(1)}_{i}\big)\,\right) e^{\,x_{1}\, F^{\text{(e/v)}}\left(\bm{J},\bm{h}\big(\,\bm{W}^{(0)},\,\bm{W}^{(1)}\,\big)\,\,\right)\,}\,\,\right)\,,
\end{multline}
where $\nu$ is the normal distribution.

Now, let us define the probability density distribution of the cavity field $h$, conditionally to a fixed value for the random variable $W^{(0)}$:
\begin{equation}
\label{Mypi}
\widehat{\pi}\big(y\big|W^{(0)}\,\big)=\mathbb{E}_{\mathbb{W}_1^{\otimes (2/c)}}\left[\delta\big(\,y-h\big(\,W^{(0)},W^{(1)}\,\big)\,\big)\big|W^{(0)}\right]= \int \text{d}\nu\big(W^{(1)}\big)\delta\big(\,h\big(\,W^{(0)},W^{(1)}\,\big)-y\,\big)\,,\quad y\in \mathbb{R}
\end{equation}
For each value of $y\in \mathbb{R}$ fixed, the quantity $\widehat{\pi}(y|\Cdot)$ is a positive random variable, since it depends on $W^{(0)}$. Then, we can define the probability density distribution of the random probability density distribtion $\widehat{\pi}:=\big\{\widehat{\pi}\big(y\big|\Cdot\big);\,y\in \mathbb{R}\,\big\}$ in such a way:
\begin{equation}
\label{MyPi}
\mathbb{P}[\,\pi\,]=\mathbb{E}_{\mathbb{W}_0^{\otimes (2/c)}}\left[\delta\big[\,\pi\big(\Cdot\big)-\widehat{\pi}\big(\Cdot\big|W^{(0)}\,\big)\,\big]\,\right]= \int \text{d}\nu\big(W^{(0)}\big)\delta\big[\,\pi\big(\Cdot\big)-\widehat{\pi}\big(\Cdot\big|W^{(0)}\,\big)\,\big]\,,
\end{equation}
where $\delta[\Cdot]$ is the functional Dirac delta.
Substituting the equation \eqref{Mypi} and \eqref{MyPi} in\eqref{start1RSB}, one get
\begin{multline}
\label{ParMez}
x_1 \Phi=x_1 \Phi\big[\mathbb{P},\pi,x_1\big]\\=\overline{\int \left(\prod^c_{i=1} d[\pi_i]\mathbb{P}[\,\pi_i\,]\right)\,\,\log\left(\int\left(\prod^{c}_{i=1} dy \,\pi\big(y\big)\,\right) e^{\,x_{1}\, F^{\text{(v)}}\left(\,y_1,\cdots\,y_c\,\big)\,\right)\,}\,\,\right)}\\-\frac{c}{2}\overline{ \int \text{d}[\pi_1]\mathbb{P}[\,\pi_1\,]\,\,d[\pi_2]\mathbb{P}[\,\pi_2\,] \log\left(\int_{\mathbb{R}^{(2/c)}}\,dy_1\,\pi_i(y_1)\,dy_2\,\pi_i(y_2)\,e^{\,x_{1}\, F^{\text{(e)}}\left(y_1,\,y_2\,\right)\,}\,\,\right)}\,.
\end{multline}
Then we recover the 1$-$RSB variational free enrgy \eqref{1RSB_varFun}. The functional $\mathbb{P}$ is the 1$-$RSB cavity method order parameter and $x_1$ is the Parisi 1$-$RSB parameter.

\subsection{Iterated Gibbs principle}
In this subsection, we get a variational representation of the recursive law \eqref{start}. The variational representation turns out to be a powerful tool to get the $r\to \infty$ limit. 

In this subsection, and in the rest of the thesis, the martingale formalism \cite{YoRev} is deeply used.
Let first introduce some notation.

Let $ D^{\text{(e/v)}}_r$ be the set of the stochastic processes adapted to the filtration $\{(\mathcal{B}^{W}_n)^{\otimes (2/c)}\}_{n\leq r}$ (both for edge and vertex contribution).

Let $\mathfrak{M}^{\text{(e/v)}}_{r}\subset D^{\text{(e/v)}}_r$ the subspace of $\{(\mathcal{B}^{W}_n)^{\otimes (2/c)}\}_{n\leq r}-$adapted martingales and $\mathfrak{M}^{\text{(e/v)}}_{r,1,>}\subset \mathfrak{M}^{\text{(e/v)}}_r$ the subset of strictly positive $\{(\mathcal{B}^{W}_n)^{\otimes (2/c)}\}_{n\leq r}-$adapted martingales with average equal to $1$.

Furthermore, for each level $0<n\leq r$, let $\mathfrak{R}^{\text{(e/v)}}_{n,1,>}$ be the set of strictly positive random variables, depending on the random variables $\big\{\bm{W}^{(m)}\big\}_{m\leq n}$, with expectation value over $\bm{W}^{(n)}$, conditionally to a fixed realization of the variables $\big\{\bm{W}^{(m)}\big\}_{m\leq n-1}$, equal to $1$.

The second member of the recursion formula \eqref{start} has the form of the usual Helmotz free energy in the canonical ensemble, with $x_{n+1}$ as inverse temperature and $-\phi_{n+1}$ as Hamiltonian. As a consequence, it can be represented via the Gibbs variational principle.

For each level $n\leq r$ and each fixed realization of the first $n$ random variables $\big\{\bm{W}^{(m)}\big\}_{m\leq n}$, let us generate a strictly positive random variable $\rho_{n+1}^{\text{(e/v)}}(\Cdot|\,\{\bm{W}^{(m)}\}_{m\leq n}\,\big)$, depending on the random variable $\bm{W}^{(n+1)}$, satisfying the following normalization condition:
\begin{equation}
\label{normalRho}
\mathbb{E}_{\mathbb{W}_r^{\otimes (2/c)}}\left[\rho_{n+1}^{\text{(e/v)}}\big(\bm{W}^{(n+1)}\big|\,\{\bm{W}^{(m)}\}_{m\leq n}\,\big) \,\big|\{\bm{W}^{(m)}\}_{m\leq n}\right]=1\,\, \longrightarrow\,\, \rho_{n+1}^{\text{(e/v)}}\in \mathfrak{R}^{\text{(e/v)}}_{n,1,>}\,.
\end{equation}
The variable $\rho_{n+1}^{\text{(e/v)}}$, actually, plays the role of an effective conditional probability density function for the variable $\bm{W}^{(n+1)}$, given the realization of the past variables $\big\{\bm{W}^{(m)}\big\}_{m\leq n}$. The Gibbs principle provides a variational criterion on the space of the density functions
\begin{multline}
\label{iterVariatinal}
\phi^{\text{(e/v)}}_{n}\big(\{\bm{W}^{(m)}\}_{m\leq n}\,\big)=\\\max_{\rho_{n+1}^{\text{(e/v)}}\in\mathfrak{R}^{\text{(e/v)}}_{n,1,>}}\bigg\{
\mathbb{E}_{\mathbb{W}_r^{\otimes (2/c)}}\left[\rho_{n+1}^{\text{(e/v)}}\big(\bm{W}^{(n+1)}|\,\{\bm{W}^{(m)}\}_{m\leq n}\,\big) \phi_{n+1} ^{\text{(e/v)}}\big(\{\bm{W}^{(m)}\}_{m\leq n+1}\,\big)\,\Big|\{\bm{W}^{(m)}\}_{m\leq n}\right]\\
-\frac{1}{x_{n+1}}\mathbb{E}_{\mathbb{W}_r^{\otimes (2/c)}}\left[\rho_{n+1}^{\text{(e/v)}}\big(\bm{W}^{(n+1)}|\,\{\bm{W}^{(m)}\}_{m\leq n}\,\big)\log\,\rho_{n+1}^{\text{(e/v)}}\big(\bm{W}^{(n+1)}|\,\{\bm{W}^{(m)}\}_{m\leq n}\,\big)\,\Big|\{\bm{W}^{(m)}\}_{m\leq n}\right] \bigg\}\,,
\end{multline}
where the maximum is attained by:
\begin{multline}
\label{extremal}
\rho_{n+1}^{(e/v)\star}\big(\bm{W}^{(n+1)}|\,\{\bm{W}^{(m)}\}_{m\leq n}\,\big)=\frac{1}{Z^{\text{(e/v)}}_{n}\big(\{\bm{W}^{(m)}\}_{m\leq n}\,\big)}\exp\left(x_{n+1}\phi_{n+1} ^{\text{(e/v)}}\big(\{\bm{W}^{(m)}\}_{m\leq n+1}\,\big)\,\right),\\
\text{with}\quad Z_{n}\big(\{\bm{W}^{(m)}\}_{m\leq n}\,\big)=\mathbb{E}_{\mathbb{W}_r^{\otimes (2/c)}}\left[\exp\left(x_{n+1}\phi_{n+1} ^{\text{(e/v)}}\big(\{\bm{W}^{(m)}\}_{m\leq n+1}\,\big)\,\right)\big|\{\bm{W}^{(m)}\}_{m\leq n}\right]\,.
\end{multline}

The non-linear maps $\phi_{n} ^{\text{(e/v)}}\,\mapsto \phi_{n-1} ^{\text{(e/v)}}$ in \eqref{start} are now represented as linear maps \eqref{extremal}, with the help of suitable variational parameters. 

In physics literature, the first part in the representation \eqref{iterVariatinal}, depending on $\phi_{n+1}^{e/v}$, is called "energy", whilst the second one, with the logarithm, is the "entropic" part.

Such kind of manipulation is also at the basis of the Boué-Dupuis representation formula for the expectation value of exponential Brownian functionals\cite{BoueDepuis}.

From the effective conditional density functions $\{\,\rho_{1}^{\text{(e/v)}}, \rho_{2}^{\text{(e/v)}},\cdots,\rho_{r}^{\text{(e/v)}}\}\in \mathfrak{R}^{\text{(e/v)}}_{0,1,>}\times \mathfrak{R}^{\text{(e/v)}}_{1,1,>}\times \cdots \times \mathfrak{R}^{\text{(e/v)}}_{r,1,>}$, we can compute an effective probability density function $R_r^{\text{(e/v)}}$ for the whole collection of random varables $\{\bm{W}^{(n)}\}_{n\leq r}$:
\begin{equation}
R_r^{\text{(e/v)}}\big(\{\bm{W}^{(m)}\}_{m\leq r}\big)=\prod_{m=1}^r \rho_m^{\text{(e/v)}} \big(\bm{W}^{(m)}|\{\bm{W}^{(l)}\}_{l\leq m-1}\big)\,.\\
\end{equation}
By condition \eqref{normalRho}, the function $R_r^{\text{(e/v)}}$ is normalized:
\begin{equation}
\label{RtoRho}
\mathbb{E}_{\mathbb{W}_r^{\otimes (2/c)}}\left[R_r^{\text{(e/v)}}\big(\{\bm{W}^{(m)}\}_{m\leq r}\big)\right]=1.\\
\end{equation}

For any level $0\leq n \leq r$, the effective marginal probability density $R_n^{\text{(e/v)}}$ over the first $n$ random variables $\{\bm{W}^{(m)}\}_{m\leq n}$ is given by averaging $R_r^{\text{(e/v)}}$ over the last $r-n$ random variables $\{\bm{W}^{(m)}\}_{n\leq m\leq r}$
\begin{multline}
\label{marginalR}
R_n^{\text{(e/v)}}\big(\{\bm{W}^{(m)}\}_{m\leq n}\big)=\mathbb{E}_{\mathbb{W}_r^{\otimes (2/c)}}\left[R_r^{\text{(e/v)}}\big(\{\bm{W}^{(m)}\}_{m\leq r}\big) \big|\{\bm{W}^{(m)}\}_{m\leq n}\,\right]\\=\prod_{m=1}^n \rho_m^{\text{(e/v)}} \big(\bm{W}^{(m)}|\{\bm{W}^{(l)}\}_{l\leq m-1}\big)
\end{multline}
and
\begin{equation}
\label{R0}
R_0^{\text{(e/v)}}\big(\bm{W}^{(0)}\big)=1
\end{equation}

All the marginal density functions, defined by \eqref{marginalR} and \eqref{R0}, are already normalized, by construction.

The "entropic" part in \eqref{iterVariatinal} can be rewitten as a functional of the probability density function $R_r^{\text{(e/v)}}$ and the marginals, in such a way:
\begin{equation}
\begin{aligned}
&\mathbb{E}_{\mathbb{W}_r^{\otimes (2/c)}}\left[\rho_{n+1}^{\text{(e/v)}}\big(\{\bm{W}^{(m)}\}_{m\leq n+1}\,\big)\log\,\rho_{n+1}^{\text{(e/v)}}\big(\{\bm{W}^{(m)}\}_{m\leq n+1}\,\big)\,\big|\{\bm{W}^{(m)}\}_{m\leq n}\right]=\\
& \! \begin{multlined}
\mathbb{E}_{\mathbb{W}_r^{\otimes (2/c)}}\left[\,\frac{R_{n+1}^{\text{(e/v)}}\big(\{\bm{W}^{(m)}\}_{m\leq n+1}\big)}{R_n^{\text{(e/v)}}\big(\{\bm{W}^{(m)}\}_{m\leq n}\big)}\,\log\left(\frac{R_{n+1}^{\text{(e/v)}}\big(\{\bm{W}^{(m)}\}_{m\leq n+1}\big)}{R_{n}^{\text{(e/v)}}\big(\{\bm{W}^{(m)}\}_{m\leq n}\big)}\,\right)\,\,\Bigg| \{\bm{W}^{(m)}\}_{m\leq n}\right]=
\end{multlined}\\
& \! \begin{multlined}
\frac{1}{R_n^{\text{(e/v)}}\big(\{\bm{W}^{(m)}\}_{m\leq n}\big)}\mathbb{E}_{\mathbb{W}_r^{\otimes (2/c)}}\bigg[\,R_{r}^{\text{(e/v)}}\big(\{\bm{W}^{(m)}\}_{m\leq r}\big)\log\left(\frac{R_{n+1}^{\text{(e/v)}}\big(\{\bm{W}^{(m)}\}_{m\leq n+1}\big)}{R_{n}^{\text{(e/v)}}\big(\{\bm{W}^{(m)}\}_{m\leq n}\big)}\,\right)\,\Bigg| \{\bm{W}^{(m)}\}_{m\leq n}\bigg]
\end{multlined}
\end{aligned}
\end{equation}

The monotonicity of the expectation value on $ D^{\text{(e/v)}}_r$ and the tower property \cite{YoRev}\cite{Billingsley}, the iteration of the representations \eqref{iterVariatinal} leads to a unique variational representation for the entire map $\phi_{r} ^{\text{(e/v)}}\,\mapsto \phi_{0} ^{\text{(e/v)}}$:

\begin{multline}
\label{GibbsFunctional}
\phi^{\text{(e/v)}}_{0}(\bm{W}^{(0)}\,)=
\max_{R^{\text{(e/v)}}\in \mathfrak{M}^{\text{(e/v)}}_{r,1,>}}
\Bigg\{\,\mathbb{E}_{\mathbb{W}_r^{\otimes (2/c)}}\left[R_r^{\text{(e/v)}}\big(\{\bm{W}^{(m)}\}_{m\leq r}\big) \phi_{r} ^{\text{(e/v)}}\big(\{\bm{W}^{(m)}\}_{m\leq r}\big)\,\big| \bm{W}^{(0)}\, \right]\\-\mathbb{E}_{\mathbb{W}_r^{\otimes (2/c)}}\Bigg[\,R_r^{\text{(e/v)}}\big(\{\bm{W}^{(m)}\}_{m\leq r}\big)\sum_{n=1}^{r}\frac{1}{x_n}\Delta_n^{\text{(e/v)}}\big(\{\bm{W}^{(m)}\}_{m\leq n}\big)\,\bigg| \,\bm{W}^{(0)}\,\Bigg]\,\Bigg\}\,,
\end{multline}
where
\begin{equation}
\Delta_n^{\text{(e/v)}}\big(\{\bm{W}^{(m)}\}_{m\leq n}\big)=
\,\log\Big(R_n^{\text{(e/v)}}\big(\{\bm{W}^{(m)}\}_{m\leq n}\big)\,\Big)-\log\Big(R_{n-1}^{\text{(e/v)}}\big(\{\bm{W}^{(m)}\}_{m\leq n-1}\big)\,\Big)\,.
\end{equation}
The maximum is formally attained substituting the solutions \eqref{extremal} in \eqref{RtoRho}, for each level $n$.

Note that the collection of effective marginal probability densities, defined in \eqref{marginalR}, defines an average $1$ strictly positive martingale:
\begin{equation}
R^{\text{(e/v)}}:=\big\{R_n^{\text{(e/v)}}\big(\{\bm{W}^{(m)}\}_{m\leq n}\big)\big\}_{n\leq r}\in \mathfrak{M}^{\text{(e/v)}}_{r,1,>}\,. 
\end{equation}
The martingale property, indeed, is stated by the definition of marginal density functions in the equation \eqref{marginalR}.

The Gibbs variational principle, then, combines a cumbersome recursive composition of conditional non-linear expectations values in a single variational problem over the space $\mathfrak{M}^{\text{(e/v)}}_{r,1,>}$ of positive, average $1$, martingales on the filterd probability space $(\Omega_r,\mathcal{B}_r,\{(\mathcal{B}^{W}_n)\}_{n\leq r},\mathbb{W}_r)^{\otimes (2/c)}$. This is a big deal of such approach, since martingales are well-defined mathematical object in any generic probability space. 

The total free energy functional is given by the sum of the edge and vertex contributions, averaged over the root random variables $\bm{W}^{(0)}$ and the random couplings:
\begin{equation}
\label{GibbsFunctional2}
\Phi= \int_{\mathbb{R}^{c}}\left(\prod^{c}_{i=1} d\nu\big(W^{(0)}_{i}\big)\,\right) \, \phi^{\text{(v)}}_0\big(\bm{J},\bm{W}^{(0)}\big)-\frac{c}{2}\int_{\mathbb{R}^{2}}\left(\prod^{2}_{i=1} d\nu\big(W^{(0)}_{i}\big)\,\right) \, \phi^{\text{(e)}}_0\big(\bm{J},\bm{W}^{(0)}\big)\,.
\end{equation}

The equilibrium free-energy is given by the extremization of the total free energy functional \eqref{GibbsFunctional2} with respect to the physical order parameter, i.e. the cavity field functional $h_r\in \mathfrak{F}(\Omega_r,\mathbb{R})$; the representations of the two free energy contributions, given by\eqref{GibbsFunctional}, constitute two independent variational problems inside a larger variational problem.

In the following, we will refer to \eqref{GibbsFunctional} as edge and vertex auxiliary variational problems, whilst the extremization over the cavity field functional is the physical variational problem. The auxiliary variational problems must be solved before, keeping the variational parameters of the physical variational problem fixed.

The careful reader may argue that the Gibbs variational principle seems not to provides a real simplification since the computation of the non-linear expectations still remains in the solutions \eqref{extremal}. This is actually true in the discrete replica symmetry breaking case.

In the "continuous limit" of replica symmetry breaking, however, the martingale $R^{\text{(e/v)}}$ has a nice representation and we do not need to deal with the solution \eqref{extremal}, as explained in the next section.

\section{The full replica symmetry breaking theory}
\label{Cont_extension}
In the previous section, the $r$$-$RSB free energy \eqref{cavityFE} is obtained by an auxiliary variational representation over the space of positive discrete time martingales.

In this section, the auxiliary variational problem \eqref{GibbsFunctional} is extended to continuous-time martingales that have also a continuous sample path \cite{YoRev}.

It is worth stressing that the sample path continuity is not a restrictive assumption; indeed, we will prove, in the next chapter, that the free energy functional that we present at the end of this section contains all the discrete-RSB theories as possible solutions.

From a rigorous mathematical point of view, the extension to continuous martingale is not the $r\to \infty$ limit of the $r$$-$RSB theory, but rather a generalization of the previous stochastic analysis to another class of martingales.

The first subsection provides a formal derivation of the auxiliary variational problem with continuous martingale.

In the second subsection, we consider the case of the It\^o stochastic processes and an explicit formulation of the full$-$RSB free energy functional is finally obtained.

\subsection{The generalized Chen-Auffinger variational representation}
\label{subsec4.4.2}
In this subsection, the variational auxiliary problem for continuous martingales is derived.

We rephrase the martingale approach of the previous section, with a suitable formalism, representing both continuous and discrete martingales. Then we concentrate on path continuous martingales, getting a generalization of the Auffinger-Chen variational representation of the Parisi solution of the SK models \cite{ChenAuf}. 

The extension of the auxiliary variational problem \eqref{GibbsFunctional} to the continuous martingale space relies on several steps.

\subsubsection{Continuous-time formalism.} Given any ordered collection of random variables $\{W^{(m)}\}_{m\leq r} \in \Omega_r$, together with the increasing sequence
\begin{equation}
\label{Q_sequence}
Q:=(q_0,\, q_1,\, \cdots ,\, q_{r+1})\,,
\end{equation}
with
\begin{equation}
\label{increasing}
0=q_0\leq q_1\leq \cdots \leq q_{r+1}=1\,,
\end{equation}
we can define a continuous-time, piece-wise constant, bounded random function $W_{\text{c}}:=\{W_{\text{c}}(q);\,q\in [0,1]\}$, such as for each level $0\leq n\leq r$ and time $q_{n}< q\leq q_{n+1}$, the random quantity $W_{\text{c}}(q)$ depends only on the first $n$ variables $\{W^{(m)}\}_{m\leq n}$. A possible choice may be
\begin{equation}
\label{continuation}
W_{\text{c}}(q)= W^{(0)}+\sum_{i=1}^{r} \sqrt{q_{n}-q_{n-1}}\, \,W^{(n)} \theta(q-q_{n-1}),\quad q\in [0,1],
\end{equation}
 where the function $\theta$ is the Heaviside function and the function $\mathbb{1}_{\{0\}}(q)$ has the value $1$ at $q=0$ and the value $0$ for all $q>0$.

Any $r-$ step adapted process $O:=\{O_0,O_1,\cdots,O_r\} \in D^{(e/v)}_r$ can now be considered as an ordered collection of functionals of the vector random function $\bm{W_{\text{c}}}$, constituted by $2$ or $c$ independent realizations of the random function defined in \eqref{continuation}:
\begin{equation}
O_n=O_n\big(\,\{\bm{W}^{(m)}\}_{m\leq n}\,\big)\longrightarrow O_n\big(\{\bm{W}_\text{c}(q);0\leq q\leq q_n\}\big)
\end{equation}
where the symbol $\{\bm{W}_\text{c}(q);0\leq q\leq q_n\}$ denotes that the functional $O_n$ depends on all the values $\bm{W}_\text{c}(q)$, assumed by the realization of the random function $\bm{W}_\text{c}$ at each time $q\in[0,q_n]$.

We define a continuous-time process, depending on $\bm{W}_{\text{c}}$ and based on $O$, in such a way:
\begin{equation}
\label{continuation2}
O_{\text{c}}(q)= O_0 \mathbb{1}_{[0,q_1]}(q)+\sum_{i=1}^{r} O_n \mathbb{1}_{(q_{n},q_{n+1}]}(q),\quad q\in [0,1].
\end{equation}
where the functions $\mathbb{1}_{(q_{n},q_{n+1}]}(q)$, for $1\leq n\leq r$, are equal to $1$ if $q_{n}<q\leq q_{n+1}$ and vanish elsewhere.

With this construction, the values of $O_{\text{c}}(q)$, at each time $0\leq q\leq 1$, depends on the values $\bm{W}_{\text{c}}(q')$, assumed by a given realization of the stepwise random function $\bm{W}_{\text{c}}$, at all the times $q'\in[0,q]$. So the process $O_{\text{c}}$ is adapted to the continuous-time stepwise random function $\bm{W}_{\text{c}}$. 

In the following, the subscript $\Cdot_{\text{c}}$ will be omitted and we will deal only with continuous-time stochastic processes.

In the $r\to\infty$ limit, the sequence $Q$ "fills" the $[0,1]$ segment, so this limit can be formally achieved enlarging the space of stepwise processes to a wider class of continuous-time processes.

The probability space will be carefully defined in the last paragraph of this section. For now, we consider that the process $W$ is an element of a generic sample space of functions $\Omega$, with a proper $\sigma-$algebra $\mathcal{B}$ and a given probability measure $\mathbb{W}$.
A process $O:\Omega^{\otimes(2/c)}\to \mathbb{R}$ is said to be adapted (or non-anticipating) if it is adapted with respect to the usual $\mathbb{W}^{\otimes (2/c)}$-augmentation\footnote{The usual augmentation of a continuous-time filtration, with respect to a given probability measure, is the smallest right-continuous filtration that contains the original one, enlarged with the set with probability $0$ with respect to the given probability measure. Usually, a rigorous treatment of continuous-time stochastic processes requires the usual completion.} of the natural filtration of the vector process $\bm{W}$, defined on product probability space $(\Omega,\mathcal{B},\mathbb{W})^{\otimes(2/c)}$ (\cite{YoRev}). The augmented natural filtration of the process $\bm{W}$ is denoted by $\{\mathcal{F}_q^{\bm{W}}\}_{q\in [0,1]}$.

Let us remind some standard crucial mathematical tools, defined for continuous-time stochastic processes, that will be useful in the next paragraphs \cite{YoRev}. 

We call \emph{quadratic variation} of a process $O$ the process $[O]$ defined as:
\begin{equation}
\label{var}
[O](q)=\lim_{\substack{M\to \infty\\ \delta \to 0}} \sum^M_{n=0} \left(O(q_{n+1})-O(q_{n})\right)^2,\,\,\text{with} \,\,0=q_0\leq q_1\leq \cdots\leq q_M=q\,,
\end{equation}
where $\delta$ is the mesh of the partition. 

For stepwise processes, as \eqref{continuation2}, this quantity is reduced to the sum over the discontinuity jumps in such a way:
\begin{equation}
\label{varfinit}
[O](q)= \sum^r_{n=1} \mathbb{1}_{[q_{n},q_{n+1})}(x)\left(O_{n}-O_{n-1}\right)^2,\,\,\text{with}\, 0=q_0\leq q_1\leq \cdots\leq q_{r+1}=1\,.
\end{equation}

In the same manner, the \emph{covariation} between two processes $O$ and $P$ is the process $\Braket{O,P}$ such that:
\begin{multline}
\label{covar}
\Braket{O,P}(q)=\lim_{\substack{M\to \infty\\ \delta \to 0}} \sum^M_{n=0} \left(O(q_{n+1})-O(q_{n})\right)\left(P(q_{n+1})-P(q_{n})\right),\\\,\text{with}\, 0=q_0\leq q_1\leq \cdots\leq q_M=q\,.
\end{multline}
Obviously, the quadratic variation and the covariation vanish for smooth functions, but they are not trivial for all continuous martingale.

We can also define the integration with respect to a continuous-time stochastic process $O$ by the Lebesgue-Stieltjes integral \cite{Billingsley} or by the It\^o integral \cite{Ito} whether the integration process $O$ is of bounded variation or of bounded quadratic variation respectively\cite{YoRev}.

Note that, by definition, for any bounded function $f:[0,1]\to \mathbb{R}$ and stepwise process $O$, we have:
\begin{equation}
\int^1_0 f(q) dO(q)=\sum^r_{n=1} f(q_{n-1})(\,O_{n}-O_{n-1}\,).
\end{equation}
\subsubsection{The discrete-RSB auxiliary problem with continuous-time martingales }
By this formalism, the auxiliary variational representation \eqref{GibbsFunctional} can be easily reformulated in term of generic continuous-time processes as:
\begin{multline}
\label{GibbsFunctional3}
\phi^{\text{(e/v)}}(0,\bm{W}{(0)}\,)=
\max_{R^{\text{(e/v)}}\in \mathfrak{M}^{\text{(e/v)}}_{[0,1],1,>}}
\Bigg\{\mathbb{E}_{\mathbb{W}^{\otimes(2/c)}}\left[ R^{\text{(e/v)}}\left(\,1,\,\bm{W}\,\right)\phi ^{\text{(e/v)}}\left(\,1,\,\bm{W}\,\right)\,\bigg| \,\bm{W}{(0)} \right]\\-\mathbb{E}_{\mathbb{W}^{\otimes (2/c) }}\left[R^{\text{(e/v)}}\left(\,1,\,\bm{W}\,\right)\int^1_0 \frac{1}{x(q)}\,d\left(\,\log\, R^{\text{(e/v)}}\left(\,q,\,\bm{W}\,\right)\,\,\right)\,\,\bigg| \,\bm{W}{(0)}\,\right]\,\Bigg\}\,,
\end{multline}
where we have used the short-hand notation
\begin{equation}
\label{short}
R^{\text{(e/v)}}\left(\,q,\,\bm{W}\,\right)=R^{\text{(e/v)}}\left(\,x,\,\{\bm{W}(q);0\leq q'\leq q\}\,\right)\,,\quad \forall 0 \leq q \leq 1\,,
\end{equation}
and
\begin{equation}
\phi ^{\text{(e/v)}}\left(\,1,\,\bm{W}\,\right)=\log \,\Delta^{\text{(e/v)}} \left(\,\bm{h}\left(\,\{\bm{W}(q);0\leq q\leq 1\}\,\right)\,\right)\,.
\end{equation}
and $x(q)$ is the deterministic function obtained by the sequence \eqref{X_sequence} and \eqref{Q_sequence} in such a way
\begin{equation}
x(q)=\sum^r_{i=0}x_{n+1}\mathbb{1}_{(q_n,q_{n+1}]}(q)\,.
\end{equation}
The function $x:[0,1]\to[0,1]$ is the so-called Parisi order parameter (POP), that we defined in subsection \ref{POSK} for fully connected systems.

In fully connected system the "time" parameter $q$ is the overlap of the magnetizations of two pure states (see chapter \ref{C2}). In the next chapters, we will that $q$ is related to the correlation between two different states, but it is not actually an overlap. Moreover, it is clear that the choice of the sequence $Q$ is arbitrary, provided the increasing condition \eqref{increasing}.

The range set $\mathfrak{M}^{\text{(e/v)}}_{[0,1],1,>}$ of the variational parameter $R^{\text{(e/v)}}$ is the space of positive martingales on the probability space $(\Omega,\mathcal{B},\mathbb{W})^{\otimes(2/c)}$, with average $1$.

The representation \eqref{GibbsFunctional3} holds both for stepwise and continuous martingales, providing a general formulation for all levels of discrete replica symmetry breaking and for the full replica symmetry breaking.

Note that the auxiliary variational representation turns out to be a powerful tool in getting the continuous limit of replica symmetry breaking free energy functional. The extension to the continuous case of the representation \eqref{GibbsFunctional2} has been easily defined in \eqref{GibbsFunctional3}, simply by introducing a proper notation, allowing us to deal with quantities that are well defined even in the $r\to \infty$ limit. By contrast, the $r\to \infty$ limit of the iterative low \eqref{start1} appears to be quite cumbersome, since the cavity field functional \eqref{AustinRep} depends on the random quantities $\{W^{(m)}\}_{m\leq r}$ in a non trivial way.

The set $\mathfrak{M}^{\text{(e/v)}}_{[0,1],1,>}$ is too generic for many practical computation, so in the next paragraphs further assumptions on the martingale $R^{\text{(e/v)}}$ will be imposed.

\subsubsection{Path continuity assumption}
\label{pathcontinuity}
A generic martingale $R^{\text{(e/v)}}\in\mathfrak{M}^{\text{(e/v)}}_{[0,1],1,>}$ can be decomposed as the sum of a purely continuous martingale and a purely discontinuous martingale.

In the $r-$step RSB described in the previous section, the continuous part vanishes, and the martingale $R^{\text{(e/v)}}$ is purely discontinuous. Here, we are interested in the case where $R^{\text{(e/v)}}$ is purely continuous. We will see, in the next chapter, that also the discrete-RSB can be described by path continuous martingales.

In this paragraph, we use the shorthand notation $(q,\bm{W})$ to indicate the non-anticipating dependence of the processes to the random function $\bm{W}$, as in \eqref{short}.

Any strictly positive and path continuous martingale, can be represented as the exponential function of a path continuous supermartingale \cite{YoRev}. More precisely, using the It\^o lemma of the stochastic differential calculus \cite{Ito}, the martingale condition and the normalization of $R^{\text{(e/v)}}$ lead to the following representation ( Proposition 1.6 in Chapter VIII of \cite{YoRev}):
\begin{equation}
\label{exp}
R^{\text{(e/v)}}\left(\,q,\bm{W}\,\right)=\mathcal{E}\left(\,L^{\text{(e/v)}};q,\bm{W}\,\right)=e^{L^{\text{(e/v)}}(q,\bm{W})-\frac{1}{2}[L^{\text{(e/v)}}](q,\bm{W})}\,,
\end{equation}
where $L^{\text{(e/v)}}(q,\bm{W})$ is a sample continuous martingale with
\begin{equation}
\label{0Cond}
\mathbb{E}_{\mathbb{W}^{\otimes (2/c) }}\left[L^{\text{(e/v)}}(q,\bm{W})\big| \,\bm{W}^{(0)}\right]=0\,
\end{equation}
and $\mathcal{E}(L^{\text{(e/v)}})$ is the Doléans-Dade exponential (DDE) \cite{YoRev} of the martingale $L^{\text{(e/v)}}$.

It is worth noting that, for a generic martingale that verifies \eqref{0Cond}, the associated DDE, defined as in \eqref{exp} is a\emph{local martingale} (Definition (1.5) in Chapter IV of \cite{YoRev}). However, In order to verify the normalization condition is verified if and only if the DDE is a \emph{true martingale}. This condition provides some further requirements on the martingale $L^{\text{(e/v)}}$ \cite{YoRev,Hitsuda,Lipster}.

Using the stochastic integral representation of the exponential \eqref{exp}
\begin{equation}
\label{stocInt}
R^{\text{(e/v)}}(1,\bm{W})=1+\int^1_0R^{\text{(e/v)}}(1,\bm{W}) \,dL^{\text{(e/v)}}(q,\bm{W})
\end{equation}
and substituting the equations \eqref{0Cond}, \eqref{exp} and \eqref{stocInt} in the entopic part of \eqref{GibbsFunctional3}, one gets
\begin{multline}
\mathbb{E}_{\mathbb{W}^{\otimes (2/c) }}\left[R^{\text{(e/v)}}(1,\bm{W})\int^1_0 \frac{1}{x(q)}\,d\left(\,\log\, R^{\text{(e/v)}}(q,\bm{W})\,\,\right)\,\,\bigg| \,\bm{W}^{(0)}\,\right]\\=
\mathbb{E}_{\mathbb{W}^{\otimes (2/c) }}\left[\,\int^1_0 \frac{1}{x(q)}\,\,dL^{\text{(e/v)}}(q,\bm{W})\,\,\,\,\bigg| \,\bm{W}^{(0)}\,\right]\\+\mathbb{E}_{\mathbb{W}^{\otimes (2/c) }}\left[\,\int^1_0 R^{\text{(e/v)}}(q,\bm{W})dL^{\text{(e/v)}}(q,\bm{W})\,\int^1_0 \frac{1}{x(q)}\,\,dL^{\text{(e/v)}}(q,\bm{W})\,\,\,\,\bigg| \,\bm{W}^{(0)}\,\right]\\
-\frac{1}{2}\mathbb{E}_{\mathbb{W}^{\otimes (2/c) }}\left[\int^{1}_0 \frac{1}{x(q)}\,\mathbb{E}_{\mathbb{W}^{\otimes (2/c) }}\left[R^{\text{(e/v)}}(1,\bm{W})\,\big| \,\mathcal{F}^{\bm{W}}_{q}\,\right]d[L]^{\text{(e/v)}}(q,\bm{W})\bigg| \,\bm{W}^{(0)}\,\right]=\\
\frac{1}{2}\mathbb{E}_{\mathbb{W}^{\otimes (2/c) }}\left[\int^{1}_0 \frac{1}{x(q)}\,R^{\text{(e/v)}}(q,\bm{W})d[L]^{\text{(e/v)}}(q,\bm{W})\bigg| \,\bm{W}^{(0)}\,\right]\,
\end{multline}
and the functional \eqref{GibbsFunctional3} becomes:
\begin{multline}
\label{GibbsFunctional4}
\phi^{\text{(e/v)}}(0,\bm{W}{(0)}\,)=
\max_{L^{\text{(e/v)}}\in \mathfrak{M}^{\text{(e/v)}}_{[0,1],C}}
\Bigg\{\mathbb{E}_{\bm{h},\mathbb{W}^{\otimes (2/c)}}\left[e^{L^{\text{(e/v)}}(1,\bm{W})-\frac{1}{2}[L^{\text{(e/v)}}](1,\bm{W})}\phi ^{\text{(e/v)}}\left(\,1,\,\bm{W}\,\right)\,\bigg| \,\bm{W}^{(0)} \right]\\-\frac{1}{2}\mathbb{E}_{\bm{h},\mathbb{W}^{\otimes(2/c)}}\left[\int^1_0\frac{1}{x(q)}\,\, e^{L^{\text{(e/v)}}(q,\bm{W})-\frac{1}{2}[L^{\text{(e/v)}}](q,\bm{W})}\,d[\,L^{\text{(e/v)}}\,](q,\bm{W})\,\,\bigg| \,\bm{W}^{(0)}\,\right]\,\Bigg\},
\end{multline}
where $\mathfrak{M}^{\text{(e/v)}}_{[0,1],C}$ is the set of adapted martingales under the probability measure $\mathbb{W}^{\otimes(2/c)}$, with continuous sample path and starting from $0$.

The total full$-$RSB free energy functional is finally given by:
\begin{equation}
\label{freeFull1}
\Phi=\overline{\int \prod^c_{i=1} d\nu\big(W_i^{(0)}\,\big)\,\,\,\, \phi^{\text{(v)}}\big(0,\bm{W}^{(0)}\,\big)}-\frac{c}{2}\overline{ \int \text{d}\nu\big(W_1^{(0)}\big)d\nu\big(W_2^{(0)}\big) \phi^{\text{(e)}}\big(\,0,\,W_1^{(0)},W_2^{(0)}\,\big)}\,,
\end{equation}
where the overline $\overline{\Cdot}$ represents the average over the random couplings.

The representation \eqref{GibbsFunctional4}, is an extension to the actual model of the Chen-Auffinger representation of the Parisi functional, defined for the $SK$ model. Note that, in the Chen-Auffinger representation, the process $L^{\text{(e/v)}}$ is a Markov process \cite{ChenAuf}.

The free energy \eqref{freeFull1} actually provides a mathematical representation of the full$-$RSB free energy, then it may represent an interesting starting point for further qualitative analysis about replica symmetry breaking on sparse graphs.

For the quantitative evaluation of the free energy, we reduce the present computation to It\^o processes. 

\subsection{Full-RSB: final formulation}
\label{fullRSBsec}
In the preceding subsection, the full$-$RSB ansatz is presented using a variational representation based on continuous martingales on a whatever probability space. This formalism allows deriving a well defined "continuous version” of the variational representation \eqref{GibbsFunctional} of the edge and vertex contributions. However, the generality of the probability space does not enable practical calculations. 

In this section, we consider that the cavity field and the auxiliary martingale $L^{\text{(e/v)}}$ are explicit functionals of a vectorial Brownian motion on $[0,1]$, with a random starting point. 

We set $\Omega:=C([0,1],\mathbb{R})$ to be the space of continuous functions $\omega:[0,1]\to \mathbb{R}$ and we endow this space by a proper $\sigma-$algebra $\mathcal{F}$. Let $\mathbb{W}_{\nu}$ be the probability measure such as the coordinate map process
\begin{equation}
W(\omega):=\{W(q,\omega)=\omega(q);\,q\in[0,1]\}\,,
\end{equation}
together with its natural filtration, is a Brownian motion, starting from a random normal distributed point $\omega(0)$ \cite{YoRev}. Such Brownian motion will be simply indicated by $\omega$. The normal distribution will be indicated by $\nu$, according to the notation described at the end of the subsection \eqref{finitersbpara}. The usual augmentation of the natural filtration of the Brownian motion is denoted by $\{\mathcal{F}_q\}_{q\in[0,1]}$.

The cavity field functional $h$ is a real valued functional $h:C([0,1],\mathbb{R})\to \mathbb{R}$, measurable with respect to the completion of the $\sigma-$algebra generated by the Brownian motion, and the martingales $L^{\text{(v)}}$ and $L^{\text{(e)}}$ are adapted to the usual augmentation of the natural filtration of the vectorial Brownian motion with $2$ or $c$ components, respectively for the edge and the vertex contribution.

The martingale representation theorem \cite{Clark}\cite{ItoRep}, assures that the matingale $L^{\text{(e/v)}}$ can be represented as It\^o integrals \cite{Ito}:
\begin{equation}
\label{RRep}
L^{\text{(e/v)}}(q,\bm{\omega}\,)=\sum^{(2/c)}_{j=1}\int^{q}_0x(q)\, r_j^{\text{(e/v)}}(q',\,\bm{\omega}\,)d\omega_j(q')=\int^{q}_0 x(q) \,\bm{r}^{\text{(e/v)}}(q', \bm{\omega}\,)\cdot \text{d}\bm{\omega}(q')\,,
\end{equation}
where $\bm{\omega}$ is the vectorial Brownian motion, with $2$ or $c$ independent components. The subscripts $i$ and $j$ indicate the single components of the vectorial processes $\bm{r}^{\text{(e/v)}}$, and $\bm{\omega}$. The components of the vector $\bm{r}^{\text{(e/v)}}$ are locally square integrable and adapted processes and the integrals are It\^o stochastic integrals \cite{YoRev}. 

By abuse of notation, we denote by $\mathcal{E}\big(x \bm{r}^{\text{(e/v)}}\big)$ the DDE associated to the martinagale $L^{\text{(e/v)}}$ :
\begin{multline}
\label{DDE_xr}
\mathcal{E}\big(x \bm{r}^{\text{(e/v)}};q,\bm{\omega}\big)\\=\exp\left(\int^{q}_0 x(q) \,\bm{r}^{\text{(e/v)}}(q', \bm{\omega}\,)\cdot \text{d}\bm{\omega}(q')-\frac{1}{2}\int^{q}_0 dq x^2(q) \,\big\|\bm{r}^{\text{(e/v)}}(q', \bm{\omega}\,)\big\|^2\right)
\end{multline}
The processes $\bm{r}^{\text{(e)}}$ and $\bm{r}^{\text{(v)}}$ are determined by the edge and vertex auxiliary variational problems, respectively. 
 
The shorthand notation $(q,\bm{\omega})$ ( or $(q,\omega)$), after a symbol indicating a stochastic process, denotes that such quantity depends on $q$ and has a non-anticipating functional dependence on the Brownian motion, i.e. it depends on the realization of the Brownian motion $\bm{\omega}(q')$ (or $\omega(q')$ ) at each time $0\leq q' \leq q$:
\begin{equation}
\label{adapted_notation}
\begin{gathered}
h(1,\omega_i)=h(q,\{\omega_i(q');0\leq q'\leq 1\}),\\
\bm{r}^{\text{(e/v)}}(q,\bm{\omega})=\bm{r}^{\text{(e/v)}}(q,\{\bm{\omega}(q');0\leq q'\leq q\}),\\
L^{\text{(e/v)}}(q,\bm{\omega})=L^{\text{(e/v)}}(q,\{\bm{\omega}(q');0\leq q'\leq q\}).
\end{gathered}
\end{equation}
Obviously, the functionals $h$ and $\bm{r}^{\text{(e/v)}}$ must be assumed to be regular enough to ensure the functional \eqref{freeFull1} to be bounded.

Let $ D([0,1]\times \Omega^{\otimes 2},\mathbb{R}^2)$ and $ D([0,1]\times \Omega^{\otimes c},\mathbb{R}^c)$, be some proper spaces of the processes $\bm{r}^{\text{(e)}}$ and $\bm{r}^{\text{(v)}}$ respectively (a definition will be provided in the next chapter).
Let us define the following random variables
\begin{gather}
\label{claimV}
\Psi^{\text{(e)}}\left(\,1,\,\omega_1,\omega_2\,\right)=\log \,\Delta^{\text{(e)}} \left(h(1,\omega_1),\,h(1,\omega_2)\right)\,,\\
\label{claimE}
\Psi ^{\text{(v)}}\left(\,1,\,\bm{\omega}\,\right)=\log \,\Delta^{\text{(v)}} \left(h(1,\omega_1),\cdots,\,h(1,\omega_c)\right).
\end{gather}
Note that at non-zero temperature ($\beta\leq \infty$), the above quantities are bounded:
\begin{equation}
\begin{gathered}
\label{boundedness0}
-\beta\leq  \Psi^{\text{(e)}}\left(\,1,\,\omega_1,\omega_2\,\right)\leq \beta\,,\\
-c \beta\leq  \Psi^{\text{(v)}}\left(\,1,\,\bm{\omega}\,\right)\leq c\beta\,,
\end{gathered}
\end{equation}
for all $\bm{\omega}\in \Omega^{\otimes\text{(2/c)}}$.

With this notation, we define 
\begin{multline}
\label{funcAux}
\Gamma^{\text{(2/c)}}\big(\,\Psi^{\text{(e/v)}},\,x,\,\bm{r}^{\text{(e/v)}};\,0,\,\bm{\omega}(0)\,\big)\\=\mathbb{E}_{(\mathbb{W}_{\nu})^{\otimes(2/c)}}\left[\mathcal{E}\big(x \bm{r}^{\text{(e/v)}};1,\bm{\omega}\big)\Psi^{\text{(e/v)}} ( \,1,\bm{\omega}\,)\Big| \,\{\bm{\omega}(0)\}\, \right]\\
-\frac{1}{2}\mathbb{E}_{(\mathbb{W}_{\nu})^{\otimes(2/c)}}\left[\int^1_0 dq\,x(q)\mathcal{E}\big(x \bm{r}^{\text{(e/v)}};q,\bm{\omega}\big)\,\big\|\,\bm{r}^{\text{(e/v)}}(\,q, \,\bm{\omega}\,) \,\big\|^2\,\,\bigg| \,\{\bm{\omega}(0)\}\,\right]\,,
\end{multline}
where $\mathbb{E}_{\mathbb{W}^{\otimes(2/c)}}\left[\cdot|\{\bm{\omega}(0)\}\right]$ is the expectation value with respect to the vectorial Brownian motion $\bm{\omega}$, conditionally to a fixed realization of the starting point $\bm{\omega}(0)$.

 The auxiliary variational representation \eqref{GibbsFunctional4} is given by

\begin{equation}
\label{GibbsTimeFunctional2}
\phi^{\text{(e/v)}}(\,0,\bm{\omega}{(0)}\,)=
\max_{\bm{r}^{\text{(e/v)}}\in D([0,1]\times \Omega^{\otimes {\text{(2/c)}}},\mathbb{R}^{\text{(2/c)}})} \Gamma^{\text{(2/c)}}\left(\,\Psi^{\text{(e/v)}},\,x,\,\bm{r}^{\text{(e/v)}};\,0,\,\bm{\omega}(0)\,\right)\,.
\end{equation}

The maximum in \eqref{GibbsTimeFunctional2} must be obtained by taking the cavity field functional $h$ fixed.

By Cameron-Martin/Girsanov (CMG) theorem \cite{CameronMartin,Girsanov}, the DDE in the expectation values can be reabsorbed in a proper change of probability measure, so the representation \eqref{GibbsTimeFunctional2} actually recovers a Chen-Auffinger-like variational representation. Change of measure, however, acts in a non-trivial way on the cavity field $h$, so, for the actual problem, we avoid the CMG transformation and just deal with the representation \eqref{GibbsTimeFunctional2}.

We may guess that the auxiliary variational problem \eqref{GibbsTimeFunctional2} can be solve by imposing a proper stationary condition:
\begin{equation}
\label{stationary0}
\frac{\delta \Gamma^{\text{(2/c)}}\left(\,\Psi^{\text{(e/v)}},\,x,\,\bm{r}^{\text{(e/v)}};\,0,\,\bm{\omega}(0)\,\right)}{\delta r_i^{\text{(e/v)}}(\,q,\bm{\omega}\,)} =0\,.
\end{equation}
In the next section we will provide a detailed mathematical analysis of the auxiliary variational problem.

The total free energy functional is given by replacing the solution of \eqref{GibbsTimeFunctional2} in \eqref{freeFull1}.
\chapter{The replica symmetry breaking expectation}
\label{C5}
\thispagestyle{empty}
In the previous chapter, we derived the full-RSB free energy functional, using a proper variational representation, the so-called auxiliary variational problem, of both the edge contribution and vertex contribution that appear in\eqref{cavityFE}. Moreover, we guessed that the auxiliary variational problem may be solved 
by imposing a proper stationary condition (equation \eqref{stationary0}) on the variational parameter.

This chapter provides a deep mathematical analysis of the auxiliary variational problem. We consider a generalization of the functional $\Gamma^{\text{(2/c)}}$, defined \eqref{GibbsTimeFunctional2}; the solution of such generalized variational problem yield a functional that we call \emph{RSB expectation}.

In the first section we give the notation that we will use throughout the rest of the thesis and provides a definition of the RSB expectation operator.

In Section \ref{sec5.2} we give an explicit expression of the stationary condition \eqref{stationary0} and we explore some property of the solution, proving that the solution of the stationary condition is the global maximum of the auxiliary variational functional.. 

Section \ref{sec5.3} is devoted to the proof of the existence of the solution of the stationary equation obtained in Section \ref{sec5.2}. We also prove that, for a particular choice of the Parisi order parameter $x:[0,1]\to [0,1]$, the solution of the full-RSB auxiliary variational problem is equivalent to the solution of the discrete-RSB recursion \eqref{start1}.

\section{Formulation of the problem}
\label{sec5.1}
In this section we define a generalization of the auxiliary variational problem \eqref{GibbsTimeFunctional2}. In particular, we aim to define a functional with the same form as \eqref{GibbsTimeFunctional2}, depending on a random variables $\Psi$ and on a function $x:[0,1]\to[0,1]$.

Throughout this and the next chapter, the symbol $\bm{\omega}$ denotes a $n-$dimensional Brownian motion ($n\in\mathbb{N}$), starting from a random point $\bm{\omega}(0)$ at $q=0$, defined on a given probability space $(\Omega,\mathcal{F},\nu\times \mathbb{W})$ (it is not the same probability space of the previous chapter). The symbol $\nu$ denotes the probability measure associated to $\bm{\omega}(0)$, while $\mathbb{W}$ is the probability measure  associated to $\bm{\omega}-\bm{\omega}(0)$. Let $\{\mathcal{F}_q\}_{q\in[0,1]}$ be the usual augmented natural filtration of $\bm{\omega}$. 

We denote by $\mathbb{E}_{\nu}[\Cdot]$ the expectation value with respect the probability $\nu\times \mathbb{W}$. We use the short-hand notation $\mathbb{E}[\Cdot]$ for the expectation value with respect the probability measure $\mathbb{W}$, conditionally to a given realization of $\bm{\omega}(0)$:
\begin{equation}
\mathbb{E}[\Cdot]=\mathbb{E}_{\nu}[\Cdot|\mathcal{F}_0]=\mathbb{E}_{\nu}[\Cdot|\{\bm{\omega}(0)\}]\,.
\end{equation}
We denote by $\int \text{d}\nu(\bm{\omega}(0))\,\Cdot$ the average with respect the starting point:
\begin{equation}
\mathbb{E}_{\nu}[\Cdot]=\int \text{d}\nu(\bm{\omega}(0))\mathbb{E}[\Cdot]\,.
\end{equation}
 Let us define define
\begin{itemize}
\item  $L_1^{\infty}(\Omega)$,the space of $\mathcal{F}_1-$measurable bounded\footnote{A $\mathcal{F}_1-$measurable random variable $X: \Omega\to\mathbb{R}$ is bounded if there exist a constant $M\leq \infty$ such as $|X(\bm{\omega})|\leq M$ with probability $1$} random variables $X: \Omega\to\mathbb{R}$.
\item  $L_1^{p}(\Omega)$,the space of $\mathcal{F}_1-$measurable bounded random variables $X: \Omega\to\mathbb{R}$, such as $\mathbb{E}_{\nu}\left[|X(1,\bm{\omega})|^p\right]<\infty$.
\item $H_{[0,1]}^p(\Omega)$, the space of $n-$dimensional adapted processes $\bm{r}: [0,1]\times\Omega\to\mathbb{R}^n$ satisfying $\mathbb{E}_{\nu}\left[\,\left(\int^1_0\text{d}q\,\|\bm{r}(q,\bm{\omega})\|^2\right)^{\frac{p}{2}}\,\right]<\infty$, with $p\geq 1$.
\item $S_{[0,1]}^p(\Omega)$, the space of $n-$dimensional adapted processes $\phi:[0,1]\times \Omega\to\mathbb{R}$ satisfying $\mathbb{E}_{\nu}\Big[\,\underset{q\in[0,1]}{\sup}|\phi(q,\bm{\omega})|^{p}\,\Big]< \infty$.
\end{itemize}
For convenience, we introduce the following notation
\begin{equation} 
\label{norms}
\begin{gathered}
X\in L_1^{\infty}(\Omega) \,:\quad \underset{\bm{\omega}\in\Omega}{\Max} |X(1,\bm{\omega})|=\inf \,\left\{M\in \mathbb{R}; |X(1,\bm{\omega})|\leq M\,\,a.s.\right\}\,,\\
\bm{r}\in H_{[0,1]}^p(\Omega) \,:\quad \|\bm{r}\|_{2,p}=\mathbb{E}_{\nu}\left[\,\left(\int^1_0\text{d}q\,\|\bm{r}(q,\bm{\omega})\|^2\right)^{\frac{p}{2}}\,\right]^{\frac{1}{p}}\,,\\
\phi\in S_{[0,1]}^p(\Omega) \,:\quad \|\phi\|_{\infty,p}=\mathbb{E}_{\nu}\Big[\,\underset{q\in[0,1]}{\sup}|\phi(q,\bm{\omega})|^{p}\,\Big]^{\frac{1}{p}}\,.
\end{gathered}
\end{equation}
We denote by $\chi$ the set of increasing deterministic function $x:[0,1]\to[0,1]$:
\begin{equation}
\chi :=\left\{x:[0,1]\to[0,1];\,\, \forall\,0\leq q\leq q' \leq 1\,,\,\,x(q)\leq x(q')\,\,\right\}.
\end{equation}
Let us endow the set $\chi$ wuth the uniform norm $\|x\|_{\infty}=\sup_{q\in[0,1]}\,x(q)$.

Given a process $\bm{r}\in H_{[0,1]}^p(\Omega)$, with $p\geq 2$, a function $x\in\chi$ and a number $q'\in[0,1]$, let $\zeta(\bm{r},x|q')$ be the process defined as:
\begin{multline}
\zeta(\bm{r},x; q,\bm{\omega}|q')\\=
\int^q_{q'}x(q') \bm{r}(q'\,, \,\bm{\omega}) \cdot \text{d}\bm{\omega}(q') -\frac{1}{2}\int^q_{q'}dq'\,x^2(q')\,\left\|\bm{r}(q'\,, \,\bm{\omega})\right\|^2\,,\quad\text{if}\,\quad q\geq q'
\end{multline}
and
\begin{equation}
\zeta(\bm{r},x; q,\bm{\omega}|q')=0\,,\quad \text{if}\,\quad q<q'\,,
\end{equation}
where the shorthand notation $(q,\bm{\omega})$ ( or $(q,\omega)$) denotes the non-anticipating functional dependence on the Brownian motion as in \eqref{adapted_notation}. We set $\zeta(\bm{r},x)=\zeta(\bm{r},x|0)$.

As in \eqref{DDE_xr}, let $\mathcal{E}(x \bm{r})$ and $\mathcal{E}(x \bm{r}|q')$ be the DDE defined as
\begin{equation}
\begin{gathered}
\mathcal{E}(x \bm{r})=e^{\zeta(\bm{r},x)}\to \mathcal{E}(x \bm{r};q,\bm{\omega})=e^{\zeta(\bm{r},x;q,\bm{\omega})}\,,\\
\mathcal{E}(x \bm{r}|q')=e^{\zeta(\bm{r},x|q')}\to \mathcal{E}(x \bm{r};q,\bm{\omega}|q')=e^{\zeta(\bm{r},x;q,\bm{\omega}|q')}\,.
\end{gathered}
\end{equation}
The symbol $x\bm{r}\in$ denotes the process taking values $x(q)\bm{r}(q,\bm{\omega})$.

In the following, given any symbols $K$ and $\alpha$ and a number $q'\in[0,1]$, the notation $K(\alpha|q)$ refers to an adapted process with a functional dependency on a parameter $\alpha$ and the number $q'$, while $K(\alpha;q,\bm{\omega}|q')=K(\alpha;q,\{\bm{\omega}(q'');0\leq q''\leq q\}|q')$ is the a value of the process at a given "time" $q\in [0,1]$ and a given realization of the Brownian motion $\bm{\omega}$; we also set $K(\alpha)=K(\alpha|0)$.

Now, let us define the set $\widehat{D}_{[0,1]}(\Omega)\subset H_{[0,1]}^p(\Omega)$, with $p\geq 1$, where, for each adapted process $\bm{r}\in \widehat{D}_{[0,1]}(\Omega)$, there exist a constant $C_{\bm{r}}\geq 0$ such as
\begin{equation}
\label{boundAux}
\left| \zeta(\bm{r},x; 1,\bm{\omega}) \right|\leq C_{\bm{r}}\quad a.s.\,.
\end{equation}
As far as we know, the set $\widehat{D}_{[0,1]}(\Omega)$ is not a vector space. Note that, if $\bm{r}\in \widehat{D}_{[0,1]}(\Omega)$, then the DDE $\mathcal{E}(x \bm{r})$ is a true martingale. This property assures that:
\begin{equation}
\label{martingality}
\mathbb{E}[\mathcal{E}(x \bm{r};q,\bm{\omega}|q')|\mathcal{F}_{q''}]=1,\quad \forall \,\,\, 0\leq q''\leq q'\leq q\leq 1\,\,\,\text{and}\,\,\,\bm{r}\in \widehat{D}_{[0,1]}(\mathbb{R}^d)\,.
\end{equation}
and
\begin{equation}
\label{positivity}
\mathcal{E}(x \bm{r};q,\bm{\omega}|q')\geq e^{-C_{\bm{r}}}\quad a.s.\,.
\end{equation}
For this reason, the martingale $\mathcal{E}(x \bm{r})$ can be considered as a probability density function of a probability measure $\widetilde{\mathbb{W}}_{x \bm{r}}$ equivalent\footnote{Two probability measures $\widetilde{\mathbb{W}}$ and $\mathbb{W}$, defined on the same measurable space $(\Omega,\mathcal{F})$, are equivalent if, given any set $A\in\mathcal{F}$, then $\mathbb{W}[A]=0$ if and only if $\widetilde{\mathbb{W}}[A]=0$} to $\mathbb{W}$.

Given two processes $\bm{r}$ and $\bm{v}$ in $\widehat{D}_{[0,1]}(\Omega)$,  let us introduce the binary functional $\mathbb{D}_{KL}(\, \cdot \, \|\,\cdot \,):\widehat{D}_{[0,1]}(\Omega)\times \widehat{D}_{[0,1]}(\Omega)\to [0,\infty)$ defined by:
\begin{multline}
\mathbb{D}_{KL}(\bm{r}  \, \|\,\bm{v} \,)=\mathbb{E}_{\nu}\left[\mathcal{E}(x \bm{r};q,\bm{\omega})\log\left(\frac{\mathcal{E}(x \bm{r};q,\bm{\omega})}{\mathcal{E}(x \bm{v};q,\bm{\omega})}\right)\right]\\=\mathbb{E}_{\nu}\left[\mathcal{E}(x \bm{r};q,\bm{\omega})\int^1_0 dq\,x^2(q)\|\bm{r}(q,\bm{\omega})-\bm{v}(q,\bm{\omega})\|^2\right]
\end{multline}
Because of the property \eqref{boundAux}, the above quantity is defined for all pair of processes in $\widehat{D}_{[0,1]}(\Omega)\times \widehat{D}_{[0,1]}(\Omega)$. Let also define
\begin{equation}
\label{metric}
\mathbb{D}^{\text{sym}}_{KL}(\bm{r}  \, \|\,\bm{v} \,)=\max\big\{\mathbb{D}_{KL}(\bm{r}  \, \|\,\bm{v} \,),\mathbb{D}_{KL}(\bm{v}  \, \|\,\bm{r} \,)\big\}
\end{equation}
Note that the quantity $\mathbb{D}_{KL}(\bm{r}  \, \|\,\bm{v} \,)$ and $\mathbb{D}_{KL}(\bm{v}  \, \|\,\bm{r} \,)$ are actually the \emph{relative entropies} (or Kullback–Leibler divergences) between the two probability-densities/DDEs $\mathcal{E}(x \bm{v})$ and $\mathcal{E}(x \bm{r})$. Indeed, for any $\bm{r} \in \widehat{D}_{[0,1]}(\Omega)$, $\mathbb{D}_{K,L}(\bm{r}  \, \|\,\bm{r} \,)=0$.

Moreover, if two processes $\bm{r}$ and $\bm{v}$ in $\widehat{D}_{[0,1]}(\Omega)$ verify the relation
\begin{equation}
\label{eqRel}
\mathbb{D}^{\text{sym}}_{KL}(\bm{r}  \, \|\,\bm{v} \,)=0\,,
\end{equation}
then the two corresponding DDEs are two statistical equivalent densities, i.e for each $\mathcal{F}-$measurable bounded random variable $A:\Omega\to \mathbb{R}$, they verify the equivalence:
\begin{equation}
\label{MyEquality}
\mathbb{E}\left[\,\mathcal{E}(x \bm{r})A(1,\bm{\omega})\right]=\mathbb{E}\left[\,\mathcal{E}(x \bm{v})A(1,\bm{\omega})\right].
\end{equation} 
The relations \eqref{eqRel} and \eqref{MyEquality} are equivalence relations.

By the relation \eqref{MyEquality}, we may argue that if \eqref{eqRel} holds, then the processes $\bm{r}$ and $\bm{v}$ are "similar", in some sense. This observation justifies the introduction of the quotient set $D_{[0,1]}(\Omega)$.
\begin{definition}
\label{define_D}
Given a process $\bm{r}\in\widehat{D}_{[0,1]}(\Omega)$, let $[\bm{r}]$ be the equivalence class
\begin{equation}
[\bm{r}]:=\{\bm{v}\in\widehat{D}_{[0,1]}(\Omega);\,\,\mathbb{D}^{\text{sym}}_{KL}(\bm{r}  \, \|\,\bm{v} \,)=0\}\,.
\end{equation}
We denote by $D_{[0,1]}(\Omega)$ the set of the equivalence classes:
\begin{equation}
D_{[0,1]}(\Omega):=\{[\bm{r}];\,\,\bm{r}\in\widehat{D}_{[0,1]}(\Omega)\}\,.
\end{equation}
\end{definition}
By abuse of notation, we will henceforth omit the square bracket around the elements of the space $D_{[0,1]}(\Omega)$.

Now, we provides the fundamental definitions.
\begin{definition}
\label{definitions}
For a given $n\in\mathbb{N}$, we call RSB value process the functional $\Gamma^{(n)}: L_1^{\infty}(\Omega)\times \chi \times D_{[0,1]}(\Omega)\to S_{[0,1]}^p(\mathbb{R})$ defined as follows:
\begin{multline}
\label{value_function}
\Gamma^{(n)}(\Psi,x,\bm{r};q,\bm{\omega})=\mathbb{E}\left[\mathcal{E}(x \bm{r};1,\bm{\omega}|q)\Psi ( \,1,\bm{\omega}\,)\bigg|\mathcal{F}_q \right]\\
-\frac{1}{2}\mathbb{E}\left[\int^1_q dq'\,x(q')\mathcal{E}(x \bm{r};q',\bm{\omega}|q)\,\|\bm{r}(q', \,\bm{\omega}) \,\|^2\bigg|\mathcal{F}_q\,\right]\,,
\end{multline}
where
\begin{itemize}
\item the random variable $\Psi\in L_1^{\infty}(\Omega)$ is the claim;
\item the function $x\in\chi$ is the POP (Parisi Order Parameter);
\item the process $\bm{r}\in D_{[0,1]}(\Omega)$ is the control parameter.
\end{itemize}
We say that a pair $(\Psi,x)\in L_1^{\infty}(\Omega)\times \chi$ allows the RSB expectation if there exists a solution pair $(\phi^{(n)}(\Psi,x),\bm{r}(\Psi,x))\in S_{[0,1]}^p(\Omega)\times  D_{[0,1]}(\Omega)$ such that:
\begin{equation}
\phi^{(n)}(\Psi,x;q,\bm{\omega})=\Gamma^{(n)}(\Psi,x,\bm{r}(\Psi,x);q,\bm{\omega}))=\sup_{\bm{r}\in  D_{[0,1]}(\Omega)}\Gamma^{(n)}(\Psi,x,\bm{r};q,\bm{\omega}))\,.
\end{equation}
and the following quantity
\begin{equation}
\label{auxiliary_variational_problem}
\Sigma^{(n)}(\Psi,x)=\int \text{d}\nu(\bm{\omega}(0))\phi^{(n)}(\Psi,x;0,\bm{\omega}(0))=\int \text{d}\nu(\bm{\omega}(0))\sup_{\bm{r}\in  D_{[0,1]}(\Omega)}\Gamma^{(n)}(\Psi,x,\bm{r};0,\bm{\omega}(0) )
\end{equation}
is the RSB expectation of $\Psi$, driven by $x$.
\end{definition}
For the rest of the chapter we will omit the superscript $\cdot^{(n)}$ and we consider a generic dimension $n$. We consider only bounded claims because of the fact that, at non-zero temperature, the random variable $\Psi^{\text{(e)}}$ and $\Psi^{\text{(v)}}$ , defined in \eqref{claimE} and \eqref{claimV}, are bounded. 

Throughout the chapter, we consider a real constant $c< \infty$ and assume that the claim $\Psi$ is bounded by:
\begin{equation}
\label{boundedness}
|\Psi(1,\bm{\omega})|\leq c\,,\quad a.s. \,.
\end{equation}

The aim of this chapter is obtaining the equation for the solution pair.

\section{Backward Stochastic Differential equations}
\label{sec5.2}
In this section we compute the variation of RSB$-$value process with respect the control parameter.

In the first subsection, we remind some properties of the Doléan-Dade exponential. In the second subsection, we provide a proper definition of the stationary condition \eqref{stationary0} and we get an equation for the control parameter.
\subsection{Properties of the Doléans-Dade exponential (DDE)}
\label{subsec6.2.1}
In this subsection we remind some fundamental facts about the DDEs and fix some notations. The following relations relies on the fact that, for any given process $\bm{r}\in D_{[0,1]}(\Omega)$, the DDE $\mathcal{E}(x\bm{r})$ defined in \eqref{DDE_xr} is a true martingale. The DDEs play a crucial role in stochastic theory and a vast literature has been produced about (see for example Chapter VIII of \cite{YoRev}).

The Doéans-Dade exponential (DDE) is defined as the unique strong solution of the following stochastic differential equation \cite{YoRev,Oksendal}:
\begin{equation}
\label{DDE_derivative}
\mathcal{E}(x\bm{r};q',\bm{\omega}|q)=1+\int^{q'}_{q}\mathcal{E}(x\bm{r};q'',\bm{\omega}|q') x(q)\bm{r}(q'',\bm{\omega})\cdot d\bm{\omega}(q)\,,\quad 0\leq q\leq q'\leq 1\,.
\end{equation}
As we stated above, the martingale condition implies that DDE  $\mathcal{E}(x\bm{r})$ is a positive process with
\begin{equation}
\label{MartiProp}
\mathbb{E}[\mathcal{E}(x\bm{r};q,\bm{\omega})]=1\quad \forall q\in[0,1].
\end{equation}
Then, we can define a probability measure $\widetilde{\mathbb{W}}_{x \bm{r}}$ on the measurable space $(\Omega,\mathcal{F})$, equivalent to the Wiener measure $\mathbb{W}$ and such as the DDE $\mathcal{E}(x\bm{r})$ is the Radon-Nikodym derivative of $\widetilde{\mathbb{W}}_{x \bm{r}}$ with respect to $\mathbb{W}$ \cite{YoRev,Billingsley}:
\begin{equation}
\frac{d\widetilde{\mathbb{W}}_{x \bm{r}} }{d\mathbb{W}}(\bm{\omega})=\mathcal{E}\left(\,x\bm{r};1,\bm{\omega}\,\right)\,.
\end{equation}

\begin{definition}
\label{DDE_expectation}
Let $\bm{r}\in D_{[0,1]}(\Omega)$ and $x\in \chi$. For any $\mathcal{F}_1-$measurable random variable $A$, the expectation of $A$ with respect the probability measure $\widetilde{\mathbb{W}}_{x \bm{r}}$ is the linear functional $A\mapsto \widetilde{\mathbb{E}}_{x\bm{r}}[A]\in \mathbb{R}$, defined as follows:
\begin{equation}
\widetilde{\mathbb{E}}_{x\bm{r}}\left[\,A(\,\bm{\omega}\,)\right]=\mathbb{E}\left[\mathcal{E}\left(\,x\bm{r};1,\bm{\omega}\,\right)\,A(\,\bm{\omega}\,)\,\right].
\end{equation}
Moreover, from Bayes Theorem, the conditional expectation value is given by
\begin{equation}
\label{conditionalNot}
\widetilde{\mathbb{E}}_{x\bm{r}}\left[\,A(\,\bm{\omega}\,)\big|\mathcal{F}_q\right]=\mathbb{E}\left[\mathcal{E}(x\bm{r};1,\bm{\omega}|q)\,A(\,\bm{\omega}\,)\,\big|\mathcal{F}_q\right],\quad \forall q\in[0,1]\,.
\end{equation}
\end{definition}
The symbol $\widetilde{\mathbb{E}}_{x\bm{r}}$ will be widely used throughout the thesis. For any process $\bm{r}\in D_{[0,1]}(\Omega)$, we define the vector semimartingale $\bm{W}_{x\bm{r}}$ such as:
\begin{equation}
\label{semimart}
\bm{W}_{x\bm{r}}\left(\,q,\bm{\omega}\right)=\bm{\omega}(q)-\int^q_0 dq'\, x(q')\, \bm{r}(\bm{\omega},\,q'\,)
\end{equation}
and
\begin{equation}
\label{AuxProb}
d\bm{W}_{x\bm{r}}\left(\,q,\bm{\omega}\right)=\text{d}\bm{\omega}(q)- \text{d}q\, x\, \bm{r}\left(\bm{\omega},\,x\,\right).
\end{equation}

By CMG Theorem \cite{CameronMartin,Girsanov}, the vector semimartingale $\bm{W}_{x\bm{r}}$ is a vector Brownian motion with respect the probability measure $\widetilde{\mathbb{W}}_{x \bm{r}}$ and the filtration $\{\mathcal{F}_q\}_{q\in[0,1]}$.

As a consequence, the stochastic integral of any vector processes $\bm{u} \in H^p_{[0,1]}(\Omega)$ (for any $p\geq 1$)  with respect the process $\bm{W}_{x\bm{r}}$, is a $\{\mathcal{F}_q\}_{q\in[0,1]}-$martingale with respect $\widetilde{\mathbb{W}}_{x \bm{r}}$, that implies:
\begin{equation}
\label{MartiPropMod}
\widetilde{\mathbb{E}}_{x\bm{r}}\left[\,\int^{q_1}_{0} \,\bm{u}(\,q\,, \,\bm{\omega}\,) \cdot d\bm{W}_{x\bm{r}}\left(\,q,\bm{\omega}\right)\Bigg|\mathcal{F}_{q_2}\right]=
\begin{cases}
\int^{q_1}_{0} \,\bm{u}(\,q\,, \,\bm{\omega}\,) \cdot d\bm{W}_{x\bm{r}}\left(\,q,\bm{\omega}\right),\, \text{if}\, q_1\leq q_2,\\
\int^{q_2}_{0} \,\bm{u}(\,q\,, \,\bm{\omega}\,) \cdot d\bm{W}_{x\bm{r}}\left(\,q,\bm{\omega}\right),\, \text{if}\,q_1>q_2,
\end{cases}
\end{equation}
and
\begin{equation}
\label{MartiPropMod1}
\widetilde{\mathbb{E}}_{x\bm{r}}\left[\,\int^{1}_0 \,\bm{u}(\,q\,, \,\bm{\omega}\,) \cdot d\bm{W}_{x\bm{r}}\left(\,q,\bm{\omega}\right)\,\}\right]=0.
\end{equation}
Moreover, the expectation value of the product between two stochastic integrals verify the It\^o isometry
\begin{multline}
\label{MartiPropMod2}
\widetilde{\mathbb{E}}_{x\bm{r}}\left[\,\int^{1}_0 \,\bm{u}(q,\bm{\omega}) \cdot d\bm{W}_{x\bm{r}}\left(\,q,\bm{\omega}\right)\int^{1}_0 \,\bm{v}(q,\bm{\omega}) \cdot d\bm{W}_{x\bm{r}}\left(\,q,\bm{\omega}\right)\,\}\right]=\\
\widetilde{\mathbb{E}}_{x\bm{r}}\left[\,\int^{1}_0 \text{d}q\,\bm{u}(q,\bm{\omega}) \cdot \bm{v}(q,\bm{\omega})\,\}\right],
\end{multline}
for any pair of processes $\bm{u}$ and $\bm{v}$ in $H_{[0,1]}^p(\Omega)$, with $p\geq 2$.

By combining \eqref{MartiPropMod1} with the definition \eqref{semimart}, we also have:
\begin{equation}
\label{integration_stoc}
\widetilde{\mathbb{E}}_{x\bm{r}}\left[\,\int^{1}_0 \,\bm{u}(\,q\,, \,\bm{\omega}\,) \cdot d\bm{\omega}(q)\right]=\widetilde{\mathbb{E}}_{x\bm{r}}\left[\,\int^{1}_0 \text{d}q\,x(q)\bm{u}(q,\bm{\omega}) \cdot \bm{r}(q,\bm{\omega})\right]
\end{equation}
We end the subsection by providing the following notation. Let $\bm{r}$ and $\bm{v}$ be two processes in $D_{[0,1]}(\Omega)$, we define
\begin{equation}
\mathcal{E}(x \bm{v};q',\bm{W}_{x \bm{r}}|q)=\exp\left(\int^{q'}_q \,\bm{v}(q,\bm{\omega}) \cdot d\bm{W}_{x\bm{r}}(q,\bm{\omega})-\frac{1}{2}\int^{q'}_q \text{d}q\,x(q)\|\bm{v}(q,\bm{\omega}) \|^2\right)\,.
\end{equation}
The DDE $\mathcal{E}(x \bm{v})$ is a martingale with respect the probability measure $\widetilde{\mathbb{W}}_{x \bm{r}} $ and verify all the above relation, by substituting $\mathbb{E}\to\widetilde{\mathbb{E}}_{x\bm{r}}$ and $\bm{\omega}(q)\to \bm{W}_{x\bm{r}}(q,\bm{\omega})$.

Now, we have all the necessary tools to address the study of the RSB-expectation.
\subsection{The stationary condition}
Now, we provides a proper definition of "stationary condition". Throughout the subsection the  claim $\Psi$ and the POP $x$ are kept fixed.

Given two vector processes $\bm{r}\in D_{[0,1]}(\Omega)$ and $\bm{u}\in D_{[0,1]}(\Omega)$ and a number $\epsilon \in [0,1]$, let
\begin{equation}
\Theta^{\epsilon}\left(\,\Psi,\,x,\,\bm{r},\bm{u};  q,\bm{\omega}\,\right)=\epsilon \Gamma\left(\,\Psi,x,\bm{r};\,q,\bm{\omega}\,\right)+(1-\epsilon)\Gamma\left(\,\Psi,x,\bm{u};\,q,\bm{\omega}\,\right)\,.
\end{equation}
We may consider the probability density function, defined as follow
\begin{equation}
\rho(\Psi,\,x,\,\bm{r},\bm{u};  q,\bm{\omega})=(1-\epsilon)\, \mathcal{E}(\,x\bm{r};q,\bm{\omega})+\epsilon\mathcal{E}(\,x\bm{u};q,\bm{\omega})\,.
\end{equation}
Note that, by \eqref{martingality} and \eqref{positivity}, since $0\leq \epsilon\leq 1$, then $\rho(\Psi,\,x,\,\bm{r},\bm{u})$ is a strictly positive and bounded martingale of mean $1$. As a consequence, there exists a process $\bm{v}^{\epsilon}(\bm{r},\bm{u})\in D_{[0,1]}(\Omega)$ such as:
\begin{equation}
\rho(\Psi,\,x,\,\bm{r},\bm{u};  q,\bm{\omega})=\mathcal{E}(\,x\bm{v}^{\epsilon}(\bm{r},\bm{u});q,\bm{\omega})\,,
\end{equation}
and
\begin{equation}
\Theta^{\epsilon}\left(\,\Psi,\,x,\,\bm{r},\bm{u};  q,\bm{\omega}\,\right)= \Gamma\left(\,\Psi,x,\bm{v}^{\epsilon}(\bm{r},\bm{u});\,q,\bm{\omega}\,\right)\,.
\end{equation}
A straightforward computation yields:
\begin{equation}
\bm{v}^{\epsilon}(\bm{r},\bm{u};q,\bm{\omega})=\frac{(1-\epsilon)\,\bm{r}(q,\bm{\omega}) \mathcal{E}(\,x\bm{r};q,\bm{\omega})+\epsilon\,\bm{u}(q,\bm{\omega})\mathcal{E}(\,x\bm{u};q,\bm{\omega})}{(1-\epsilon)\, \mathcal{E}(\,x\bm{r};q,\bm{\omega})+\epsilon\,\mathcal{E}(\,x\bm{u};q,\bm{\omega})}\,.
\end{equation}
Let
\begin{equation}
\label{direction}
\delta \bm{u}(q,\bm{\omega})=\partial_{\epsilon}\bm{v}^{\epsilon}(\bm{r},\bm{u};q,\bm{\omega})\big|_{\epsilon=0}\,,
\end{equation}
where the symbol $\partial_{\epsilon}$ denote the derivative over $\epsilon$ by taking all the other parameters fixed. 

We define the directional derivative of the functional $\bm{r}\mapsto \Gamma(\Psi,x,\bm{r})$ along the path $\{\bm{r},\bm{u}\}$ by
\begin{multline}
\label{GatDer}
\begin{aligned}
&\Pi\left(\,\Psi,\,x,\,\bm{r},\delta \bm{u};  q,\bm{\omega}\,\right)=\partial_{\epsilon}\Gamma\left(\,\Psi,x,\bm{v}^{\epsilon}(\bm{r},\bm{u});\,q,\bm{\omega}\,\right)\big|_{\epsilon=0}\\
&=\widetilde{\mathbb{E}}_{x\bm{r}}\left[\,\Psi ( \,1,\bm{\omega}\,)\int^1_q x(q')\,\delta \bm{u}(q',\bm{\omega}) \cdot d\bm{W}_{x\bm{r}}(q',\bm{\omega})\,\,\Bigg|\mathcal{F}_q \right]
\end{aligned}\\
-\frac{1}{2}\,\,\widetilde{\mathbb{E}}_{x\bm{r}}\left[\int^1_q \text{d}q'\,x(q')\,\int^{q'}_q x(q'')\,\delta \bm{u}(q'',\bm{\omega}) \cdot d\bm{W}_{x\bm{r}}(q'',\bm{\omega})\,\,\left \|\,\bm{r}(q', \bm{\omega}) \,\right\|^2\,\Bigg|\mathcal{F}_q\right]\\
-\,\widetilde{\mathbb{E}}_{x\bm{r}}\left[\int^1_q \text{d}q'\,x(q') \,\delta \bm{u}(q', \,\bm{\omega})\cdot\bm{r}(q', \bm{\omega}) \Bigg|\mathcal{F}_q\right].
\end{multline}
In the following, the process $x\delta \bm{u}$ will be called \emph{direction}.

We guess that the process $x\delta \bm{u}$, defined in \eqref{direction}, is actually a generic process on $H_{[0,1]}^p(\Omega)$. For this reason, we give the following definition of stationary point.
\begin{definition}[Stationary point]
A stationary point $\bm{r}^{*}$ of the functional $\Gamma(\Psi,x,\Cdot)$ is a vector stochastic process in $D_{[0,1]}(\Omega)$, such as the directional derivative \eqref{GatDer} vanishes  for any direction $x\,\delta\bm{u}\in H_{[0,1]}^p(\Omega)$:
\begin{equation}
\label{statCondPi}
\Pi\left(\,\Psi,\,x,\,\bm{r}^*,\delta \bm{u};  q,\bm{\omega}\,\right)=0\,,\quad \forall\,x\,\delta\bm{u}\in H^p_{[0,1]}(\Omega)\,.
\end{equation}
The above equation is the stationary condition.
\end{definition}
Such definition will be justified a posteriori. We want to obtain an equation for the auxiliary order parameter that is equivalent to the above stationary condition and does not depend on the derivative direction $x\delta \bm{u}$. 

By relation \eqref{MartiPropMod}, the second expectation value in \eqref{GatDer} can be rewritten in such a way
\begin{multline}
\frac{1}{2}\,\,\widetilde{\mathbb{E}}_{x\bm{r}}\Bigg[\,\int^1_q \text{d}q'\,x(q')\int^{q'}_q x(q'')\,\bm{u}(q'',\bm{\omega}) \cdot d\bm{W}_{x\bm{r}}(q'',\bm{\omega})\|\,\bm{r}(q', \,\bm{\omega}) \,\|^2\Bigg|\mathcal{F}_q\Bigg]=\\
\frac{1}{2}\widetilde{\mathbb{E}}_{x\bm{r}}\left[\,\int^1_q x(q')\,\bm{u}(q',\bm{\omega}) \cdot d\bm{W}_{x\bm{r}}\left(\,q,\bm{\omega}\right)\,\,\int^1_q \text{d}q'\,x(q')\, \|\,\bm{r}(q', \,\bm{\omega}) \|^2\Bigg|\mathcal{F}_q\right]
\end{multline}
and,  by formula \eqref{MartiPropMod2}, the third expectation value is
\begin{multline}
\widetilde{\mathbb{E}}_{x\bm{r}}\left[\int^1_q \text{d}q'\,x(q') \,\delta \bm{u}(q', \,\bm{\omega})\cdot\bm{r}(q', \,\bm{\omega}) \Bigg|\mathcal{F}_q\right]\\=\widetilde{\mathbb{E}}_{x\bm{r}}\left[\int^1_q \,x(q') \,\delta \bm{u}(q', \,\bm{\omega})\cdot d\bm{W}_{x\bm{r}}(q',\bm{\omega})\int^1_q \bm{r}(q', \,\bm{\omega}) \cdot d\bm{W}_{x\bm{r}}(q',\bm{\omega})\Bigg|\mathcal{F}_q\right]\,.
\end{multline}
Combining the above formulas in \eqref{GatDer}, we can rewrite the directional derivative \eqref{GatDer} in such a way:
\begin{equation}
\label{NewGat}
\Pi\left(\,\Psi,\,x,\,\bm{r},\delta\bm{u};  q,\bm{\omega}\,\right)
=\widetilde{\mathbb{E}}_{x\bm{r}}\left[\pi\left(\,\Psi,\,x,\,\bm{r};  1,\bm{\omega}\,|q\right)\int^1_q x(q')\,\bm{u}(q',\bm{\omega}) \cdot d\bm{W}_{x\bm{r}}(q',\bm{\omega})\Bigg|\mathcal{F}_q\right]\,,
\end{equation}
with
\begin{multline}
\label{piii}
\pi\left(\,\Psi,\,x,\,\bm{r};  1,\bm{\omega}\,|q\right)\\=\Psi ( \,1,\bm{\omega}\,)-\int^{1}_q \,\bm{r}(\,q'\,, \bm{\omega}\,) \cdot d\bm{W}_{x\bm{r}}(q',\bm{\omega})-\frac{1}{2}\int^1_q \text{d}q'\,x(q')\left \|\,\bm{r}(q', \,\bm{\omega}) \,\right\|^2.
\end{multline}
In the following we will refer to this quantity as \emph{random RSB}. 

A process $\bm{r}^{*}$ is a stationary point, according to the definition \eqref{statCondPi}, if and only if the random RSB $\pi\left(\,\Psi,\,x,\,\bm{r}|q\right)$ is uncorrelated, under the measure $\widetilde{\mathbb{W}}_{x \bm{r}}$, to all the random variables of the form
\begin{equation}
\label{A_direction}
A(1,\bm{\omega}|q\,)=\int^1_q x(q')\,\delta\bm{u}(q',\bm{\omega}) \cdot d\bm{W}_{x\bm{r}}(q',\bm{\omega}),\quad\text{with}\quad x\delta\bm{u}\in H^p_{[0,1]}(\Omega).
\end{equation}
This condition provides the stationary equation for the control parametr $\bm{r}$.  The following theorem is the most important results of this chapter.

\begin{theorem}
\label{stationariTheo}
Let us consider a claim $\Psi \in L_1^{\infty}(\Omega)$ and a POP $x\in \chi$. A control parameter $\bm{r}^{*}\in D_{[0,1]}(\Omega)$  verifies the condition \eqref{statCondPi} if and only if, at any time $q\in[0,1]$, the random RSB $\pi\left(\,\Psi,\,x,\,\bm{r}^*|q\right)$ is $\mathcal{F}_q-$measurable. 

In particular, this implies that there exists a process $\phi:[0,1]\times \Omega \to \mathbb{R}$, adapted to the filtration $\{\mathcal{F}_{q'}\}_{q'\in[0,1]}$, such as
\begin{equation}
\pi\left(\,\Psi,\,x,\,\bm{r};  1,\bm{\omega}\,|q\right)=\phi(q,\bm{\omega})\,,
\end{equation}
so we find the equation:
\begin{equation}
\label{selfEq1Aux}
\phi(q,\bm{\omega})=\Psi ( \,1,\bm{\omega}\,)-\int^{1}_q \,\bm{r}^{*}(q'\,, \bm{\omega}\,) \cdot d\bm{\omega}(q')+\frac{1}{2}\int^1_q \text{d}q'\,x(q') \|\,\bm{r}^{*}(q', \,\bm{\omega}) \|^2.
\end{equation}
The above equation is the stationary equation that generate the RSB expectation.
\end{theorem}
We remind that the random RSB is $\mathcal{F}_q-$measurable if, given a realization of the vector Brownian motion $\bm{\omega}$, $\pi\left(\,\Psi,\,x,\,\bm{r};1,\bm{\omega}|q\right)$ depends only on $\{\bm{\omega}(q'), \,0\leq q' \leq q\,\}$.
 The equation \eqref{selfEq1Aux} is the stationary equation of the control parameter.

Note that all the components of $\bm{r}^{*}$ and the process $\phi$ are unknowns of the equation. However, such class of equations may have a unique solution, since the condition that both the process $\phi$ and $\bm{r}$ are adapted provides a remarkable constrained. The right-hand member of the equation, indeed, is a sum of random quantities that are $\mathcal{F}_1$ measurable. We must look for a control parameter $\bm{r}^{*}$, depending only on the past, such as to "delete the dependence of the future".

The stationary equation can be rewritten in stochastic differential notation as
\begin{equation}
d\phi(q,\bm{\omega})=d\bm{\omega}(q)\cdot\bm{r}^{*}(q\,, \bm{\omega}\,)  -\frac{1}{2}\text{d}q\,x(q) \|\,\bm{r}^{*}(q, \,\bm{\omega}) \|^2
\end{equation}
together with the end point condition
\begin{equation}
\phi(1,\bm{\omega})=\Psi(1,\bm{\omega}).
\end{equation}
This kind of equation are called \emph{backward stochastic differential equation} (BSDE) \cite{PaPeng}.

BSDEs arise in many optimization and control problems, where we aim to fulfill a given "claim" (the claim $\Psi(1,\bm{\omega})$) and the control parameters depend only on the past(so we consider only adapted process).

\begin{proof}[Proof of Theorem \ref{stationariTheo}]
If the control parameter $\bm{r}^{*}\in D_{[0,1]}(\Omega)$ verifies the condition of \eqref{selfEq1Aux}, then there exists an adapted process $\phi$ such as
\begin{multline}
\Pi(\Psi,\,x,\,\bm{r}^{*},\delta\bm{u};q,\bm{\omega})\\
=\widetilde{\mathbb{E}}_{x\bm{r}^{*}}\left[\,\pi(\Psi,\,x,\,\bm{r}^{*};1,\bm{\omega}|q)\,\int^1_0 x(q)\delta \bm{u}(q,\bm{\omega}) \cdot d\bm{W}_{x\bm{r}^{*}}\left(\,q,\bm{\omega}\right)\Bigg|\mathcal{F}_q\right]\\
=\phi(q,\bm{\omega})\widetilde{\mathbb{E}}_{x\bm{r}^{*}}\left[\,\int^1_0 x(q)\delta \bm{u}(q,\bm{\omega}) \cdot d\bm{W}_{x\bm{r}^{*}}\left(\,q,\bm{\omega}\right)\Bigg|\mathcal{F}_q\right].
\end{multline}
The process $\phi$ can be pulled outside the conditional expectation value, since it is $\mathcal{F}_q-$measurable. Then, the stationarity \eqref{statCondPi} follows from the martingale property \eqref{MartiPropMod1}.

Conversely, suppose that the process $\bm{r}^{*}$ is a stationary point, according to the definition \eqref{statCondPi}. We have to show that the set of the random variables of the form \eqref{A_direction} is dense in $L_1^p(\Omega)$.

Consider the process $\bm{v}\in D_{[0,1]}(\Omega)$, defined as
\begin{equation}
\bm{v}(q,\bm{\omega})=\bm{f}(q) -x(q)\bm{r}^{*}(q,\bm{\omega})\,,
\end{equation}
where $\bm{f}:[0,1]\to \mathbb{R}^m$ is any deterministic and function such as $\int^1_0\text{d}q \,\|\bm{f}(q)\|^p=1$, with $p\geq 2$ . Using standard notation \cite{Berry}, we wright that $\bm{f}\in L^p([0,1],\mathbb{R}^n)$.

Since $x\, \delta \bm{u}\in H_{[0,1]}^p(\Omega)$, and $x(q)>0$ for all $q>0$, then we can consider a class of directions of the form:
\begin{equation}
\label{protoTypeder}
x(q') \delta \bm{u}(q',\bm{\omega})=\bm{v}(q',\bm{\omega})\mathcal{E}\big( \bm{v};q',\bm{W}_{x \bm{r}^{*}}\big|q\big)\,,\quad \text{with}\quad q'\geq q
\end{equation}
Then by property \eqref{DDE_derivative} of DDEs, one get
\begin{equation}
\int^1_q x(q')\delta \bm{u}(q',\bm{\omega}) \cdot d\bm{W}_{x\bm{r}^{*}}(q',\bm{\omega})\\
=\frac{\mathcal{E}(\bm{f};q',\bm{\omega}|q)}{\mathcal{E}(x\bm{r}^{*};q',\bm{\omega}|q)}-1\,.
\end{equation}
Replacing the above direction in \eqref{NewGat}, the stationary condition \eqref{statCondPi} yields
\begin{multline}
\mathbb{E}\left[\mathcal{E}(\bm{f};q',\bm{\omega}|q)\pi(\Psi,\,x,\,\bm{r}^{*};1,\bm{\omega}|q)\,\big|\mathcal{F}_q\right]\\
=\widetilde{\mathbb{E}}_{x\bm{r}^{*}}\left[\,\pi(\Psi,\,x,\,\bm{r}^{*};1,\bm{\omega}|q)\,\big|\mathcal{F}_q\right]\,,\,\forall\,\bm{f}\in L^p([0,1],\mathbb{R}^n).
\end{multline}
Note that the expectation value on the left-hand side member of the equation is with respect the probability measure $\mathbb{W}$ and on the right-hand the expectation is with respect $\widetilde{\mathbb{W}}_{x \bm{r}^{*}}$.

The linear span of the set $\left\{\mathcal{E}(\bm{f};q',\bm{\omega}|q),\bm{f} \in L^p([0,1],\mathbb{R}^n) \right\}$ is dense in $L^2(\Omega)$ (Lemma 4.3.2. in \cite{Oksendal}), so the above equation implies:
\begin{equation}
\label{in_the_proof}
\pi(\Psi,\,x,\,\bm{r}^{*};1,\bm{\omega}|q)
=\widetilde{\mathbb{E}}_{x\bm{r}^{*}}\left[\,\pi(\Psi,\,x,\,\bm{r}^{*};1,\bm{\omega}|q)\,|\mathcal{F}_q\right]\,\quad a.s..
\end{equation}
By the definition of conditional expectation, the right member in the above equation is an $\mathcal{F}_q-$measurable random variable, so we may consider an adapted process $\phi$, such as
\begin{equation}
\label{MeaningOfC}
\phi(q,\bm{\omega} )=\pi(\Psi,\,x,\,\bm{r}^{*};1,\bm{\omega}|q)\,.
\end{equation}
that conclude the proof.
\end{proof}
The meaning of the process $\phi$ is stated in the following 
\begin{corollary}
\label{GlobalMinimumC}
Given a claim $\Psi \in L_1^{\infty}(\Omega)$ and a POP $x\in \chi$, if the pair of processes $(\phi,\bm{r}^*)\in S_{[0,1]}^p(\Omega)\times D_{[0,1]}(\Omega)$ is a solution of the BSDE \eqref{selfEq1Aux}, then
\begin{equation}
\label{value_process_solution}
\phi(q,\bm{\omega})=\Gamma(\Psi,x,\bm{r}^*;q,\bm{\omega})\,.
\end{equation}
\end{corollary}
\begin{proof} 
Since $\phi(q,\Cdot)$ is $\mathcal{F}_q-$measurable, then $\phi(q,\bm{\omega})=\widetilde{\mathbb{E}}_{x\bm{r}^{*}}\left[\phi(q,\bm{\omega})|\mathcal{F}_q\right]$; then the proof is given by replacing $\phi$ by \eqref{selfEq1Aux} and using the relation \eqref{integration_stoc}.
\end{proof}
If the solution of the stationary condition has a unique solution and provides the global maximum of the RSB value process, then the BSDE \eqref{selfEq1Aux} determines completely the RSB expectation. We will discuss this matter in the next subsection.

\subsection{Global maximum condition}
In this subsection we prove that the solution of the stationary condition provides the global maximum of the RSB value function. We also discuss some property of the so-called RSB expectation, that we defined in \eqref{auxiliary_variational_problem}.

First of all, we need to state the following result. 
\begin{theorem}
\label{existence_uniqueness}
For any give claim and POP $(\Psi,x)\in L_1^{\infty}(\Omega)\times \chi $, there exist a unique pair of processes $(\,\phi(\Psi,x),\bm{r}(\Psi,x)\,)\in S_{[0,1]}^p(\Omega)\times D_{[0,1]}(\Omega)$ that is a soluion of the BSDE \eqref{selfEq1Aux}.
\end{theorem}
We will devote the next chapter to the proof of the existence result. The uniqueness is discussed in this subsection.

We proceed in the same way as in the proof of Corollary \ref{GlobalMinimumC}. Let $(\phi,\bm{r}^*)$ be a solution of \eqref{selfEq1Aux} corresponding to a claim $\Psi$ and a POP $x$. Since $\phi(q,\Cdot)$ is a $\mathcal{F}_q-$measurable random variable, the conditional expectation $\widetilde{\mathbb{E}}_{x\bm{v}}[\Cdot|\mathcal{F}_q]$ of both side in the equation \eqref{selfEq1Aux}, for any $\bm{v}\in D_{[0,1]}(\Omega)$, yields
\begin{multline}
\label{StraightfordComputation}
\phi(q,\bm{\omega})=\widetilde{\mathbb{E}}_{x\bm{v}}\left[\pi(\Psi,\,x,\,\bm{r}^{*};\,1,\bm{\omega}|q)\,|\mathcal{F}_q\right]\\=\widetilde{\mathbb{E}}_{x\bm{v}}\left[\Psi( \,1,\bm{\omega}\,)|\mathcal{F}_q\right]+\frac{1}{2}\widetilde{\mathbb{E}}_{x\bm{v}}\left[ \int \text{d}q x(q)\,\bm{r}^{*}(q, \,\bm{\omega})\,\cdot\left(\bm{r}^{*}(q, \,\bm{\omega})-2 \bm{v}(q, \,\bm{\omega})\right)\,\Bigg|\mathcal{F}_q\right].
\end{multline}
In the rest of the thesis, we will prefer to use a notation that explicitates the dependence of $\phi$ and $\bm{r}^*$ on the calim and the POP.The solution of the BSDE \eqref{selfEq1Aux}, for a given pair $(\Psi,x)\in L_1^{\infty}\times \chi$ will be denoted by $(\,\phi(\Psi,x),\bm{r}(\Psi,x)\,)$.

The following Lemma is an immediate consequence of the above identity.
\begin{lemma}
\label{lemma_Global_minimum}
Let us consider a claim $\Psi$ and a POP $x$ and the corresponding BSDE solution $(\,\phi(\Psi,x),\bm{r}(\Psi,x)\,)$. For any $\bm{v}\in D_{[0,1]}(\Omega)$ the RSB value process $\Gamma$ verifies:
\begin{multline}
\Gamma(\Psi,x,\bm{r}(\Psi,x);q,\bm{\omega})\\=\Gamma(\Psi,x,\bm{v};q,\bm{\omega})+\frac{1}{2}\mathbb{E}_{x \bm{v}}\left[ \int \text{d}q x(q)\,\|\bm{r}(\Psi,x;q, \,\bm{\omega})- \bm{v}(q, \,\bm{\omega})\|^2\Bigg|\mathcal{F}_q\right]
\end{multline}
\end{lemma}

\begin{proof}
Let us consider two processes  $\bm{r}$ and $\bm{v}$ in $ D_{[0,1]}(\Omega)$. By definition \ref{definitions}, the RSB value process $\Gamma(\Psi,x,\bm{v})$ is given by
\begin{equation}
\begin{aligned}
&\Gamma(\Psi,\,x,\,\bm{v};\,q,\bm{\omega})\\&=\widetilde{\mathbb{E}}_{x\bm{v}}\left[\Psi( \,1,\bm{\omega}\,)\,|\mathcal{F}_q\right]-\frac{1}{2} \widetilde{\mathbb{E}}_{x\bm{v}}\left[\int \text{d}q x(q)\,\left \|\,\bm{v}(q, \,\bm{\omega}) \,\right\|^2\,\Bigg|\mathcal{F}_q\right]\\
&\begin{multlined}
=\widetilde{\mathbb{E}}_{x\bm{v}}\left[\Psi ( \,1,\bm{\omega}\,)\,|\mathcal{F}_q\right]\\+\frac{1}{2} \widetilde{\mathbb{E}}_{x\bm{v}}\left[\int \text{d}q x(q)\,\,\bm{r}(q, \,\bm{\omega})\,\cdot\left(\bm{r}(q, \,\bm{\omega})-2 \bm{v}(q, \,\bm{\omega})\right)\,\Bigg|\mathcal{F}_q\right]\\-\frac{1}{2} \widetilde{\mathbb{E}}_{x\bm{v}}\left[\int \text{d}q x(q)\,\left \|\,\bm{v}(q, \,\bm{\omega}) -\bm{r}(q, \,\bm{\omega})\,\right\|^2\,\Bigg|\mathcal{F}_q\right]\,.
\end{multlined}
\end{aligned}
\end{equation}
Then, if $\bm{r}=\bm{r}(\Psi,x)$, then, by replacing the identity \eqref{StraightfordComputation} in the above formula, we get
\begin{equation}
\Gamma(\Psi,\,x,\,\bm{v};\,q,\bm{\omega})=\phi(\Psi,x;q,\bm{\omega})-\frac{1}{2} \widetilde{\mathbb{E}}_{x\bm{v}}\left[\int \text{d}q x(q)\,\left \|\,\bm{v}(q, \,\bm{\omega}) -\bm{r}(\Psi,x;q, \,\bm{\omega})\,\right\|^2\,\Bigg|\mathcal{F}_q\right]
\end{equation}
so, using Corollary \eqref{MeaningOfC}, we end the prove.
\end{proof}
From the above lemma, we obtain the most remarkable result of this section.
\begin{theorem}
\label{GlobalMinTheo}
Given the pair $(\Psi,x)\in L_1^{\infty}(\Omega)\times \chi$, let $(\,\phi(\Psi,x),\bm{r}(\Psi,x)\,)$ be the solution of the BSDE \eqref{selfEq1Aux}, then
\begin{equation}
\phi(\Psi,x;q,\bm{\omega})=\underset{\bm{r}\in D_{[0,1]}(\Omega)}{\max}\Gamma(\Psi,x,\bm{r};q,\bm{\omega})\,.
\end{equation}
\end{theorem}
An other important consequence of Lemma \ref{lemma_Global_minimum}, is the uniqueness result.
\begin{proof}[Proof of Theorem \ref{existence_uniqueness}: uniqueness]
For any pair of processes $\bm{r}_1$ and $\bm{r}_2$ in $D_{[0,1]}(\Omega)$, let
\begin{equation}
\mathbb{D}(\bm{r}_1,\bm{r}_2;q,\bm{\omega})=\widetilde{\mathbb{E}}_{x\bm{r}_1}\left[\int \text{d}q x(q)\,\left \|\,\bm{r}_1(q, \,\bm{\omega}) -\bm{r}_2(q, \,\bm{\omega})\,\right\|^2\,\Bigg|\mathcal{F}_q\right]
\end{equation}
The above quantity is obviously non-negative. 

Assume, by contradiction, that there exist two distinct pairs $(\phi_1,\bm{r}_1)$ and $(\phi_2,\bm{r}_2)$, on $ S_{[0,1]}^p(\Omega)\times D_{[0,1]}(\Omega)$, that are solutions of the BSDE \eqref{selfEq1Aux}. Applying Lemma \ref{lemma_Global_minimum} for both, we get
\begin{equation}
\Gamma\left(\Psi,x,\bm{r}_2\big|\,\bm{\omega}(0)\,\right)=\Gamma\left(\Psi,x,\bm{r}_1\big|\,\bm{\omega}(0)\,\right)-\frac{1}{2}\mathbb{D}(\bm{r}_2,\bm{r}_1;q,\bm{\omega})
\end{equation}
and
\begin{equation}
\Gamma\left(\Psi,x,\bm{r}_1\big|\,\bm{\omega}(0)\,\right)=\Gamma\left(\Psi,x,\bm{r}_2\big|\,\bm{\omega}(0)\,\right)-\frac{1}{2}\mathbb{D}(\bm{r}_1,\bm{r}_2;q,\bm{\omega}).
\end{equation}
After some straightforward manipulations, we get
\begin{equation}
\mathbb{D}(\bm{r}_1,\bm{r}_2;q,\bm{\omega})=-\mathbb{D}(\bm{r}_1,\bm{r}_2;q,\bm{\omega}).
\end{equation}
Since both $\mathbb{D}(\bm{r}_2,\bm{r}_1;q,\bm{\omega})$ and $\mathbb{D}(\bm{r}_1,\bm{r}_2;q,\bm{\omega})$ are non-negative, then the above relation implies
\begin{equation}
\mathbb{D}(\bm{r}_1,\bm{r}_2;q,\bm{\omega})=\mathbb{D}(\bm{r}_2,\bm{r}_1;q,\bm{\omega})=0.
\end{equation}
Moreover it is easy to show that:
\begin{equation}
0\leq \mathbb{D}_{KL}(\bm{r}_1\|\bm{r}_2)\leq\widetilde{\mathbb{E}}_{x\bm{r}_1}\left[\mathbb{D}(\bm{r}_1,\bm{r}_2;q,\bm{\omega})\right]\,.
\end{equation}
Then, by the definition of the vector space $ D_{[0,1]}(\Omega)$ \ref{define_D}, if the comparing function between two process $\bm{r}^{*}_1$ and $\bm{r}^{*}_2$ vanishes, than the two processes are equivalent and correspond to the same element of $ D_{[0,1]}(\Omega)$.
\end{proof}

We end the section by providing some property of the solution $(\phi(\Psi,x),\bm{r}(\Psi,x))$ associate to $(\Psi,x)$. The following statement are direct consequences of the maximum principle \ref{GlobalMinTheo} and the uniqueness in \ref{existence_uniqueness}.
\begin{proposition}
\label{Non:llinear_prop}
The solution of the BSDE \eqref{selfEq1Aux} verifies
\begin{itemize}
\item Let $\alpha$ be a constant, $(\phi(\alpha,x;q,\bm{\omega}) ,\bm{r}(\alpha,x;q,\bm{\omega}))=(\alpha,\bm{0})$ $a.s.$, $\forall q\in[0,1]$;
\item Let $\alpha$ be a constant, $(\phi(\alpha+\Psi,x;q,\bm{\omega}) ,\bm{r}(\alpha+\Psi,x;q,\bm{\omega}))=(\alpha+\phi(\Psi,x;q',\bm{\omega}) ,\bm{r}(\Psi,x;q,\bm{\omega}))$ $a.s.$, $\forall q\in[0,1]$;
\item If $\Psi_1\leq \Psi_2\,a.s.$, then $\phi(\Psi_1,x;q,\bm{\omega})\leq \phi(\Psi_2,x;q,\bm{\omega})\,a.s.$
\item $\phi(\alpha\,\Psi_1+\beta \,\Psi_2,x;q,\bm{\omega})\leq \alpha\,\phi(\Psi_1,x;q,\bm{\omega})+\beta\,\phi(\Psi_2,x;q,\bm{\omega})\,a.s.$ for any constant $\alpha$ and $\beta$.
\end{itemize}
\end{proposition}
\begin{proof}
The first and the second properties are trivially proved by observing that $(\alpha,\bm{0})$ (resp. $(\alpha+\phi(\Psi,x;q,\bm{\omega}) ,\bm{r}(\Psi,x;q,\bm{\omega}))$ ) is actually a solution of the BSDE.

For the third property, let $(\phi_1,\bm{r}_1)$ and $(\phi_2,\bm{r}_2)$ be the solutions associated to $(\Psi_1,x)$ and $(\Psi_2,x)$ respectively, with $\Psi_1\leq \Psi_2$, then
\begin{multline}
\phi(\Psi_1,x;q,\bm{\omega})=\Gamma(\Psi_1,x,\bm{r}_1;q,\bm{\omega})\\
\leq \Gamma(\Psi_2,x,\bm{r}_1;q,\bm{\omega})\leq \Gamma(\Psi_2,x,\bm{r}_2;q,\bm{\omega})=\phi(\Psi_2,x;q,\bm{\omega}).
\end{multline}
The fourth relation can be proved in a similar way.
\end{proof}
The above results extends to the non-Markovian RSB the stochastic representation of the Parisi PDE \eqref{start}, proposed by Chen and Auffinger (Theorem 3 in \cite{ChenAuf}). It is worth noting that Chen and Auffinger prove the variational representation of the Parisi formula, starting from the Parisi PDE and providing a stochastic representation.

Because non-Markovianity, we can not deal with a PDE, so the result is obtained by a completely different approach, based on the stochastic representation of the cavity method in \ref{C4}.

\section{Solution of Backward Stochastic Differential equations}
\label{sec5.3}
In this section, the existence of the solution of the stationary equation \eqref{selfEq1Aux} is proved. In the first paragraph, the solution is explicitly derived for a piecewise constant Parisi order parameter function $x$. Thence, in the second paragraph, the existence result is extended to any allowable Parisi order parameter by continuity. The results of the first section prove that the full-RSB-scheme of this thesis provides a complete theory, that take into account all the discrete-RSB solutions.

\subsection{The discrete-RSB solution}
We start by explicitly deriving the solution in the case where the Parisi parameter $x$ is a piecewise constant function:
\begin{equation}
x\in \chi^{\circ}
\end{equation}
A remarkable result of this paragraph is that, in this case, the free energy functional is equivalent to the one obtained in the discrete$-$RSB case. This proves that the continuity assumption described in section \eqref{subsec4.2.1} incorporates all the discrete$-$RSB theories, as we claimed previously.

Consider two increasing sequences of $K+2 \in \mathbb{N}$ numbers $q_0,\dots, q_{K+1}$ and $ x_0,\dots, x_{K+1}$ with
\begin{equation}
0=q_0\leq q_1\leq\dots\leq q_K\leq q_{K+1}=1
\end{equation}
and
\begin{equation}
0=x_0<x_1\leq\dots\leq x_K\leq x_{K+1}=1.
\end{equation}
The number $x_1$ must be non-zero.

The piecewise constant Parisi parameter function is constructed by such two sequences in such a way:
\begin{equation}
\label{discrete}
x(q)=\sum^{K+1}_{n=1} x_i \mathbb{1}_{(q_{i-1},q_{i}]}(q).
\end{equation}

In this case, the right hand member of the BSDE \eqref{selfEq1Aux} can be written as a sum of integrals defined on the intervals $[q_{i},q_{i+1})$; in each interval, the Parisi parameter $x$ is a constant and it can be put outside the integral:
\begin{multline}
\label{equation_Discrete}
\phi\left(\,q\,,\bm{\omega}\,\right)=\Psi(q_{K+1},\bm{\omega})\\-
\sum^{K-1}_{i=n_q}\left(\int^{q_{i+1}}_{q_{i}\wedge q} \,\bm{r}(\,p\,, \bm{\omega}\,) \cdot \text{d}\bm{\omega}(p)-\frac{1}{2} x_{i+1}\int^{q_{i+1}}_{q_{i}\wedge q} dp\,\,\left \|\,\bm{r}(\,p, \,\bm{\omega}\,) \,\right\|^2\,\right),
\end{multline}
where $n_q$ is the integer number such as
\begin{equation}
q_{n_q}\leq q<q_{n_q+1}
\end{equation}
and
\begin{equation}
q_{i}\wedge q=\max \left\{\,q,\,q_i\,\right\}.
\end{equation}
We want to derive a self-consistency equation for the process $\phi$ that is equivalent to the BSDE \eqref{equation_Discrete}.

Let us consider the sub-martingale $\zeta$, defined in such a way:
\begin{multline}
\zeta(q,\bm{\omega})=\log \,\mathcal{E}(x\bm{r};1,\bm{\omega}|q)\\
=\sum^{K}_{i=n_q}\left(x_{i+1}\int^{q_{i+1}}_{q_{i}\wedge q} \,\bm{r}(\,p\,, \bm{\omega}\,) \cdot \text{d}\bm{\omega}(p)-\frac{1}{2} x^2_{i+1}\int^{q_{i+1}}_{q_{i}\wedge q} dp\,\,\left \|\,\bm{r}(\,p, \,\bm{\omega}\,) \,\right\|^2\,\right),
\end{multline}
Using the equation \eqref{equation_Discrete}, one finds:
\begin{equation}
\label{relation_r_phi}
\zeta(q,\bm{\omega})=\sum^{K}_{i=n_q} x_{i+1}\left(\phi (\,q_{n+1}\,,\bm{\omega}\,)-\phi (\,q_{n}\wedge q\,,\bm{\omega}\,)\right),
\end{equation}
thus
\begin{equation}
\label{relation_G_phi}
\exp\left(x(q)\,\phi(q_{n_q+1},\bm{\omega})\right)=\exp\left(x(q)\,\phi(q,\bm{\omega})\right)\mathcal{E}\left(\,x\bm{r};q,\bm{\omega}|q_{n_q+1}\,\right).
\end{equation}
Let us assume for now that the DDE $\mathcal{E}(x\bm{r})$ is a true martingale. This assumption will be tested a posteriori. The martingale property implies that
\begin{equation}
\mathbb{E}\left[\mathcal{E}\left(\,x\bm{r};q,\bm{\omega}|q_{n_q+1}\,\right)\Big|\mathcal{F}_{q}\right]=1,
\end{equation}
and the random variable $\phi(q, \Cdot )$ is $\mathcal{F}_{q}$ measurable. Then the self-consistency equation for the RSB value process $\phi$ can be derived by considering the following identity:
\begin{multline}
\label{discrete_equation_manipulation}
\phi\left(\,q\,,\bm{\omega}\,\right)=\frac{1}{x(q)} \log\, \exp \left(\,x(q) \phi\left(\,q\,,\bm{\omega}\,\right)\,\right)\\
=\frac{1}{x(q)} \log\left(\, \exp \left(\,x(q) \phi\left(\,q\,,\bm{\omega}\,\right)\,\right)\mathbb{E}\left[\mathcal{E}\left(\,x\bm{r};q,\bm{\omega}|q_{n_q+1}\,\right)\Big|\mathcal{F}_{q}\right]\,\right).
\end{multline}
By replacing the equality \eqref{relation_G_phi} in the above representation, we get:
\begin{equation}
\label{discrete_equation}
\phi\left(\,q\,,\bm{\omega}\,\right)=\frac{1}{x(q)} \log\,\mathbb{E}\left[\exp \left(\,x(q)\phi(\,q_{n_q+1}\,,\bm{\omega}\,)\,\,\right)\big|\mathcal{F}_{q}\right].
\end{equation}

The equation \eqref{discrete_equation} computed at the discontinuity points $0,\,q_1,\cdots,\,q_{K}$ leads to an iterative backward map that allows to derive progressively the $K+2$ random variables $\phi(\,q_{K},\bm{\omega}\,)$, $\cdots$, $\phi(\,0,\bm{\omega}\,)$ from the Wiener functional $\Psi(\,1,\bm{\omega}\,)$:
\begin{equation}
\label{iterative_discrete_cont}
\phi\left(\,q_n,\bm{\omega}\,\right)=\frac{1}{x_{n+1}} \log\,\mathbb{E}\left[\exp \left(\,x_{n+1}\phi(\,q_{n+1}\,,\bm{\omega}\,)\,\,\right)\big|\mathcal{F}_{q_n}\right].
\end{equation} 
Note that the above iteration is equivalent to the discrete-RSB iteration \eqref{start1}, proving that the random variable $\phi\left(\,0,\cdot\,\right)$ has the same distribution of \eqref{start2}, as we claimed at the beginning of this paragraph.
\begin{proposition}
\label{inequality_prop_theo}
For any function $x\in \chi^{\circ}$, the process $\phi$, solution of the equation \eqref{discrete_equation}, is bounded, with
\begin{equation}
\label{inequality_prop}
\underset{\bm{\omega}\in\Omega}{\Max} |\phi\left(\,q,\bm{\omega}\,\right)|\leq c.
\end{equation}
\end{proposition}
As a consequence $\phi\in S^p_{[0,1]}(\Omega)$, for any $p\geq 1$.
\begin{proof}
We start by proving the boundedness of the random variables $\phi(\,q_{K},\bm{\omega}\,)$, $\cdots$, $\phi(\,0,\bm{\omega}\,)$, by decreasing induction on $q_n$. For $q=q_{K+1}=1$, the Wiener functional $\phi(1, \Cdot )$ is bounded by \eqref{boundedness}:
\begin{equation}
\label{inequality_tot}
 \underset{\bm{\omega}\in\Omega}{\Max} |\phi\left(\,1,\bm{\omega}\,\right)|\leq c 
\end{equation} 
and using the decreasing induction hypothesis on $q_n$ we get
\begin{multline}
\label{bound}
 \underset{\bm{\omega}\in\Omega}{\Max} |\phi\left(\,q_{n-1},\bm{\omega}\,\right)|= \underset{\bm{\omega}\in\Omega}{\max}\left|\frac{1}{x_{n}} \log\,\mathbb{E}\left[\exp \left(\,x_{n}\phi(\,q_{n}\,,\bm{\omega}\,)\,\,\right)\Big|\mathcal{F}_{q_{n-1}}\right]\right|\\
\leq \frac{1}{x_{n}} \log\,\mathbb{E}\left[\exp \left(\,x_{n}\underset{\bm{\omega}\in\Omega}{\Max}|\phi(\,q_{n}\,,\bm{\omega}\,)|\,\,\right)\Bigg|\mathcal{F}_{q_{n-1}}\right]\\= \underset{\bm{\omega}\in\Omega}{\Max}|\phi(\,q_{n}\,,\bm{\omega}\,)|,
\end{multline} 
proving that
\begin{equation}
\label{inequality_tot}
 \underset{\bm{\omega}\in\Omega}{\Max} |\phi\left(\,q_n,\bm{\omega}\,\right)|\leq \underset{\bm{\omega}\in\Omega}{\Max} |\phi\left(\,1,\bm{\omega}\,\right)|\leq c.
\end{equation} 
Boundedness property trivially extends to the whole process $\phi$ through the equation \eqref{bound}. 
\end{proof}
Now, it remains to derive the vector process $\bm{r}$ (the auxiliary order parameter) that verifies the equation \eqref{equation_Discrete} together with the process $\phi$ and such as the DDE $\mathcal{E}(x\bm{r})$ is a martingale.

The auxiliary order parameter is derived as follows. In each interval $[q_n,q_{n+1})$, with $0\leq n\leq K$, we define a process $J_n$ adapted to the filtration $\{\mathcal{F}_q,\,q\in[q_n,q_{n+1})\}$, in such a way:
\begin{equation}
J_n(q,\bm{\omega})=\mathbb{E}\left[\exp \left(\,x_{n+1}\phi(\,q_{n+1}\,,\bm{\omega}\,)\,\,\right)\Big|\mathcal{F}_{q}\right],\,\,q_{n}\leq q\leq q_{n+1}.
\end{equation}
Since the Wiener functional $\phi(\,q_{n+1}\,,\cdot\,)$ is bounded, the process $J_n$ is a strictly positive bounded martingale for $q\in[q_n, q_{n+1}]$. By the martingale representation theorem for the Brownian motion, there exists a unique process $\bm{M}_n$ in $H^p_{[q_n, q_{n+1}]}(\Omega)$, for any $p\geq1$, such as
\begin{equation}
J_n(q,\bm{\omega})=J_n(q_{n},\bm{\omega})+\int^{q}_{q_{n}}\bm{M}_n(p,\bm{\omega})\cdot \text{d}\bm{\omega}(p),\,\,q_{n}\leq q\leq q_{n+1}.
\end{equation}
Since the process $J_n$ is strictly positive and continuous for all $q\in[q_{n},q_{n+1})$, we can appy the It\^o formula to the process $\log J_n$, with the result that
\begin{multline}
\log J_{n_q}(q,\bm{\omega})=\\\log J_{n}(q_{n},\bm{\omega})+\int^q_{q_{n}}\frac{1}{J_n(p,\bm{\omega})}\bm{M}_n(p,\bm{\omega})\cdot \text{d}\bm{\omega}(p)-\frac{1}{2}\int^q_{q_{n}}\frac{1}{J^2_n(p,\bm{\omega})}\|\bm{M}(p,\bm{\omega})\|^2\,dp
\end{multline}
and thus we get
\begin{multline}
\label{proof_discrete}
\phi(q_{n_q+1},\bm{\omega})-\phi(q,\bm{\omega})=\frac{1}{x_{n_q}}\log J_{n_q}(q_{n_q},\bm{\omega})-\frac{1}{x_{n_q}}\log J_{n_q}(q,\bm{\omega})\\
=\int^{q_{n_q+1}}_{q}\frac{1}{x_{n_q+1} J_{n_q}(p,\bm{\omega})}\bm{M}_{n_q}(p,\bm{\omega})\cdot \text{d}\bm{\omega}(p)-\frac{1}{2}\int^{q_{n_q+1}}_{q}\frac{x_{n_q+1}}{x^2_{n_q+1}J^2_{n_q}(p,\bm{\omega})}\|\bm{M}_{n_q}(p,\bm{\omega})\|^2\,dp.
\end{multline}
Since $x_n>0$ for all $n>0$, then the integrals in the above equation are defined. 
Put
\begin{equation}
\label{auxiliry_Order_piecewise}
\bm{r}(q,\bm{\omega})=\sum^K_{n=1} \frac{1}{x_i J_{i-1}(q,\bm{\omega})}\bm{M}_{i-1}(q,\bm{\omega})\mathbb{1}_{[q_{i-1},q_{i})}(q)
\end{equation}
then the pair $(\phi,\,\bm{r})$ satisfies the BSDE \eqref{equation_Discrete}. Moreover the process $\bm{r}$ verifies the following remarkable property, that will be crucial for the rest of the section.
\begin{proposition}
\label{trivial_bb}
For all $\bm{\omega}\in \Omega$, the process $\bm{r}$, given by \eqref{auxiliry_Order_piecewise}, verifies the following inequality 
\begin{equation}
\label{bound_zet}
\left| \int^1_0x(q') \bm{r}(x'\,, \,\bm{\omega}) \cdot \text{d}\bm{\omega}(q') -\frac{1}{2}\int^1_0dq'\,x^2(q')\,\left\|\bm{r}(q'\,, \,\bm{\omega})\right\|^2\right|
\leq2 c\,
\end{equation}
that implies:
\begin{equation}
\label{bound_DDE}
e^{-2 c}\leq \mathcal{E}(\bm{r};q,\bm{\omega}) \leq e^{2c},\quad \forall \bm{\omega}\in\Omega
\end{equation}
\end{proposition}
The above results implies that the DDE process $\mathcal{E}(\bm{r})$ is a true martingale, so the vector process $\bm{r}$ is a proper solution of the auxiliary variational problem and identify an element of the domain set $ D_{[0,1]}(\Omega)$.
\begin{proof}
The proof of the inequality \eqref{bound_zet} is given by combining the equation \eqref{relation_r_phi} with the inequality \eqref{inequality_prop}, we obtain that
\begin{equation}
\begin{multlined}
\left| \int^1_0x(q') \bm{r}(q'\,, \,\bm{\omega}) \cdot \text{d}\bm{\omega}(q) -\frac{1}{2}\int^1_0dq'\,x^2(q')\,\|\bm{r}(q'\,, \,\bm{\omega})\|^2\right|
=|\zeta(0,\bm{\omega})|\\
\leq |\,\Psi(1,\bm{\omega})\,|+\sum^{K-1}_{i=n_q} (x_{i+1}-x_{i})|\phi (\,q_{i}\,,\bm{\omega})\,|\leq2 c.
\end{multlined}
\end{equation}
\end{proof}
It is worth noting that, since the processes $J_{i-1}$ and the function $x$ are strictly positive, and the process $\bm{M}_i$ is in $H_{[0,1]}^p(\Omega)$, then the process $\bm{r}$ is in $H_{[0,1]}^p(\Omega)$, for any $p\geq 1$. However, the boundedness of the process $\phi$ implies a stronger properties for the process $\bm{r}$ that is stated in the following proposition.

\begin{proposition}
\label{bound}
For every $p\in [0,\infty)$, there exist a universal constant $K_p$ such as for all the functions $x\in \chi^{\circ}$, the vector process $\bm{r}$ obtained by solving the equation \eqref{selfEq1Aux} verifies:
\begin{equation}
\label{bound_rr}
\widetilde{\mathbb{E}}_{x\bm{r}}\left[\left(\int^{1}_{q}dq'\,\|\bm{r}(q',\bm{\omega})\|^{2}\right)^p\right]\leq K_p
\end{equation}
and
\begin{equation}
\label{bound_rr_2}
\mathbb{E}\left[\left(\int^{1}_{q}dq'\,\|\bm{r}(q',\bm{\omega})\|^{2}\right)^p\right]\leq e^{2c} K_p.
\end{equation}
\end{proposition}
The second inequality is a trivial consequence of the first one and the proposition \ref{trivial_bb}, indeed:
\begin{equation}
\mathbb{E}\left[\left(\int^{1}_{q}dq'\,\|\bm{r}(q',\bm{\omega})\|^{2}\right)^p\right]\leq e^{2c}\widetilde{\mathbb{E}}_{x\bm{r}}\left[\left(\int^{1}_{q}dq'\,\|\bm{r}(q',\bm{\omega})\|^{2}\right)^p\right].
\end{equation}
So, we just prove \eqref{bound_rr}.
\begin{proof}
Let $(\tau_n,n\in\mathbb{N})$ be the sequence of stopping times defined as follows
\begin{equation}
\tau_n=\sup \left\{q\in[0,1];\,\,\int^q_0 dq'\,x(q')^2\|\bm{r}(q',\bm{\omega})\|^2\leq n^2 \right\}
\end{equation}
and put $\inf\, \emptyset=1$. For each $n\in\mathbb{N}$, we set
\begin{equation}
\bm{r}_n(q,\bm{\omega})=\bm{r}(q,\bm{\omega})\theta(\tau_n-q),
\end{equation}
where the function $\theta$ is the Heaviside theta function. Since the stochastic integral $\int^1_0 dq\,x(q)^2\|\bm{r}(q,\bm{\omega})\|^2$ has a finite expectation, then $\tau_n\uparrow 1\,a.s.$. We start by proving the proposition for $\bm{r}_n$ and then we take the limit $n\to \infty$.

Let us define
\begin{equation}
\zeta_n(\alpha,\beta; \bm{\omega})=\,\alpha\, \int^{\tau_{n}}_{0}x(q)\bm{r}(q,\bm{\omega})\cdot \text{d}\bm{\omega}(q)- \frac{\beta\,}{2} \int^{\tau_{n}}_{0}dp x(q)^2\,\|\bm{r}(q,\bm{\omega})\|^2,
\end{equation}
where $\alpha$ and $\beta$ are two real numbers. We have
\begin{equation}
\zeta(\bm{r},x;\tau_n,\bm{\omega})=\zeta_n\left(\tfrac{1}{2},\tfrac{1}{4},\bm{\omega}\right)+\zeta_n\left(\tfrac{1}{2},\tfrac{3}{4},\bm{\omega}\right)\leq \zeta_n\left(\tfrac{1}{2},\tfrac{1}{4},\bm{\omega}\right)+\tfrac{1}{2}\zeta(\bm{r},x;\tau_n,\bm{\omega}).
\end{equation}
The definition of $\tau_n$ and of the process $\bm{r}_n$ implies that DDE $\mathcal{E}(x\bm{r}_n/2)$, is a true strictly positive martingale, that is
\begin{equation}
\mathbb{E}\left[\mathcal{E}\left(\tfrac{1}{2}x\bm{r}_n;1,\bm{\omega}\right)\,\right]=\mathbb{E}\left[e^{\zeta_n\left(\frac{1}{2},\frac{1}{4},\bm{\omega}\right)}\,\right]=1,
\end{equation}
so we can consider the Girsanov change of measure from the $n-$component Wiener measure $\mathbb{W}$ to the equivalent measure $\widetilde{\mathbb{W}}_{x\bm{r}_n/2}$.

As usual, the symbol $\widetilde{\mathbb{E}}_{x\bm{r}_n/2}[\Cdot]$ will denotes the expectation value with respect the measure $\widetilde{\mathbb{W}}_{x\bm{r}_n/2}$ and the process $\bm {W}_{x\bm{r}_n/2}$ is the $n-$components vector Brownian motion with respect the measure $\widetilde{\mathbb{W}}_{x\bm{r}_n/2}$:
\begin{equation}
\label{brownian_proofrr}
\bm{W}_{\bm{r}_n/2}(q,\bm{\omega})=\bm{\omega}(q)-\frac{1}{2}\int^q_0 \text{d}q' x(q') \bm{r}(q',\bm{\omega}).
\end{equation}
From a straightforward computation, we get
\begin{equation}
\widetilde{\mathbb{E}}_{x\bm{r}_n}\left[\left(\int^{\tau_n}_0 \text{d}q'\,\|\bm{r}(q',\bm{\omega})\|^{2}\,\right)^{\frac{p}{2}}\,\right]
\leq \widetilde{\mathbb{E}}_{\frac{1}{2}x\bm{r}_n}\left[e^{\frac{1}{2}\zeta(\bm{r},x;\tau_n,\bm{\omega})}\left(\int^{\tau_n}_0 \text{d}q'\,\|\bm{r}(q',\bm{\omega})\|^{2}\Bigg|\,\right)^{\frac{p}{2}}\,\right].
\end{equation}
and H\"{o}lder inequality for any $p\geq1$ yields 
\begin{multline}
\label{holder_prrof}
\widetilde{\mathbb{E}}_{\frac{1}{2}x\bm{r}_n}\left[e^{\tfrac{1}{2}\zeta(\bm{r},x;\tau_n,\bm{\omega})}\left(\int^1_0 \text{d}q'\,\|\bm{r}_n(q',\bm{\omega})\|^{2}\Bigg|\,\right)^{\frac{p}{2}}\,\right]\\
\leq \widetilde{\mathbb{E}}_{\frac{1}{2}x\bm{r}_n}\left[\left(\int^1_0 \text{d}q'\,\|\bm{r}_n(q',\bm{\omega})\|^{2}\Bigg|\,\right)^{p}\,\right]^{\tfrac{1}{2}}\widetilde{\mathbb{E}}_{\frac{1}{2}x\bm{r}_n}\left[e^{\zeta(\bm{r},x;\tau_n,\bm{\omega})}\,\right]^{\tfrac{1}{2}}.
\end{multline}
and by Burkholder-Davis-Gundy inequality\cite{YoRev}, there exist a universal constant $C_p$, depending on $p$, such as:
\begin{multline}
\label{Burk}
\widetilde{\mathbb{E}}_{\frac{1}{2}x\bm{r}_n}\left[\left(\int^1_0 \text{d}q'\,\|\bm{r}_n(q',\bm{\omega})\|^{2}\Bigg|\,\right)^{p}\,\right]\\
\leq C_p \widetilde{\mathbb{E}}_{\frac{1}{2}x\bm{r}_n}\left[\left(\underset{q\in[0,1]}{\sup}\left|\int^1_0 d\bm{W}_{\frac{1}{2}\bm{r}_n}(q',\bm{\omega}) \,\bm{r}_n(q',\bm{\omega})\right|\,\right)^{2p}\,\right].
\end{multline}
By definition \eqref{brownian_proofrr} and the stationary equation \eqref{selfEq1Aux}, we have
\begin{multline}
\int^{q}_0 d\bm{W}_{\frac{1}{2}\bm{r}_n}(q',\bm{\omega}) \,\bm{r}_n(q',\bm{\omega})\\
=\int^q_0 \text{d}\bm{\omega}(q') \,\bm{r}_n(q',\bm{\omega})-\frac{1}{2}\int^q_0 \text{d}q' x(q') \,\|\bm{r}_n(q',\bm{\omega})\|^2\\=\phi(q)-\phi(0).
\end{multline}
By propositions \ref{inequality_prop} and \ref{bound_DDE}, the inequalities \eqref{holder_prrof} and \eqref{Burk} yield:
\begin{equation}
\widetilde{\mathbb{E}}_{x\bm{r}_n}\left[\left(\int^{\tau_n}_0 \text{d}q'\,\|\bm{r}(q',\bm{\omega})\|^{2}\,\right)^{\frac{p}{2}}\,\right]\leq \sqrt{C_p} e^{c}c^{p},\quad \forall n\in\mathbb{N}
\end{equation}
that proves the inequality \eqref{bound_rr} with $K_p= \sqrt{C_p} e^{c}c^{p}$.

The inequality \eqref{bound_rr_2} is an immediate consequence of the inequality \eqref{bound_DDE} and \eqref{bound_rr}.
\end{proof}
By proposition \ref{bound} the $H_{[0,1]}^p(\Omega)-$norm of the process $\bm{r}$ is dominated by a constant that doe not depends on the POP. This property will play a crucial role in the next paragraph.
\subsection{Extension to continuous POP}
In this paragraph we prove the existence of the solution of the BSDE \eqref{selfEq1Aux} when the Parisi order parameter is a generic increasing function $x\in \chi$. 
The proof is quite technical and relies on several intermediate results, that will be important also in the next chapters.

Intuitively, we may proceed by approximating the POP through elements of $\chi^{\circ}$. We show that, given a proper sequence of functions in $\chi^{\circ}$ that converges uniformly to the POP $x\in\chi$, the sequence of the solutions converges to a solution of the stationary equation \eqref{selfEq1Aux} corresponding to $x$. 

To this aim, we need to study the dependence of the processes defined in \eqref{discrete_equation} and \eqref{auxiliry_Order_piecewise} on the corresponding POP. Let us denote by $(\phi(x),\bm{r}(x) )$ the solution of the BSDE \eqref{selfEq1Aux} corresponding to a given POP $x\in\chi^{\circ}$.

Note that, since the elements of $\chi^{\circ}$ are strictly positive functions, the map $\chi^{\circ}\ni x\mapsto \phi(x)\in S^p_{[0,1]}(\Omega)$ is continuous and infinitely differentiable. By contrast, a continuous POP $x$ may be arbitrary close to $0$ at $q\to 0$, so the extension of this property to the general case is not obvious. 

The results in the next proposition allows to compare two process $\phi(x^{(1)})$ and $\phi(x^{(2)})$, corresponding to the piecewise constant POPs $x^{(0)}$ and $x^{(1)}$. 
\begin{proposition}
\label{derivative_theo}
Let $(\phi^{(t)},\bm{r}^{(t)})$ be the solutions relating to the Parisi order parameters $x^{(t)}$ respectively. The next results
Consider two POPs $x^{(0)}$ and $x^{(1)}$ in $\chi^{\circ}$. Let
\begin{equation}
\delta x=x^{(1)}-x^{(0)},
\end{equation}
and consider
\begin{equation}
x^{(t)}=(1-t)x^{(0)}+t x^{(1)}\in \chi^{\circ}.
\end{equation}
Let $(\phi^t,\bm{r}^t)$ be the solution corresponding to the POP $x^{(t)}$. Then, for all $q\in [0,1]$ and $t\in [0,1]$ and almost all $\bm{\omega}\in \Omega$, the quantity $\phi^t(q,\bm{\omega})$ is derivable on $t$ and
\begin{equation}
\label{first_derivative}
\frac{\partial \phi^t(q,\bm{\omega})}{\partial t}=\frac{1}{2}\widetilde{\mathbb{E}}_{\bm{r}^t}\left[\int^{1}_{q}\text{d}p\,\delta x(p)\,\|\bm{r}^t(p,\bm{\omega})\|^2\Bigg| \mathcal{F}_q\right],
\end{equation}
\end{proposition}
An immediate consequence of the above proposition is:
\begin{corollary}
Given tow POPs $x^{(0)}$ and $x^{(1)}$ in $\chi^{\circ}$ such as
\begin{equation}
x^{(0)}(q)\leq x^{(1)}(q),\,\forall q\in[0,1]
\end{equation}
then
\begin{equation}
\phi(x^{(0)})\leq \phi(x^{(1)}).
\end{equation}
\end{corollary}
\begin{proof}[Proof of Proposition \ref{derivative_theo}]
Let $K$ be the number of discontinuity points $0= q_0<q_1\leq \cdots q_K<q_{K+1}=1$ of the function $x^{(t)}$. 

We start by proving the formula of the first derivative. At $q=1$, the random variable $\phi^{(t)}(1,\Cdot)$ does not depends on $t$, that is
\begin{equation}
\label{recursion_derivative_start}
\frac{\partial}{\partial t}\phi(1,\bm{\omega})=0.
\end{equation}

For $q<1$, we proceed by differentiating the right member of the recursion \eqref{iterative_discrete_cont}. The chain rule yields a recursive equation for the derivative of $\phi^{(t)}$. For $q\in[q_n,q_{n+1})$, with $0\leq n\leq K$, we have:
\begin{multline}
\label{iterative_der_proto}
\frac{\partial}{\partial t}{\phi^{(t)}}(q,\bm{\omega})=\frac{\partial}{\partial t}\left(\frac{1}{x^{(t)}(q)} \log\,\mathbb{E}\left[\exp \left(\,x^{(t)}(q)\phi^{(t)}(\,q_{n_q+1}\,,\bm{\omega}\,)\,\,\right)\bigg|\mathcal{F}_{q}\right]\right)\\
=\frac{\delta x(q)}{x^{(t)}(q)}\frac{\mathbb{E}\left[\exp \left(\,\phi^{(t)}(\,q_{n_q+1}\,,\bm{\omega}\,)\,\,\right)\phi^{(t)}(\,q_{n_q+1}\,,\bm{\omega}\,)\big|\mathcal{F}_{q}\right]}{\mathbb{E}\left[\exp \left(\,x^{(t)}(q)\phi^{(t)}(\,q_{n_q+1}\,,\bm{\omega}\,)\,\,\right) \big|\mathcal{F}_{q}\right]}\\-\frac{\delta x(q)}{(x^{(t)}(q))^2} \log\,\mathbb{E}\left[\exp \left(\,x^{(t)}(q)\phi^{(t)}(\,q_{n_q+1}\,,\bm{\omega}\,)\,\,\right)\big|\mathcal{F}_{q}\right]\\
+\frac{\mathbb{E}\left[\exp \left(\,x^{(t)}(q)\phi^{(t)}(\,q_{n_q+1}\,,\bm{\omega}\,)\,\,\right)\tfrac{d}{dt}\phi^{(t)}(\,q_{n_q+1}\,,\bm{\omega}\,)\big|\mathcal{F}_{q}\right]}{\mathbb{E}\left[\exp \left(\,x^{(t)}(q)\phi^{(t)}(\,q_{n_q+1}\,,\bm{\omega}\,)\,\,\right) \big|\mathcal{F}_{q}\right]}.
\end{multline}
Now, the equation \eqref{iterative_discrete_cont} implies
\begin{equation}
\mathbb{E}_{\mathbb{W}^{\otimes(2/c)}}\left[\exp \left(\,x^{(t)}(q)\phi^{(t)}(\,q_{n_q+1}\,,\bm{\omega}\,)\,\,\right) \big|\mathcal{F}_{q}\right]=\exp(x^{(t)}_n\phi^{(t)}\left(\,q_{n_q+1}\,,\bm{\omega}\,)\right )
\end{equation}
and
\begin{equation}
\frac{\delta x(q)}{(x^{(t)}(q))^2} \log\,\mathbb{E}_{\mathbb{W}^{\otimes(2/c)}}\left[\exp \left(\,x^{(t)}(q)\phi^{(t)}(\,q_{n_q+1}\,,\bm{\omega}\,)\,\,\right)\big|\mathcal{F}_{q}\right]=\frac{\delta x(q)}{x^{(t)}(q)}\phi^{(t)}\left(\,q\,,\bm{\omega}\,\right)
\end{equation}
and, since $\phi(q,\Cdot)$ is $\mathcal{F}_q$ measurable, we have the following identity:
\begin{equation}
\phi(q,\bm{\omega})
=\widetilde{\mathbb{E}}_{\bm{r}^t}\left[\phi^t(q,\bm{\omega})\big| \mathcal{F}_q\right].
\end{equation}
By replacing the above three relations in the equation \eqref{iterative_der_proto}, we finally get
\begin{equation}
\label{iterative_der_0}
\frac{\partial \phi^{(t)}(q,\bm{\omega})}{\partial t}=
\frac{\delta x_{n_q}}{x_{n_q}^t}\widetilde{\mathbb{E}}_{\bm{r}^t}\left[\phi^t(q_{n+1},\bm{\omega})-\phi^t(q,\bm{\omega})\big| \mathcal{F}_q\right]+\widetilde{\mathbb{E}}_{\bm{r}^t}\left[\frac{\partial \phi^{(t)}(q_{n+1},\bm{\omega})}{\partial t}\big| \mathcal{F}_q\right].
\end{equation}
The solution of the above recursive equation, together with the starting condition \eqref{recursion_derivative_start}, is
\begin{equation}
\label{der_1}
\frac{\partial \phi^{(t)}(q,\bm{\omega})}{\partial t}=\sum^K_{n=n_q}\frac{\delta x_{n+1}}{x_{n+1}^t}\widetilde{\mathbb{E}}_{\bm{r}^t}\left[\phi^t(q_{n+1},\bm{\omega})-\phi^t(q_n\wedge q,\bm{\omega})\big| \mathcal{F}_{q}\right].
\end{equation}
Substituting the process $\phi$ with the stationary equation for discrete Parisi order parameter \eqref{equation_Discrete}, one finds:
\begin{equation}
\begin{multlined}
\frac{\delta x_{n}}{x_{n}^t}\widetilde{\mathbb{E}}_{\bm{r}^t}\left[\phi^t(q_{n+1},\bm{\omega})-\phi^t(q_n,\bm{\omega})\big| \mathcal{F}_{q_{n_q}}\right]\\
=\frac{\delta x_{n}}{x_{n}^t}\widetilde{\mathbb{E}}_{\bm{r}^t}\left[\int^{q_{n+1}}_{q_n} \text{d}\bm{\omega}(q') \cdot\bm{r}(q',\bm{\omega})-\frac{x_n^{(t)}}{2}\int^{q_{n+1}}_{q_n} \text{d}q' \,\|\bm{r}(q',\bm{\omega})\|^2\big| \mathcal{F}_{q_{n_q}}\right]\\
=\frac{\delta x_{n}}{2}\widetilde{\mathbb{E}}_{\bm{r}^t}\left[\int^{q_{n+1}}_{q_n} \text{d}q' \,\|\bm{r}(q',\bm{\omega})\|^2\big| \mathcal{F}_{q_{n_q}}\right],
\end{multlined}
\end{equation}
that proves \eqref{first_derivative}.
\end{proof}
From the above results and the proposition \ref{bound}, we deuce that the process $\partial_t \phi^t$ is almost surely bounded.

Now, we state the most remarkable property of the map $\chi^{\circ}\ni x\mapsto (\phi(x),\bm{r}(x)\,)\in S^p_{[0,1]}(\Omega)$
\begin{theorem}
\label{uniform_convergence_theo}
Let $x^{(1)}$ and $x^{(2)}$ be two elements of $\chi^{\circ}$. Then, for any $p>1$ there exist a constant $K_p$ depending only $p$ such that:
\begin{equation}
\label{uniform_converge_phi}
\mathbb{E}_{\mathbb{W}}\left[\left(\underset{q\in [0,1]}{\sup}\left|\phi(x^{(2)};q,\bm{\omega})-\phi(x^{(1)};q,\bm{\omega})\right|\right)^p\right]^{\frac{1}{p}}\leq K_p \|x^{(2)}-x^{(1)}\|_{\infty}
\end{equation}
and
\begin{equation}
\mathbb{E}_{\mathbb{W}}\left[\left(\int^1_0\left\|\bm{r}(x^{(2)};q,\bm{\omega})-\bm{r}(x^{(1)};q,\bm{\omega})\right\|\right)^p\right]^{\frac{1}{p}}\leq K_p \|x^{(2)}-x^{(1)}\|_{\infty}.
\end{equation}
\end{theorem}
This implies, in particular, that if two POP $x^{(1)}$ and $x^{(2)}$ are "close to each other”, then the "level of approximation” of the solutions $(\phi(x^{(1)},\bm{r}(x^{(1)})$ provided by the solution $(\phi(x^{(2)},\bm{r}(x^{(2)})$ depends only by the $\|\cdot\|_{\infty}-$distance between the two POPs. This result is very important. In fact, any POP in $\chi$ is arbitrary close to a POP in $\chi^{\circ}$.
\begin{proof}
 The inequality \eqref{uniform_convergence_theo} is an immediate consequence of the equation \eqref{first_derivative} and the proposition \ref{bound}.
Put
\begin{equation}
\delta x=x^{(2)}-x^{(1)},\quad \delta\phi=\phi(x^{(2)})-\phi(x^{(1)}),\quad \delta\bm{r}=\bm{r}(x^{(2)})-\bm{r}(x^{(1)}).
\end{equation} 
Since $\phi(x^{(1)})$ and $\phi(x^{(2)})$ are bounded and the processes $\bm{r}(x^{(1)})$ and $\bm{r}(x^{(2)})$ are in $H^p_{[0,1]}(\Omega)$, then the process $\delta \phi$ is bounded and $\delta\bm{r}$ is in $H^p_{[0,1]}(\Omega)$.

We use the same notation of theorem \ref{derivative_theo}. Let
\begin{equation}
x^{(t)}(q)=t x^{(1)}(q)+(1-t)x^{(2)}(q)\in \chi^{\circ}, \quad t\in [0,1].
\end{equation}
By theorem \ref{derivative_theo}, the process $\phi^{(t)}$ is derivable over $t$, that implies 
\begin{equation}
\delta \phi(q,\bm{\omega})=\int^1_0 dt \frac{\partial \phi^t(q,\bm{\omega})}{\partial t}= \frac{1}{2}\int^1_0 dt\,\widetilde{\mathbb{E}}_{\bm{r}^t}\left[\int^{1}_{q}\text{d}p\,\delta x(p)\,\|\bm{r}^t(p,\bm{\omega})\|^2\Bigg| \mathcal{F}_q\right]
\end{equation}
from which it follows that
\begin{equation}
\label{obvious_inequality}
|\delta \phi(q,\bm{\omega})|\leq\frac{1}{2}\int^1_0 dt\,\widetilde{\mathbb{E}}_{\bm{r}^t}\left[\int^{1}_{0}\text{d}p\,\delta x(q')\,\|\bm{r}^t(q',\bm{\omega})\|^2\Bigg| \mathcal{F}_q\right]
\end{equation}
Since the process $\delta\bm{r}$ is in $H^p_{[0,1]}(\Omega)$, then the process integrated over $t$ in the right-hand side of the above inequality is a non-negative martingale bounded in $L_p$ with respect the probability measure $\widetilde{\mathbb{W}}_{\bm{r}^t}$, for all $t\in [0,1]$. As a consequence, Doob inequality and the proposition \ref{bound} yield
\begin{equation}
\begin{multlined}
\widetilde{\mathbb{E}}_{\bm{r}^t}\left[\left(\underset{q\in[0,1]}{\sup}\widetilde{\mathbb{E}}_{\bm{r}^t}\left[\int^{1}_{0}\text{d}p\,\delta x(q')\,\|\bm{r}^t(q',\bm{\omega})\|^2\Bigg| \mathcal{F}_q\right]\right)^p\right]\\\leq\left(\frac{p}{p-1}\right)^p\widetilde{\mathbb{E}}_{\bm{r}^t}\left[\left(\,\int^{1}_{0}\text{d}q'\,\delta x(q')\,\|\bm{r}^t(q',\bm{\omega})\|^2\right)^p\right]\leq\left(\frac{p}{p-1}\right)^p(K_p\|\delta x\|_{\infty})^p,\\\quad \forall t\in[0,1]\,\,\,\text{and}\,\,\, p>1
\end{multlined}
\end{equation}
Combining the above result with the inequalities \eqref{obvious_inequality} and the proposition \ref{bound}, and since the process $\delta \phi$ is bounded, we finally get:
\begin{equation}
\begin{multlined}
\mathbb{E}_{\mathbb{W}}\left[\underset{q\in[0,1]}{\sup}|\delta \phi(q,\bm{\omega})|^p\right]\leq \,\mathbb{E}_{\mathbb{W}}\left[\underset{q\in[0,1]}{\sup}\left(\int^1_0 dt\,\frac{\partial \phi^t(q,\bm{\omega})}{\partial t}\right)^p\right]\\
\leq e^{2 c}\left(\frac{p}{2(p-1)}\right)^p\underset{t\in [0,1]}{\sup}\,\widetilde{\mathbb{E}}_{\bm{r}^t}\left[\left(\,\int^{1}_{0}\text{d}q'\,\delta x(q')\,\|\bm{r}^t(q',\bm{\omega})\|^2\right)^p\right]\\\leq \left(\frac{p\,K_p}{2(p-1)}\right)^pe^{2 c}\|\delta x\|^p_{\infty},\,\forall p>1
\end{multlined}
\end{equation}
that proves the inequality \eqref{uniform_converge_phi}, with $a_p=e^{2 c/p}p\,K_p/(2p-2)$.

Now we prove that the process $\delta\bm{r}$ has the same bound. At $q=1$, the process $\delta \phi$ verifies:
\begin{equation}
\delta\phi(1,\bm{\omega})=0,
\end{equation}
and from the two auxiliary stationary equations associated to the POPs $x^{(1)}$ and to $x^{(2)}$ and the above relation, one gets
\begin{equation}
\begin{multlined}
\label{equation_delta}
0=\delta \phi(1,\bm{\omega})
=\delta\phi(0,\bm{\omega})+\int^1_0 \delta \bm{r}(q,\bm{\omega}) \cdot \text{d}\bm{\omega}(q)-\frac{1}{2}\int^1_0 \text{d}q\,\delta x(q)\,\|\bm{r}_{2}(q,\bm{\omega})\|^2\\-\frac{1}{2}\int^1_0 \text{d}q\,x^{(1)}(q)\left(\bm{r}_{1}(q,\bm{\omega})+\bm{r}_{2}(q,\bm{\omega})\right)\cdot\delta \bm{r}(q,\bm{\omega}).
\end{multlined}
\end{equation}
By applying the It\^o formula to $(\delta \phi(1,\bm{\omega}))^2$, it follows that
\begin{equation}
\begin{multlined}
(\delta\phi(0,\bm{\omega}))^2+2\int^1_0\delta\phi(q,\bm{\omega})\bm{r}(q,\bm{\omega}) \cdot \text{d}\bm{\omega}(q)
-\int^1_0 \text{d}q\,\delta x(q)\,\delta\phi(q,\bm{\omega})\|\bm{r}_{2}(q,\bm{\omega})\|^2\\-\int^1_0 \text{d}q\,x^{(1)}(q)\delta\phi(q,\bm{\omega})\left(\,\bm{r}_{1}(q,\bm{\omega})+\bm{r}_{2}(q,\bm{\omega})\,\right)\cdot\delta \bm{r}(q,\bm{\omega})+\int^1_0 \text{d}q\|\delta\bm{r}_{2}(q,\bm{\omega})\|^2=0.
\end{multlined}
\end{equation}
All the quantities in the above expression are in $L_{[0,1]}^p(\Omega)$. We put $\int^1_0 \text{d}q\|\delta\bm{r}_{2}(q,\bm{\omega})\|^2$ on the left-hand side of the equation and the other terms in the right-hand side and take the absolute value raised to the power $p$ of both side. Using the inequality $|A+B+C+D|^p\leq4^{p-1}(|A|^p+|B|^p+|C|^p+|D|^{p})$ and taking the expectation value, we gets
\begin{equation}
\label{inequality_r}
\mathbb{E}_{\mathbb{W}}\left[\left(\int^1_0 \text{d}q\|\delta\bm{r}_{2}(q,\bm{\omega})\|^2\right)^p\right]\leq\text{I}+\text{II}+\text{III}+\text{IV}
\end{equation}
where
\begin{equation}
\text{I}=4^{p-1}\mathbb{E}_{\mathbb{W}}\left[\left|\delta\phi(0,\bm{\omega})\right|^{2p}\right]\leq 4^{p-1}a^{2p}_{2p}\|\delta x\|^{2p}_{\infty},
\end{equation}
\begin{equation}
\begin{multlined}
\text{II}=2*8^{p-1}\mathbb{E}_{\mathbb{W}}\left[\left|\int^1_0\delta\phi(q,\bm{\omega})\delta\bm{r}(q,\bm{\omega}) \cdot \text{d}\bm{\omega}(q)\right|^p\right]\\\leq 2^{\frac{5}{2}p-2}p^{p-1}(p-1)\mathbb{E}_{\mathbb{W}}\left[\left(\int^1_0\text{d}q \left(\delta\phi(q,\bm{\omega})\right)^{2}\|\delta\bm{r}(q,\bm{\omega})\|^2\right)^{\frac{p}{2}}\right]\\
\leq 2^{\frac{5}{2}p-2}p^{p-1}(p-1)\mathbb{E}_{\mathbb{W}}\left[\left(\underset{q\in[0,1]}{\sup}|\delta \phi(q,\bm{\omega})|^{p}\right)\left(\int^1_0\text{d}q \|\delta\bm{r}(q,\bm{\omega})\|^2\right)^{\frac{p}{2}}\right]\\
\leq 2^{\frac{5}{2}p-2}p^{p-1}(p-1)a_{2p}^{p}\,\|\delta x\|_{\infty}^{p} \mathbb{E}_{\mathbb{W}}\left[\left(\int^1_0\text{d}q \|\delta\bm{r}(q,\bm{\omega})\|^2\right)^{p}\right]^{1/2},
\end{multlined}
\end{equation}
\begin{equation}
\begin{multlined}
\text{III}=4^{p-1}\mathbb{E}_{\mathbb{W}}\left[\left|\int^1_0 \text{d}q\,\delta x(q)\,\delta\phi(q,\bm{\omega})\|\bm{r}_{2}(q,\bm{\omega})\|^2\right|^{p}\right]\\
\leq 4^{p-1}\mathbb{E}_{\mathbb{W}}\left[\left(\underset{q\in[0,1]}{\sup}|\delta \phi(q,\bm{\omega})|^{p}\right)\left(\int^1_0 \text{d}q\,|\delta x(q)|\|\bm{r}_{2}(q,\bm{\omega})\|^2\right)^{p}\right]\\
\leq 4^{p-1}\mathbb{E}_{\mathbb{W}}\left[\underset{q\in[0,1]}{\sup}|\delta \phi(q,\bm{\omega})|^{2p}\right]^{\frac{1}{2}}\mathbb{E}_{\mathbb{W}}\left[\left(\int^1_0 \text{d}q\,|\delta x(q)|\|\bm{r}_{2}(q,\bm{\omega})\|^2\right)^{2p}\right]^{\frac{1}{2}}\\
\leq 4^{p-1}a^p_{2p}\,e^c K^p_{2p}\|\delta x\|_{\infty}^{2p}
\end{multlined}
\end{equation}
\begin{equation}
\begin{multlined}
\text{IV}=4^{p-1}\mathbb{E}_{\mathbb{W}}\left[\left|\int^1_0 \text{d}q\,x^{(1)}(q)\delta\phi(q,\bm{\omega})\left(\,\bm{r}_{1}(q,\bm{\omega})+\bm{r}_{2}(q,\bm{\omega})\,\right)\cdot\delta \bm{r}(q,\bm{\omega})\right|^{p}\right]\\
\leq 4^{p-1}\mathbb{E}_{\mathbb{W}}\left[\underset{q\in[0,1]}{\sup}|\delta \phi(q,\bm{\omega})|^{p}\left|\int^1_0 \text{d}q\,x^{(1)}(q)\left(\,\bm{r}_{1}(q,\bm{\omega})+\bm{r}_{2}(q,\bm{\omega})\,\right)\cdot\delta \bm{r}(q,\bm{\omega})\right|^{p}\right]\\
\leq 4^{p-1}e^{\frac{ c}{2}}K^{p/2}_{2p}a_{4p}^{p}\|\delta x\|^{p}_{\infty}\mathbb{E}_{\mathbb{W}}\left[\left(\int^1_0 \text{d}q\,\|\delta \bm{r}(q,\bm{\omega})\|^2\right)^{p}\right]^{1/2}
\end{multlined}
\end{equation}
If we set
\begin{equation}
X=\mathbb{E}_{\mathbb{W}}\left[\left(\int^1_0\text{d}q \|\delta\bm{r}(q,\bm{\omega})\|^2\right)^{p}\right]^{1/2}
\end{equation}
and combine the above inequalities in \label{inequality_r}, the inequality \label{inequality_r} has the form:
\begin{equation}
\label{inequality_2}
X^2\leq \alpha_p \|\delta x\|_{\infty}^p X+\beta_p\|\delta x\|_{\infty}^{2p}
\end{equation}
where $\alpha_p$ and $\beta_p$ are two positive constants that depends only on $p$. That implies that
\begin{equation}
X\leq \,\frac{1}{2}\left(\alpha_p+\sqrt{\alpha^2_p+4\beta_p\,}\right)\,\|\delta x\|_{\infty}^{p}.
\end{equation}
and the proof is ended.
\end{proof}
We now come to the main result of this paragraph.
\begin{theorem}
\label{convergence_result}
Given a POP $x\in\chi$, consider a sequence of piecewise constant POPs $(x^{(k)})\subset \chi^{\circ}$, where
\begin{equation}
\|x^{(k)}-x\|_{\infty}\leq 2^{-k}
\end{equation}
The sequence of the solutions $\left(\,(\phi_{x^{(k)}},\bm{r}_{x^{(k)}})\,\right)$ converges almost surely and in $H^p_{[0,1]}(\Omega) \times H^p_{[0,1]}(\Omega)$ norm to a pair $(\phi,\bm{r})$ that is a solution of the auxiliary stationary equation \eqref{selfEq1Aux} corresponding to the POP $x$.
\end{theorem}

\begin{proof}
Let us consider the sequence of pairs of non-negative random variables $\left(\,(U_{k},V_{k})\,\right)$, where
\begin{equation}
U_k=\underset{q\in[0,1]}{\sup}\left|\phi_{x^{(k+1)}}(q,\bm{\omega})-\phi_{x^{(k)}}(q,\bm{\omega})\right|
\end{equation}
and
\begin{equation}
V_k=\int^1_0 \text{d}q\|\bm{r}_{x^{(k+1)}}(q,\bm{\omega})-\bm{r}_{x^{(k)}}(q,\bm{\omega})\|^2.
\end{equation}
Theorem \eqref{uniform_convergence_theo} yields:
\begin{equation}
\mathbb{E}_{\mathbb{W}}\big[ U^p_k\big]\leq a_p2^{-p k},\quad
\mathbb{E}_{\mathbb{W}}\big[V^p_k\big]\leq b_p2^{-p k},
\end{equation}
 consequently
\begin{equation}
\label{inequality_start}
\sum^{\infty}_{k=1}\mathbb{E}_{\mathbb{W}}\big[ U^p_k\big]\leq \infty,\quad \sum^{\infty}_{k=1}\mathbb{E}_{\mathbb{W}}\big[V^p_k\big]\leq \infty.
\end{equation}
from which it is straightforward to obtain that the sequence $\left(\,(\phi_{x^{(k)}},\bm{r}_{x^{(k)}})\,\right)$ converges in $S^p_{[0,1]}(\Omega) \times H^p_{[0,1]}(\Omega)$ to a pair $(\phi,\bm{r})$. Moreover, by Markov inequality and \eqref{inequality_start}, one gets
\begin{equation}
\mathbb{W}\left[\left\{\bm{\omega}\in\Omega;\,U_k>\epsilon \right\}\right]\leq \frac{\mathbb{E}\big[U^p_k\big]}{\epsilon^p}\leq a_p\left(\frac{1}{2^k\epsilon}\right)^p,
\end{equation}
and in the same way:
\begin{equation}
\mathbb{W}\left[\left\{\bm{\omega}\in\Omega;\,V_k>\epsilon \right\}\right]\leq b_p\left(\frac{1}{2^k\epsilon}\right)^p,
\end{equation}
where $\mathbb{W}[A]$ denotes the probability that the event $A$ occurs, according to the probability measure $\mathbb{W}$.

By Borel-Cantelli lemma \cite{Billingsley}, the above two Markov inequalities together with the convergence results in \eqref{inequality_start} imply that the sequence $\left((U^p_k,V_k^p)\right)$ converges almost surely to $(0,0)$ with respect the probability measure $\mathbb{W}$. In particular that implies that the sequence $\left(\,(\phi_{x^{(k)}},\bm{r}_{x^{(k)}})\,\right)$ converges almost surely to the pair $(\phi,\bm{r})$. Moreover, the definition of $(U_k)$ implies that $(\phi_{x^{(k)}})$ converges to $\phi$ almost surely uniformly in the interval $[0,1]$, so the process $\phi$ is continuous.

It remains to show that the pair $(\phi,\bm{r})$ is a solution of the equation \eqref{selfEq1Aux} corresponding to the POP $x$. Since the process $\bm{r}$ is in $ H^p_{[0,1]}(\Omega)$, for any $p\geq1$, then the It\^o integral $\int^1_q \text{d}\bm{\omega}(q')\cdot \bm{r}(q',\bm{\omega})$ and the integral $\int^1_q \text{d}q' x(q') \|\bm{r}(q',\bm{\omega})\|^2$ exist and are in $S_{[0,1]}^{q}(\Omega)$, for any $q\geq1$. Let
\begin{equation}
\text{I}_k(q;\bm{\omega})=\int^1_q \text{d}\bm{\omega}(q')\cdot \bm{r}_{x^{(k)}}(q',\bm{\omega}),\quad \text{II}_k(q,\bm{\omega})=\int^1_q \text{d}q' x^{(k)}(q') \|\bm{r}_{x^{(k)}}(q',\bm{\omega})\|^2
\end{equation}
and
\begin{equation}
\text{I}(q;\bm{\omega})=\int^1_q \text{d}\bm{\omega}(q')\cdot \bm{r}(q',\bm{\omega}),\quad \text{II}(q,\bm{\omega})=\int^1_q \text{d}q' x(q') \|\bm{r}(q',\bm{\omega})\|^2.
\end{equation} 
Note that $\text{I}_k$, $\text{II}_k$, $\text{I}$ and $\text{II}$ are not adapted process, so we use the notation $(q;\bm{\omega})$ instead of $(q,\bm{\omega})$. We must prove the almost sure convergence of the sequences $(\text{I}_k)$ and $(\text{II}_k)$ to $\text{I}$ and $\text{II}$ respectively.

Let us define the following non-negative random variables
\begin{equation}
G_k=\underset{q\in[0,1]}{\sup}\left|\,\text{I}_k(q;\bm{\omega})-\text{I}(q;\bm{\omega})\,\right|=\underset{q\in[0,1]}{\sup}\left|\int^1_q \text{d}\bm{\omega}(q')\cdot \left(\bm{r}_{x^{(k)}}(q',\bm{\omega})-\bm{r}(q',\bm{\omega})\right)\right|
\end{equation}
and
\begin{equation}
\begin{multlined}
F_k=\underset{q\in[0,1]}{\sup}\left|\,\text{II}_k(q;\bm{\omega})-\text{II}(q;\bm{\omega})\,\right|\\=\underset{q\in[0,1]}{\sup}\left|\int^1_q \text{d}q'\left(x^{(k)}(q')\|\bm{r}_{x^{(k)}}(q',\bm{\omega})\|^2-x(q') \|\bm{r}(q',\bm{\omega})\|^2\right)\right|.
\end{multlined}
\end{equation}
By BDG inequality, there is a positive constant $C_p$, depending only on $p$, such as
\begin{equation}
\mathbb{E}_{\mathbb{W}}\left[G_k^p\right]\leq C_p \mathbb{E}_{\mathbb{W}}\left[V_k^{p/2}\right]\leq C_p b_p 2^{-kp}.
\end{equation}
 Moreover, the inequality \eqref{bound_rr} yields
\begin{equation}
\begin{multlined}
\mathbb{E}_{\mathbb{W}}\left[F_k^p\right]\\\leq \big\|x^{(k)}-x\big\|^{p}_{\infty}\mathbb{E}_{\mathbb{W}}\left[\left(\int^1_q \text{d}q' \|\bm{r}^{(k)}(q',\bm{\omega})\|^2\right)^p\right]\\+\mathbb{E}_{\mathbb{W}}\left[\left(\int^1_q \text{d}q' x(q') \|\bm{r}^{(k)}(q',\bm{\omega})-\bm{r}(q',\bm{\omega})\|^2\right)^p\right]\\
\leq 2^{-k} e^{2c}(2c)^2+\mathbb{E}_{\mathbb{W}}\left[V_k^{p/2}\right]\leq 2 ^{-k} c_p
\end{multlined}
\end{equation}
where $c_p$ is a positive constant depending only on $p$. As for the sequences $(U_k)$ and $(V_k)$, the above two inequalities imply that the sequences $(G_k)$ and $(F_k)$ converge in $L^p(\mathbb{W},\Omega)$ and almost surely to $0$, so the random variables $(\text{I}_k)$ and $(\text{II}_k)$ converge almost surely uniformly in $q\in [0,1]$ to $\text{I}$ and $\text{II}$ respectively, and the proof is ended.
\end{proof}
Finally, we end this section with the following obvious, but important result
\begin{theorem}
\label{extending_theo}
The propositions \ref{inequality_prop_theo}, \ref{bound}, \ref{derivative_theo} and \ref{uniform_convergence_theo} hold for all the allowable POP $x\in \chi$. In particular, we have
\begin{equation}
\label{derivative_formula_on_X}
\frac{\delta\phi(x;\,q',\bm{\omega}) }{\delta x(q)}=\theta(q-q')\widetilde{\mathbb{E}}_{x\bm{r}}\left[ \|\bm{r}(q,\bm{\omega})\|^2\big|\mathcal{F}_{q'}\right].
\end{equation}
\end{theorem}

This result is a straightforward consequence of the convergence result of Theorem \ref{convergence_result}. 

Note that, by corollary \ref{GlobalMinimumC}, the above theorem provides some important properties on the dependence of the Non-Markov RSB expectation $\Sigma(\Psi,x)$ defined in \eqref{definitions}, although, for a general POP $x\in\chi$, we do not have an explicit form for it. In particular, the derivative formula \eqref{derivative_formula_on_X} will be used to derive the stationary equation of the physical auxiliary problem. A more detailed analysis of the RSB expectation is presented in the next section

\chapter{Analytical properties of the RSB expectation}
\label{C6}
\thispagestyle{empty}
Throughout this chapter, we refer to the notation introduced in Section \ref{sec5.1}.

In the previous chapter, we investigated the functional given by the following variational representation \eqref{auxiliary_variational_problem} 
\begin{equation}
\Sigma(\Psi,x)=\underset{\bm{r}\in D_{[0,1]}(\Omega)}{\sup}\,\, \Gamma\big(\Psi,x,\bm{r};0,\bm{0}\big)
\end{equation}
where $\Gamma$ is the RSB value process described in \eqref{value_function}. We proved that for any given bounded $\mathcal{F}_1$ measurable random variable $\Psi\in L_1^{\infty}(\Omega)$ (claim) and an increasing function $x:[0,1]\to [0,1]$ (POP), the above variational problem is defined, i.e. $\Sigma$ is a proper functional on $ L_1^{\infty}(\Omega)\times \chi$, and we called $\Sigma(\Psi,x)$ RSB expectation of $\Psi$ driven by $x$ (definition \ref{definitions}). 

The main result (theorems \ref{existence_uniqueness} and \ref{GlobalMinTheo}) was proving that there exists a unique pair of adapted processes $(\,\phi(\Psi,x),\bm{r}(\Psi,x)\,)\in S_{[0,1]}^p(\Omega)\times  D_{[0,1]}(\Omega)$ ($p\geq 1$) depending on the claim $\Psi$ and on the POP $x$, such as:
\begin{equation}
\Sigma(\Psi,x)= \Gamma\big(\Psi,x,\bm{r}(\Psi,x);0,\bm{0}\big)= \int \text{d}\nu(\,\bm{\omega}(0)\,)\phi(\,\Psi,x;\,0,\bm{\omega}(0)\,)
\end{equation}
where the pair $\big(\,\phi(\Psi,x),\bm{r}(\Psi,x)\,\big)$ is the solution of the following BSDE (theorem \ref{GlobalMinimumC})
\begin{equation}
\phi(\Psi,x;q,\bm{\omega})=\Psi(1,\bm{\omega})- \int \text{d}\bm{\omega}(q)\cdot \bm{r}(q,\bm{\omega})+\frac{1}{2}\int \text{d}qx(q)\| \bm{r}(q,\bm{\omega})\|^2.
\end{equation}
Moreover, the functional $\Sigma(\Cdot,x)$ is a non-linear expectation (Proposition \ref{Non:llinear_prop} ).

In this chapter we study on the dependence of the RSB expectation $\Sigma(\Psi,x)$ on the claim $\Psi$ and on the POP $x$. In particular we consider the case where the claim and the POP have a differentiable dependence on a given set of parameters. We compute the derivatives of the RSB expectation with respect such parameters (Section \ref{sec6.1}) and then we present a recursive method to compute the Taylor expansion (Section \ref{sec6.2}). 

Our aim is to develop several mathematical tools that will be helpful in the study of the physical variational problem, the we will address in the next chapter.
\section{First derivative}
\label{sec6.1}
In this section, we consider the case where the claim is differentiable with respect a given real parameter. We want to compute the derivative of the RSB expectation with respect such parameter. At the end of the section, we show that the derivative can be represented as the solution of a proper BSDE. The results of this section are the first building blocks of the study of the study of the physical properties of a system described by such kind of equations.

Let $\Psi(\alpha)$ be of bounded claim, depending on real parameter $\alpha$, taking value on a given compact connected interval $I\in \mathbb{R}$. We denote by $(\,\phi^{\alpha},\bm{r}^{\alpha}\,)$ the solution $(\,\phi(\Psi(\alpha),\,x),\bm{r}(\Psi(\alpha),\,x)\,)$ corresponding to the pair $(\,\Psi(\alpha),\,x\,)$, for a given POP $x\in\chi$. We assume, by hypothesis, that the function $\Psi(\Cdot):I\to L^{\infty}_1(\Omega)$ is uniformly bounded by a real number $L$:
\begin{equation}
\label{hypo_1}
\Max_{\bm{\omega}\in\Omega} \left|\Psi(\alpha;1,\bm{\omega})\right|<L\quad a.s.\,\quad \forall \alpha\in I,
\end{equation}
where the symbol $\Max$ is defined in \eqref{norms}.

We also impose that, the function $\Psi(\Cdot;1,\bm{\omega})$ is continuous and differentiable fror almost all $\bm{\omega}\in\Omega$, and there exist a positive $L^{\infty}_1(\Omega)$ Wiener functional $\xi$ such as:
\begin{equation}
\label{hypo_2}
\left|\partial_{\alpha} \Psi(\alpha;1,\bm{\omega})\right|<\xi(1,\bm{\omega})\quad a.s.\quad \forall \alpha \in I, \quad\text{with}\quad\mathbb{E}[\xi(1,\bm{\omega})]<\infty,
\end{equation}
where $\partial_{\alpha} \Psi(\alpha;1,\bm{\omega})$ is the partial derivative of $\Psi$ on $\alpha$ (namely the derivative on $\alpha$ of $\Psi(\Cdot;1,\bm{\omega})$, taking $\bm{\omega}\in\Omega$ fixed).

Consider two numbers $\alpha$ and $\alpha'$ in $I$. We use the following notation:
\begin{equation}
\delta \alpha =\alpha'-\alpha
\end{equation}
\begin{equation}
\delta \Psi=\Psi(\alpha')-\Psi(\alpha),
\end{equation}
and
\begin{equation}
\delta\Sigma =\Sigma(\Psi(\alpha),x)-\Sigma(\Psi(\alpha'),x).
\end{equation}
Let us define the process:
\begin{equation}
\overline{\bm{r}}=\frac{1}{2}\big(\bm{r}^{\alpha}+\bm{r}^{\alpha'}\big).
\end{equation}
By theorems \ref{uniform_convergence_theo} and \ref{extending_theo}, since $\Psi(\Cdot;1,\bm{\omega})$ is almost surely continuous, then the process $\overline{\bm{r}}$ converges to $\bm{r}^{\alpha}$ a.s. as $\alpha'\to\alpha$.

By Theorem \ref{lemma_Global_minimum} the RSB expectations $\Sigma(\Psi(\alpha),x)$ and $\Sigma(\Psi(\alpha'),x)$ have the following representation:
\begin{equation}
\label{relation_1}
\Sigma(\Psi(\alpha),x)=\Gamma(\Psi(\alpha'),x,\bm{v})+\frac{1}{2}\mathbb{E}\left[\int^1_0\mathcal{E}\big(x \bm{v};q,\bm{\omega}\big)x(q)\big\|\bm{v}-\bm{r}^{\alpha}\big\|^2\right],
\end{equation}
\begin{equation}
\label{relation_2}
\Sigma(\Psi(\alpha'),x)=\Gamma(\Psi(\alpha'),x,\bm{v})+\frac{1}{2}\mathbb{E}\left[\int^1_0 \text{d}q \mathcal{E}\big(x \bm{v};q,\bm{\omega}\big)x(q)\big\|\bm{v}-\bm{r}^{\alpha'}\big\|^2\right].
\end{equation}
If we take
\begin{equation}
\bm{v}=\overline{\bm{r}},
\end{equation}
then, comparing the formulas \eqref{relation_1} and \eqref{relation_2}, we get:
\begin{equation}
\label{discrete_derivative}
\frac{\delta\Sigma}{\delta \alpha}=\frac{1}{\delta \alpha}\big(\,\Gamma(\Psi(\alpha'),x,\overline{\bm{r}})-\Gamma(\Psi(\alpha),x,\overline{\bm{r}})\,\big)=\mathbb{E}\left[\mathcal{E}\big(x \overline{\bm{r}};1,\bm{\omega}\big)\frac{\delta \Psi(1,\bm{\omega})}{\delta \alpha}\right].
\end{equation}
Hypotheses \eqref{hypo_1} and \eqref{hypo_2} yield
\begin{equation}
\left|\mathcal{E}\big(x \overline{\bm{r}};1,\bm{\omega}\big)\frac{\delta \Psi(1,\bm{\omega})}{\delta \alpha}\right|\leq e^L \xi(1,\bm{\omega})\quad a.s.
\end{equation}
so by dominated convergence theorem, we have:
\begin{equation}
\lim_{\alpha'\to \alpha}\left(\,\mathbb{E}\left[\mathcal{E}\big(x \overline{\bm{r}};1,\bm{\omega}\big)\frac{\delta \Psi(1,\bm{\omega})}{\delta \alpha}\right]\,\right)\\=\mathbb{E}\left[\lim_{\alpha'\to \alpha}\left(\,\mathcal{E}\big(x \overline{\bm{r}};1,\bm{\omega}\big)\frac{\delta \Psi(1,\bm{\omega})}{\delta \alpha}\,\right)\,\right],
\end{equation}
and finally we get
\begin{equation}
\label{first_der}
\frac{d\Sigma(\Psi(\alpha),x)}{d\alpha}
=\mathbb{E}\left[\,\mathcal{E}\big(x \bm{r};1,\bm{\omega}\big)\partial_{\alpha} \Psi(\alpha;1,\bm{\omega})\,\right].
\end{equation}
It is easy to show that the above expression can be related to the following BSDE
\begin{equation}
m(\alpha;q,\bm{\omega})=\partial_{\alpha} \Psi(\alpha;1,\bm{\omega})-\int^1_q\big(\text{d}\bm{\omega}(q')-\text{d}q x(q') \bm{r}(\alpha;q',\bm{\omega})\big)\cdot \bm{z}(\alpha;q',\bm{\omega}).
\end{equation}
A classical result of BSDE theory \cite{PaPeng} states that the above equation has a unique solution $(\,m(\alpha),\bm{z}(\alpha)\,)$. Using the notation \ref{sec5.2}, we have
\begin{equation}
\label{equation_m}
m(\alpha;q,\bm{\omega})=\mathbb{E}\left[\,\mathcal{E}\big(x \bm{r};1,\bm{\omega}|x\big)\partial_{\alpha} \Psi(\alpha;1,\bm{\omega})\,\big|\mathcal{F}_q\right]
\end{equation}
that implies
\begin{equation}
\frac{d\Sigma(\Psi(\alpha),x)}{d\alpha}=m(\alpha;0,\bm{0}).
\end{equation}
By a similar computation, we also obtain
\begin{equation}
\label{derivative_phi}
\partial_\alpha\phi^{\alpha}(q,\bm{\omega})=m(\alpha;q,\bm{\omega}).
\end{equation}
It is worth noting that the same results can be achieved, in a non-rigorous way, by assuming that both the processes $\phi^{\alpha}$ and $\bm{r}^{\alpha}$ are differentiable with respect $\alpha$, for almost $q\in[0,1]$ and $\bm{\omega}\in \Omega$. So, taking the derivative on $\alpha$ of both the members of the BSDE \eqref{selfEq1Aux}, we get:
\begin{equation}
\partial_\alpha\phi^{\alpha}(q,\bm{\omega})=\partial_\alpha\Psi(\alpha;1,\bm{\omega})-\int^1_q\big(\text{d}\bm{\omega}(q')-\text{d}q x(q') \bm{r}^{\alpha}(q',\bm{\omega})\big)\cdot \partial_\alpha\bm{r}^{\alpha}(q',\bm{\omega}).
\end{equation}
then, we recover the equation \eqref{equation_m} by using \eqref{derivative_phi} and imposing $\partial_\alpha\bm{r}^{\alpha}=\bm{z}(\alpha)$.

We now provide an equivalent expression for the first derivative, that will be useful for the computation of higher derivative. The BSDE \eqref{selfEq1Aux} yields:
\begin{equation}
d\phi^{\alpha}(q,\bm{\omega})= \bm{r}(q,\bm{\omega})\cdot \text{d}\bm{\omega}(q)-\frac{x(q)}{2}\| \bm{r}(q,\bm{\omega})\|^2\,\text{d}q
\end{equation}
so, by the stochastic integration by part formula (Proposition 4.5 in \cite{YoRev}), we get

\begin{multline}
\label{identity_DDE}
\int^1_0 x(q')\bm{r}(q,\bm{\omega})\cdot \text{d}\bm{\omega}(q)-\frac{1}{2}\int^1_0\text{d}q\,x^2(q)\| \bm{r}(q,\bm{\omega})\|^2\\
=-x(q)\,\phi^{\alpha}(q,\bm{\omega})-\int^1_q \text{d}x(q')\,\phi^{\alpha}(q,\bm{\omega})+\phi^{\alpha}(1,\bm{\omega}).
\end{multline}
where $\int \text{d} x(q)\, \Cdot$ is the the Lebesgue-Stieltjes integral with respect to the POP $x$. We write
\begin{equation}
\label{P}
\text{d}x(q)=\text{d}q \,P(q).
\end{equation}
If the POP $x$ is absolutely continuous, then $P$ is a properly defined non-negative function. If it is not, for example when the POP $x$ is piecewise constant, we still use \eqref{P} as a symbolic notation; in this case $P$ is formally given by a combination of Dirac delta function (see Example 3 in Chapter V of \cite{Berry}). Dirac delta function (it is actually a generalized function) is largely used in physical literature, so we prefer consider this notation.

Substituting the above formula in \eqref{derivative_phi} and using the relation
\begin{equation}
e^{\Psi(\alpha;1,\bm{\omega})}\partial_{\alpha}\Psi(\alpha;1,\bm{\omega})=\partial_{\alpha}\left(\,e^{\Psi(\alpha;1,\bm{\omega})}\,\right),
\end{equation}
we get
\begin{equation}
\label{derivative_equivalent}
\partial_\alpha\phi^{\alpha}(q,\bm{\omega})=e^{-x(q)\,\phi^{\alpha}(q,\bm{\omega})}\mathbb{E}\left[e^{-\int^1_q \text{q}'P(q')\,\phi^{\alpha}(q',\bm{\omega})}\partial_{\alpha}\left(\,e^{\Psi(\alpha;1,\bm{\omega})}\,\right)\,\Bigg|\mathcal{F}_q\right].
\end{equation}
Note that the above expression does not depends on $\bm{r}^{\alpha}$.

\section{Higher Order Derivatives}
\label{sec6.2}
In this section, we present a method to compute the higher order derivatives. A necessary condition for the existence of the $k-$times derivative of the RSB expectation, for any $k\in\mathbb{N}$, is that, for each $\alpha\in I$, the claim $\Psi(\Cdot;1,\bm{\omega})$ is $k-$times differentiable with respect to $\alpha$ a.s.. We consider only the case where the derivative of the claim are bounded processes
\begin{equation}
\label{hypo_ders}
|\partial^k_\alpha \Psi(\alpha;1,\bm{\omega})|\leq C_k<\infty\,\quad \text{for all}\, \alpha\in I,\,\bm{\omega}\in \Omega,\,k\in \mathbb{N}
\end{equation}
for some $C_k$. The above hypothesis implies \eqref{hypo_1} and \eqref{hypo_2}. We impose that the above condition is verified \emph{surely}, i.e. for all $\bm{\omega}\in\Omega$, that is a stronger condition then almost sure conditions that are usually required with stochastic process. We can consider \emph{surely boundedness} because of the surely boundedness of the function that we consider in the physical problem.

In the following we shall use the following notation 
\begin{equation}
Z(1,\bm{\omega})=e^{\Psi(\alpha;1,\bm{\omega})},
\end{equation}
\begin{equation}
 A(q',\bm{\omega}|q)=\,e^{-\int^{q'}_q dP(q'')\,\phi^{\alpha}(q'',\bm{\omega})},\quad 0\leq q\leq q'\leq 1,
\end{equation}
and
\begin{equation}
B(q,\bm{\omega})=e^{x(q)\,\phi^{\alpha}(q,\bm{\omega})}.
\end{equation}
By \eqref{identity_DDE}, we have
\begin{equation}
\label{measure_new_notation}
\mathcal{E}\big(x \bm{r}^{\alpha};q',\bm{\omega}|q\big)=A(q',\bm{\omega}|q)\frac{B(q',\bm{\omega})}{B(q,\bm{\omega})}
\end{equation}
and the first derivative \eqref{derivative_equivalent} is:
\begin {equation}
\label{der_1}
\partial_\alpha\phi^{\alpha}(q,\bm{\omega})=\frac{1}{B(q,\bm{\omega})}\mathbb{E}\left[ A(1,\bm{\omega}|q)\partial_{\alpha}Z(1,\bm{\omega})\,\big|\mathcal{F}_q\right].
\end{equation}
In the first subsection, we provide some basic tools for the computation of the derivatives and obtain the second derivative and, finally, we obtain a recursion relation for the derivatives at all orders. Thank to these results, in the second subsection we derive the convergence criterion for a particular Taylor series expansion over the parameter $\alpha$. In the third subsection we discuss the ultrametric structure underling the RSB expectation, that is a consequence of the Parisi ansatz \cite{VPM}.
\subsection{Computation of derivatives}
\label{subsec6.2.1}
In this subsection we obtain a recursion formula that allows to compute the derivatives over $\alpha$ of the RSB expectation at all orders.

First of all, we compute the first derivative of a certain class of processes that includes also the derivatives of $\phi^{\alpha}$. Let $\pi^{\alpha}$ be any $\alpha-$dependent process of the form
\begin{multline}
\label{gamma}
\pi^{\alpha}(q,\bm{\omega})=\widetilde{\mathbb{E}}_{x\bm{r}^{\alpha}}\left[\int^1_q \text{d}q'\,\Pi(\alpha;q',\bm{\omega}|q)\,\Bigg|\mathcal{F}_q\right]\\
=\mathbb{E}\left[\int^1_q \text{d}q'\,\,A(q',\bm{\omega}|q)\frac{B(q',\bm{\omega}) }{B(q,\bm{\omega})}\,\,\Pi(\alpha;q',\bm{\omega}|q)\,\Bigg|\mathcal{F}_q\right],
\end{multline}
where $\widetilde{\mathbb{E}}_{x\bm{r}^{\alpha}}$ denotes the expectation with respect the probability measure defined in \eqref{DDE_expectation} and $\Pi(\alpha|q)$ is a generalized stochastic process (i.e. it may be given by a combination of Dirac delta function on $q'\in[0,1]$). Not that, if 
\begin{equation}
\Pi(\alpha;q',\bm{\omega}|q)=\frac{\partial_{\alpha}Z(1,\bm{\omega})}{B(1,\bm{\omega})} \delta(1-q')\,, 
\end{equation}
then $\pi^{\alpha}=\partial_{\alpha}\phi^{\alpha}$.

The derivative with respect to $\alpha$ is: 
\begin{equation}
\begin{multlined}
\partial_\alpha\pi^{\alpha}(q,\bm{\omega})=\mathbb{E}_{\mathbb{W}}\Bigg[\int^1_q \text{d}q'\,A(q',\bm{\omega}|q)\frac{B(q',\bm{\omega}) }{B(q,\bm{\omega})}\Pi(\alpha;q',\bm{\omega}|q)\\ \times\Bigg(\frac{\partial_{\alpha}\Pi(\alpha;q',\bm{\omega}|q)}{\Pi(\alpha;q',\bm{\omega}|q)}+\frac{\partial_\alpha B(q',\bm{\omega})}{B(q',\bm{\omega})}-\frac{\partial_\alpha B(q,\bm{\omega})}{B(q,\bm{\omega})}+\frac{\partial_\alpha A(q',\bm{\omega}|q)}{ A(q',\bm{\omega}|q)}\,\Bigg)\,\Bigg|\mathcal{F}_q\Bigg]
\\
=\mathbb{E}\left[\int^1_q \text{d}q'\,A(q',\bm{\omega}|q)\frac{B(q',\bm{\omega})}{B(q,\bm{\omega})}\,\,\partial_{\alpha}\Pi(\alpha;q',\bm{\omega}|q)\,\Bigg|\mathcal{F}_q\right]\\
+\,\mathbb{E}\left[\int^1_q \text{d}q'\,A(q',\bm{\omega}|q)\frac{B(q',\bm{\omega})}{B(q,\bm{\omega})}\,x(q') \phi(q',\bm{\omega})\Pi(\alpha;q',\bm{\omega}|q)\, \,\Bigg|\mathcal{F}_q\right]\\
-x(q) \phi(q,\bm{\omega})\,\,\mathbb{E}\left[\int^1_q \text{d}q'\,A(q',\bm{\omega}|q)\frac{B(q',\bm{\omega})}{B(q,\bm{\omega})}\,\,\Pi(\alpha;q',\bm{\omega}|q)\, \,\Bigg|\mathcal{F}_q\right]\\
-\mathbb{E}\left[\int^1_q \text{d}q'\,A(q',\bm{\omega}|q)\frac{B(q',\bm{\omega})}{B(q,\bm{\omega})}\Pi(\alpha;q',\bm{\omega}|q)\int^{q'}_qdP(q'') \partial_{\alpha}\phi(q'',\bm{\omega})\,\Bigg|\mathcal{F}_q\right],
\end{multlined}
\end{equation}
By \eqref{gamma}, we have
\begin{equation}
 x(q) \phi(q,\bm{\omega})\mathbb{E}\left[\int^1_q \text{d}q'\,A(q',\bm{\omega}|q)\frac{B(q',\bm{\omega})}{B(q,\bm{\omega})}\Pi(\alpha;q',\bm{\omega}|q)\,\Bigg|\mathcal{F}_q\right]=x(q) \phi(q,\bm{\omega})\pi^{\alpha}(q,\bm{\omega}),
\end{equation}
and, using the identity
\begin{equation}
A(q',\bm{\omega}|q)=A(q',\bm{\omega}|q'')A(q'',\bm{\omega}|q)\quad \forall q\leq q''\leq q',
\end{equation}
 and the tower property of the expectation value, it follows that
\begin{equation}
\begin{multlined}
\mathbb{E}\left[\int^1_q \text{d}q'\,A(q',\bm{\omega}|q)\frac{B(q',\bm{\omega})}{B(q,\bm{\omega})}\Pi(\alpha;q',\bm{\omega}|q)\int^{q'}_qdP(q'') \partial_{\alpha}\phi^{\alpha}(q'',\bm{\omega})\,\Bigg|\mathcal{F}_q\right]\\
=\mathbb{E}\left[\int^{1}_qdP(q')A(q',\bm{\omega}|q)\frac{B(q',\bm{\omega})}{B(q,\bm{\omega})} \partial_{\alpha}\phi^{\alpha}(q',\bm{\omega})\,\,\pi^{\alpha}(q,\bm{\omega})\,\Bigg|\mathcal{F}_q\right].
\end{multlined}
\end{equation}
Combining the above three relations, one finds
\begin{multline}
\label{tool_for_recursion}
\partial_\alpha\pi^{\alpha}(q,\bm{\omega})=\mathbb{E}\left[\int^1_q \text{d}q'\,A(q',\bm{\omega}|q)\frac{B(q',\bm{\omega}) }{B(q,\bm{\omega})}O(\Pi(\alpha|q),\pi^{\alpha},q;q',\bm{\omega})\,\Bigg|\mathcal{F}_q\right]\\
=\widetilde{\mathbb{E}}_{x\bm{r}^{\alpha}}\left[\int^1_q\text{d}q'O(\Pi(\alpha|q),\pi^{\alpha},q;q',\bm{\omega})\,\Bigg|\mathcal{F}_q\right]
\end{multline}
where
\begin{multline}
\label{map_for_der}
O(\Pi(\alpha|q),\pi^{\alpha},q;q',\bm{\omega})=
\partial_{\alpha}
\Pi(\alpha;q',\bm{\omega}|q)+x(q')\partial_{\alpha}\phi^{\alpha}(q',\bm{\omega})\,\,\Pi(\alpha;q',\bm{\omega}|q) \\- \delta(q'-q)x(q)\,\partial_{\alpha}\phi^{\alpha}(q,\bm{\omega})\,\,\pi^{\alpha}(q,\bm{\omega})-P(q')\partial_{\alpha}\phi^{\alpha}(q',\bm{\omega})\,\,\pi^{\alpha}(q',\bm{\omega}).
\end{multline}
It follows that the differential operator $\partial_{\alpha}$ maps any process of the form \eqref{gamma} to a process of the same form, inducing the transformation $\Pi(\alpha,q)\mapsto O(\Pi(\alpha,q),\pi^{\alpha},q)$ \eqref{map_for_der}. By induction, the derivatives of $\pi^{\alpha}$ are processes of the form \eqref{gamma}, obtained by recursively applying the operator $O$. 

Putting $\Pi(\alpha;q',\bm{\omega}|q)=(B(1,\bm{\omega}) )^{-1}\partial_{\alpha}Z(1,\bm{\omega}) \delta(1-q')$, then $\partial_\alpha\pi^{\alpha}=\partial^2_\alpha\phi^{\alpha}$. Using the above formulas, one finds
\begin{multline}
\label{second_derivative}
\partial^2_\alpha\phi^{\alpha}(q,\bm{\omega})=\mathbb{E}_{\mathbb{W}}\left[A(1,\bm{\omega}|q)\partial^2_{\alpha}Z(1,\bm{\omega})\,\big|\mathcal{F}_q\right]\\-\widetilde{\mathbb{E}}_{x\bm{r}^{\alpha}}\left[\int^1_q \text{d}q'\,P(q')\left(\,\partial_{\alpha}\phi^{\alpha}(q',\bm{\omega})\,\right)^2\,\Bigg|\mathcal{F}_q\right]
-x(q)\partial_{\alpha}\left(\,\phi^{\alpha}(q,\bm{\omega})\,\right)^2.
\end{multline}
The second derivative of the RSB expectation is given by the process $\partial^2_\alpha\phi^{\alpha}$ at $q=0$.
 
It is easy to prove that, if the process $\partial_{\alpha}^2\Psi(\alpha)$ is almost surely positive, then $\partial^2_\alpha\Sigma(\Psi(\alpha),x)$ is positive, according with the convexity criterion in Proposition \ref{Non:llinear_prop}.

Now, we use the formula \eqref{tool_for_recursion} and \eqref{map_for_der} to obtain a recursion relation amongst the derivatives of $\Sigma(\Psi(\alpha),x)$. For simplicity, in the rest of this subsection we will consider only the case where $\Psi(\alpha)$ has the form:
\begin{equation}
\label{case}
\Psi(\alpha;1,\bm{\omega})=\log \left(Z_0(1,\bm{\omega})+\alpha Z_1(1,\bm{\omega})\right),
\end{equation}
where $Z_0$ and $Z_1$ are two (surely) bounded processes such as $Z_0$ is strictly positive and $Z_0+\alpha Z_1$ is strictly positive for all $\alpha\in I$, i.e. there exist three positive constants $C_0$, $C_1$ and $C_3$ such as:
\begin{equation}
\label{case_1}
|Z_0(1,\bm{\omega})|<C_0,\quad |Z_1(1,\bm{\omega})|<C_1\quad \forall \bm{\omega}\in\Omega
\end{equation}
and
\begin{equation}
\label{case_2}
Z_0(1,\bm{\omega})+\alpha Z_1(1,\bm{\omega})>C_3\quad \forall \bm{\omega}\in\Omega\,\,\text{and}\,\,\alpha\in I\,.
\end{equation}
Note that $\Psi(\alpha)$ verifies the hypothesis \eqref{hypo_ders}. A trivial computation yields
\begin{equation}
\label{derivative_claim}
\partial_{\alpha}Z(1,\bm{\omega})= Z_1(q,\bm{\omega}),\quad \partial^k_{\alpha}Z(1,\bm{\omega})=0 \quad\forall k\geq 2.
\end{equation}

Given a sequence of almost surely bounded processes $(\Sigma_k)$, parametrized by the integer index $k\geq 1$, let $\Lambda_{k,n}(\Sigma)$, with $k\geq n\geq 1$, be the process defined through the recursion
\begin{equation}
\label{Lambda}
\begin{gathered}
\Lambda _{n,n}(\Sigma;q,\bm{\omega})=\left(\Sigma_{1}(q,\bm{\omega})\right)^n\,,\\
\Lambda _{k,1}(\Sigma;q,\bm{\omega})=\Sigma_{k}(q,\bm{\omega})\,,\\
\Lambda _{k,n}(\Sigma;q,\bm{\omega})=\sum^{k-1}_{m=n-1}\Sigma _{k-m}(q,\bm{\omega})\Lambda _{m,n-1}(\Sigma;q,\bm{\omega}).
\end{gathered}
\end{equation}
It is worth remarking that the process $\Lambda _{k,n}(\Sigma)$ depends only on the processes of the truncated sequence $(\Sigma_p; \,p\leq k+1-n)$. Note also that
\begin{equation}
\label{Lambda_plus_1}
\Lambda _{n+1,n}(\Sigma;q,\bm{\omega})=n\left(\Sigma_{1}(q,\bm{\omega})\right)^{n-1}\Sigma_{2}(q,\bm{\omega})\,.
\end{equation}
Now we have the tools to obtain the main result of this subsection.
\begin{theorem}
\label{recursion_derivative}
Consider a claims $\Psi(\alpha)$ of the form \eqref{case} and verifying \eqref{case_1} and \eqref{case_2}. 
Let $(\Sigma_k)$ be the sequence of processes satisfying the following recursion:
\begin{equation}
\label{recursion_Sigma_k}
\begin{gathered}
\Sigma_1(q,\bm{\omega})=\widetilde{\mathbb{E}}_{x\bm{r}^{\alpha}}\left[\partial_{\alpha}\Psi(\alpha;1,\bm{\omega})\,\big|\mathcal{F}_q\right]\,,\\
\Sigma_k(q,\bm{\omega})=\widetilde{\mathbb{E}}_{x\bm{r}^{\alpha}}\left[\int^1_q\text{d}q'' \,P(q'')\sum^k_{n=2}\frac{x^{n-2}(q)}{n(n-1)}\Lambda _{k,n}(\Sigma;q'',\bm{\omega})\,\Bigg|\mathcal{F}_q\right]\quad \text{for}\,\,k>1\,.
\end{gathered}
\end{equation}
Then
\begin{equation}
\partial^k_\alpha\Sigma(\Psi(\alpha),x)=(-1)^{k+1} \,k! \,\Sigma_k(0,\bm{0})
\end{equation}
\end{theorem}
 Note that $\Sigma_1=\partial_{\alpha} \phi$. For the moment, we do not study the convergence of the processes $\Sigma_k$ as $k\to \infty$, and consider only the derivatives up to a finite order $K<\infty$. In such a way, the solution \eqref{recursion_Sigma_k} exists and can be obtain recursively, since the right-hand member is an bounded process depending only on $\Sigma_p$ with $p\leq k-1$, for any $k>1$. The conditions \eqref{case}, \eqref{case_1} and \eqref{case_2} implies that for any finite integer $k$, the process $\Sigma_k$ is bounded by a constant that depends on $k$.

The recursion equation \eqref{recursion_Sigma_k} is obtained by a proper reshuffling of the relation \eqref{tool_for_recursion}, in order to remove the differential operator $\partial_{\alpha}$ and the Dirac delta function that appear in the operator $O$ in \eqref{map_for_der}. The proof of Theorem \ref{recursion_derivative} relies on the following lemma.
\begin{lemma}
\label{Lemma_lambda}
Let $(\sigma_k;\, k\leq K)$ be any sequence of $\alpha-$dependent bounded process (we omit the dependence on $\alpha$), such as:
\begin{equation}
\label{condition_Lemma}
\partial_{\alpha}\sigma_k(q,\bm{\omega})=-(k+1)\sigma_{k+1}(q,\bm{\omega})-x(q)\sigma_1(q,\bm{\omega})\sigma_k(q,\bm{\omega})\quad a.s.
\end{equation}
and let $\Lambda _{k,n}(\sigma;q,\bm{\omega})$ be the process obtained from\eqref{Lambda} by substituting $\Sigma$ with $\sigma$. Then
\begin{multline}
\label{derivative_Lambda}
\partial_\alpha \Lambda _{k,n}(\sigma;q,\bm{\omega})\\=-(k+1) \Lambda _{k+1,n}(\sigma;q,\bm{\omega})+n\,\sigma _{1}(q,\bm{\omega})\Lambda_{k,n-1}(\sigma;q,\bm{\omega})-n\,x(q)\sigma_{1}(q,\bm{\omega})\Lambda_{k,n}(\sigma;q,\bm{\omega})
\end{multline}
\end{lemma}
\begin{proof}
The proof is by induction over the first index. The induction starts with $k=n$:
\begin{equation}
\begin{multlined}
\partial_{\alpha}\Lambda _{n,n}(\sigma;q,\bm{\omega})=\partial_{\alpha}\left(\sigma_{1}(q,\bm{\omega})\right)^n=n \left(\sigma_{1}(q,\bm{\omega})\right)^{n-1}\partial_{\alpha}\sigma_{1}(q,\bm{\omega})\\
=-2 n \left(\sigma_{1}(q,\bm{\omega})\right)^{n-1}\sigma_{2}(q,\bm{\omega})-n\,x(q) \left(\sigma_{1}(q,\bm{\omega})\right)^{n+1}.
\end{multlined}
\end{equation}
Now, we write
\begin{equation}
\begin{multlined}
-2 n \left(\sigma_{1}(q,\bm{\omega})\right)^{n-1}\sigma_{2}(q,\bm{\omega})\\=-(n+1) \left(n \,\left(\sigma_{1}(q,\bm{\omega})\right)^{n-1}\sigma_{2}(q,\bm{\omega})\right)+n\,\sigma_{1}(q,\bm{\omega})\left((n-1) \,\left(\sigma_{1}(q,\bm{\omega})\right)^{n-2}\sigma_{2}(q,\bm{\omega})\right)\,,
\end{multlined}
\end{equation}
so, using \eqref{Lambda_plus_1}, we obtain that 
\begin{multline}
\partial_\alpha \Lambda _{n,n}(\sigma;q,\bm{\omega})\\=-(n+1) \Lambda _{n+1,n}(\sigma;q,\bm{\omega})+n\sigma _{1}(q,\bm{\omega})\Lambda _{n,n-1}(\sigma;q,\bm{\omega})-n\,x(q) \Lambda _{n+1,n+1}(\sigma;q,\bm{\omega})\,.
\end{multline} 
that proves \eqref{derivative_Lambda} for this case.

Now, we prove that if the relation \eqref{derivative_Lambda} holds for all $\Lambda _{m,n}(\sigma;q,\bm{\omega})$ with $m\leq k-1$, than it holds for $\Lambda _{k,n}(\sigma;q,\bm{\omega})$. By the recursion \eqref{Lambda}, it follows that
\begin{multline}
\label{induction_Lemma_derivative}
\partial_{\alpha}\Lambda _{k,n}(\Sigma;q,\bm{\omega})\\
=\sum^{k-1}_{m=n-1}\partial_{\alpha}\sigma _{k-m}(q,\bm{\omega})\Lambda _{m,n-1}(\sigma;q,\bm{\omega})+\sum^{k-1}_{m=n-1}\sigma _{k-m}(q,\bm{\omega})\partial_{\alpha}\Lambda _{m,n-1}(\sigma;q,\bm{\omega})
\end{multline}
By the condition \eqref{condition_Lemma}, the first summation on the right-hand side gives
\begin{equation}
\sum^{k-1}_{m=n-1}\partial_{\alpha}\sigma _{k-m}(q,\bm{\omega})\Lambda _{m,n-1}(\sigma;q,\bm{\omega})=-(k+1)\text{I}(q,\bm{\omega})+\text{II}(q,\bm{\omega})-\text{III}(q,\bm{\omega})
\end{equation}
where
\begin{equation}
\begin{multlined}
\text{I}(q,\bm{\omega})=\sum^{k-1}_{m=n-1}\sigma _{k+1-m}(q,\bm{\omega})\Lambda _{m,n-1}(\sigma;q,\bm{\omega})\\=\left(\sum^{k}_{m=n-1}\sigma _{(k+1)-m}(q,\bm{\omega})\Lambda _{m,n-1}(\sigma;q,\bm{\omega})\right)-\sigma _{1}(q,\bm{\omega})\Lambda _{k,n-1}(\sigma;q,\bm{\omega})\\
=\Lambda _{k+1,n}(\sigma;q,\bm{\omega})-\sigma _{1}(q,\bm{\omega})\Lambda _{k,n-1}(\sigma;q,\bm{\omega})\,,
\end{multlined}
\end{equation}
\begin{equation}
\text{II}(q,\bm{\omega})=\sum^{k-1}_{m=n-1}m\,\sigma _{k+1-m}(q,\bm{\omega})\Lambda _{m,n-1}(\sigma;q,\bm{\omega}))
\end{equation}
and
\begin{equation}
\begin{multlined}
\text{III}(q,\bm{\omega})\\=\sum^{k-1}_{m=n-1}x(q) \sigma_1(q,\bm{\omega})\sigma _{k-m}(q,\bm{\omega})\Lambda _{m,n-1}(\sigma;q,\bm{\omega})=x(q) \sigma_1(q,\bm{\omega})\Lambda _{k,n}(\sigma;q,\bm{\omega}).
\end{multlined}
\end{equation}
By induction hypothesis the second summation in \eqref{induction_Lemma_derivative} gives 
\begin{equation}
\sum^{k-1}_{m=n-1}\sigma _{k-m}(q,\bm{\omega})\partial_{\alpha}\Lambda _{m,n-1}(\sigma;q,\bm{\omega})=-\text{IV}(q,\bm{\omega})+(n-1)\text{V}(q,\bm{\omega})-(n-1)\text{VI}(q,\bm{\omega})
\end{equation}
where
\begin{equation}
\begin{multlined}
\text{IV}(q,\bm{\omega})\\=\sum^{k-1}_{m=n-1}(m+1)\,\sigma _{k-m}(q,\bm{\omega})\Lambda _{m+1,n-1}(\sigma;q,\bm{\omega})=\sum^{k}_{m=n}m\,\sigma _{(k+1)-m}(q,\bm{\omega})\Lambda _{m,n-1}(\sigma;q,\bm{\omega})\\=\text{II}(q,\bm{\omega})-(n-1)\sigma _{k+2-n}(q,\bm{\omega})\Lambda _{n-1,n-1}(\sigma;q,\bm{\omega})+k\,\sigma _{1}(q,\bm{\omega})\Lambda _{k,n-1}(\sigma;q,\bm{\omega})\,,
\end{multlined}
\end{equation}
\begin{equation}
\begin{multlined}
\text{V}(q,\bm{\omega})=\sum^{k-1}_{m=n-1}\sigma _{k-m}(q,\bm{\omega})\sigma _{1}(q,\bm{\omega})\Lambda _{m,n-2}(\sigma;q,\bm{\omega})\\
=\sigma _{1}(q,\bm{\omega})\left(\,\left(\sum^{k-1}_{m=n-2}\sigma _{k-m}(q,\bm{\omega})\Lambda _{m,n-2}(\sigma;q,\bm{\omega})\right)-\sigma _{k-(n-2)}(q,\bm{\omega})\Lambda _{n-2,n-2}(\sigma;q,\bm{\omega})\,\right)\\
=\sigma _{1}(q,\bm{\omega})\Lambda _{k,n-1}(\sigma;q,\bm{\omega})-\sigma _{k+2-n}(q,\bm{\omega})\Lambda _{n-1,n-1}(\sigma;q,\bm{\omega})\,.
\end{multlined}
\end{equation}
and
\begin{equation}
\begin{multlined}
\text{VI}(q,\bm{\omega})=\sum^{k-1}_{m=n-1}x(q) \sigma_1(q,\bm{\omega})\sigma _{k-m}(q,\bm{\omega})\Lambda _{m,n-1}(\sigma;q,\bm{\omega})=\text{III}(q,\bm{\omega}).
\end{multlined}
\end{equation}
Combining the above six terms in \eqref{induction_Lemma_derivative}, one gets
\begin{equation}
\begin{multlined}
\partial_\alpha \Lambda _{k,n}(\sigma;q,\bm{\omega})\\=-(k+1) \Lambda _{k+1,n}(\sigma;q,\bm{\omega})+n\sigma _{1}(q,\bm{\omega})\Lambda _{k,n-1}(\sigma;q,\bm{\omega})-n\,x(q)\,\sigma _{1}(q,\bm{\omega})\Lambda _{k,n}(\sigma;q,\bm{\omega})\,,
\end{multlined}
\end{equation}
that proves the Lemma.
\end{proof}
\begin{proof}[Proof of Theorem \ref{recursion_derivative}]
We prove that the sequence $(\Sigma_k)$ that solves the recursion \eqref{recursion_Sigma_k}, satisfies the condition \eqref{condition_Lemma} of the preceding lemma.

We proceed by induction on $k$, starting from $k=2$. Since $\Sigma_1=\partial_{\alpha} \phi$, then, by \eqref{second_derivative} and \eqref{derivative_claim}, it follows that
\begin{equation}
\partial_{\alpha}\Sigma_1(q,\bm{\omega})=-2\Sigma_2(q,\bm{\omega})-x(q) \left(\Sigma_1(q,\bm{\omega})\right)^2\,.
\end{equation}
Now we assume by induction hypothesis that
\begin{equation}
\label{induction_hypo_ther_der}
\partial_{\alpha}\Sigma_p(q,\bm{\omega})=-(p+1)\Sigma_{p+1}(q,\bm{\omega})-x(q)\Sigma_1(q,\bm{\omega})\Sigma_p(q,\bm{\omega})\quad a.s.\,\,\forall p< k
\end{equation}
and then, we compute
\begin{equation}
\partial_{\alpha}\Sigma_{k}(q,\bm{\omega}).
\end{equation}
The equation \eqref{recursion_Sigma_k} implies that the process $\Sigma_{k+1}(q,\bm{\omega})$ is of the form \eqref{gamma}, so, by equation \eqref{tool_for_recursion}, it follows that
\begin{equation}
\partial_\alpha\Sigma_{k}(q,\bm{\omega})=\widetilde{\mathbb{E}}_{x\bm{r}^{\alpha}}\left[\int^1_q\text{d}q'\,\,O(\Pi(\alpha|q),\Sigma_{k},q;q',\bm{\omega})\,\Bigg|\mathcal{F}_q\right]
\end{equation}
where
\begin{equation}
\Pi(\alpha;q',\bm{\omega}|q)=P(q')\sum^k_{n=2}\frac{x^{n-2}(q)}{n(n-1)}\Lambda _{k,n}(\Sigma;q',\bm{\omega}).
\end{equation}
In particular, we have
\begin{multline}
\label{theo_der_fin1}
\partial_{\alpha}\Sigma_k(q,\bm{\omega})+x(q)\Sigma_1(q,\bm{\omega})\Sigma_k(q,\bm{\omega})\\=\widetilde{\mathbb{E}}_{x\bm{r}^{\alpha}}\left[\int^1_q\text{d}q'' \,P(q'')\left(\sum^k_{n=2}\text{I}_{k,n}(q',\bm{\omega})+\sum^{k+1}_{n=3}\text{II}_{k,n}(q',\bm{\omega})+\text{III}_k(q',\bm{\omega})\right)\Bigg|\mathcal{F}_q\right]
\end{multline}
where
\begin{equation}
\text{I}_{k,n}(q,\bm{\omega})=\frac{x^{n-2}(q)}{n(n-1)}\partial_{\alpha}\Lambda _{k,n}(\Sigma;q,\bm{\omega}),
\end{equation}
\begin{equation}
\begin{multlined}
\text{II}_{k,n}(q,\bm{\omega})=\frac{x^{n-2}(q)}{(n-1)(n-2)}\Sigma_1(q,\bm{\omega})\Lambda _{k,n-1}(\Sigma;q,\bm{\omega})
\end{multlined}
\end{equation}
and
\begin{equation}
\begin{multlined}
\text{III}_k(q,\bm{\omega})=-\partial_{\alpha} \phi(q,\bm{\omega}) \Sigma_k(q,\bm{\omega})=- \Sigma_1(q,\bm{\omega}) \Sigma_k(q,\bm{\omega})=-\frac{1}{2}\Lambda _{k+1,2}(\Sigma;q,\bm{\omega}).
\end{multlined}
\end{equation}
Since $\Lambda _{k,n}(\Sigma)$ depends only on the processes $\Sigma_p$ with $p\leq k+1-n$, then, under the induction hypothesis \eqref{induction_hypo_ther_der}, the computation of the processes $\text{I}_n$ can be performed by using Lemma \ref{Lemma_lambda} and so we get:
\begin{equation}
\begin{multlined}
\text{I}_{k,n}(q,\bm{\omega})=-(k+1)\frac{x^{n-2}(q)}{n(n-1)}\Lambda _{k+1,n}(\Sigma;q,\bm{\omega})\\+\frac{x^{n-2}(q)}{n-1}\Sigma_1(q,\bm{\omega})\Lambda _{k,n-1}(\Sigma;q,\bm{\omega})-\frac{x^{n-1}(q)}{n-1}\Sigma_1(q,\bm{\omega})\Lambda _{k,n}(\Sigma;q,\bm{\omega})
\end{multlined}
\end{equation}
and
\begin{equation}
\begin{multlined}
\sum^k_{n=2}\text{I}_{k,n}(q',\bm{\omega})\\=-(k+1)\sum^k_{n=2}\frac{x^{n-2}(q)}{n(n-1)}\Lambda _{k+1,n}(\Sigma;q,\bm{\omega})-\text{III}_{k}(q,\bm{\omega})-(k+2)\text{II}_{k,k+1}(q,\bm{\omega})-\sum^{k}_{n=3}\text{II}_{k,n}(q,\bm{\omega})\\=
-(k+1)\sum^{k+1}_{n=2}\frac{x^{n-2}(q)}{n(n-1)}\Lambda _{k+1,n}(\Sigma;q,\bm{\omega})-\text{III}_{k}(q,\bm{\omega})-\sum^{k+1}_{n=3}\text{II}_{k,n}(q,\bm{\omega})\\
\end{multlined}
\end{equation}
Combining all the terms,we get:
\begin{equation}
\begin{multlined}
\partial_{\alpha}\Sigma_k(q,\bm{\omega})+x(q)\Sigma_1(q,\bm{\omega})\Sigma_k(q,\bm{\omega})\\=-(k+1)\widetilde{\mathbb{E}}_{x\bm{r}^{\alpha}}\left[\int^1_q\text{d}q'\,P(q')\sum^{k+1}_{n=2}\frac{x^{n-2}(q')}{n(n-1)}\Lambda _{k+1,n}(\Sigma;q',\bm{\omega})\Bigg|\mathcal{F}_q\right]
\end{multlined}
\end{equation}
so 
\begin{equation}
\label{recursion_tool:_to_Prove}
\partial_{\alpha}\Sigma_k(q,\bm{\omega})=-(k+1)\Sigma_{k+1}(q,\bm{\omega})-x(q)\Sigma_1(q,\bm{\omega})\Sigma_k(q,\bm{\omega})
\end{equation}
Then, by induction, all the processes of the sequence $(\Sigma_k;\,k\leq K)$ verify the condition \eqref{recursion_tool:_to_Prove}. In partiular, since $x(0)=0$, it implies:
\begin{equation}
\partial_{\alpha}\Sigma_k(0,\bm{0})=-(k+1)\Sigma_{k+1}(q,\bm{\omega})\,,\quad\forall\,1\leq k< K
\end{equation}
then, since $\Sigma_1(0,\bm{0})=\partial_\alpha\Sigma(\Psi(\alpha),x)$, so
\begin{equation}
\Sigma_k(0,\bm{0})=\frac{(-1)^{k+1}}{k!}\partial^k_\alpha\Sigma(\Psi(\alpha),x)
\end{equation}
that concludes the proof.
\end{proof}
We stress that the results of this subsection, and of the whole chapter, are not referred to the actual physical problem, but they are general properties of the Parisi-Mézard ansatz. The recursive equation \eqref{recursion_Sigma_k} is a remarkable source of insight. In particular it will be play a crucial role in the computation of the series expansion, that we will present in the next subsection.
\subsection{Series expansion}
\label{subsec6.2.2}
In this subsection, we consider the series expansion of the RSB$-$ expectation with respect a parameter and provides a convergence criterion. The ultrametric structure of the Parisi-Mézard ansatz emerges from the coefficients of the series.

We consider the case where the $\Psi(\alpha)$ depends on the parameter $\alpha$ in such a way:
\begin{equation}
\label{prototype}
\Psi(\alpha;1,\bm{\omega})=\log\left(1+\alpha Z_1(1,\bm{\omega})\right)
\end{equation}
where $Z_1$ is a bounded random variable such as
\begin{equation}
\label{bound_cond}
\sup_{\bm{\omega}\in\Omega} |Z_1(1,\bm{\omega})|<1-\epsilon.
\end{equation}
for some $0<\epsilon<1$. In such a way the function $\Psi(\cdot;1,\bm{\omega})$ is an analytic function in $\alpha$ with $|\alpha|\leq 1$, for all $\bm{\omega}\in \Omega$.

Note that, for any strictly positive and bounded random variable $Z<C$ (for some constant $C$), we can define:
\begin{equation}
\label{Renorm_1}
\widetilde{\Psi}(\alpha;1,\bm{\omega})=\log\left(1+\alpha\frac{Z(1,\bm{\omega})-C'}{C'}\right)
\end{equation}
for any constant $C'>C$, so that
\begin{equation}
\label{Renorm_2}
\Psi(1,\bm{\omega})=\log(\,Z(1,\bm{\omega})\,)=\widetilde{\Psi}(1;1,\bm{\omega})+\log(\,C'\,)
\end{equation}
The effective parametric claim $\widetilde{\Psi}(\alpha)$ verifies the conditions \eqref{prototype} and \eqref{bound_cond}. Moreover, by Proposition \ref{Non:llinear_prop}, the constant $\log(\,C'\,)$ can be pulled outside the RSB expectation.

If $\alpha=0$ the claim \eqref{bound_cond} vanishes, so by Proposition \ref{Non:llinear_prop}, $(\phi^{0},\bm{r}^{0})=(0,0)$ a.s.. In particular, this implies that 
\begin{equation}
\label{start_expansion}
\mathcal{E}(x \bm{r}^0;q',\bm{\omega}|q)=1.
\end{equation}
The coefficient of the series expansion of $\Sigma(\Psi(\alpha),x)$ around $\alpha=0$ are simply given by solving the recursive equations \eqref{recursion_Sigma_k}, with $\alpha=0$. Combining the relation \eqref{start_expansion} with \eqref{recursion_Sigma_k}, we obtain:
\begin{theorem}[Series expansion]
\label{series_exp_theo}
Given a claim of the form \eqref{prototype}, then for any integer $K<\infty$
\begin{equation}
\label{series_expansion}
\Sigma(\Psi(\alpha),x)= \alpha\Sigma_1(0,\bm{0})-\alpha^2\Sigma_2(0,\bm{0})+\cdots+(-1)^{K+1}\alpha^K\Sigma_K(0,\bm{0})+o\big(\alpha^{K+1}\big)
\end{equation}
where the sequence $(\Sigma_k)$ is the solution of
\begin{equation}
\label{recursion_Sigma_k_expansion}
\begin{gathered}
\Sigma_1(q,\bm{\omega})=\mathbb{E}_{\mathbb{W}}\left[Z_1(1,\bm{\omega})\,\big|\mathcal{F}_q\right]\,,\\
\Sigma_k(q,\bm{\omega})=\mathbb{E}_{\mathbb{W}}\left[\int^1_q\text{d}q'' \,P(q'')\sum^k_{n=2}\frac{x^{n-2}(q)}{n(n-1)}\Lambda _{k,n}(\Sigma;q'',\bm{\omega})\,\Bigg|\mathcal{F}_q\right]\quad \text{for}\,\,\,k>1\,.
\end{gathered}
\end{equation}
\end{theorem}
A remarkable fact of the above theorem is that, unlike the recursion equations \eqref{recursion_Sigma_k} for the derivatives, the recursion \eqref{recursion_Sigma_k_expansion} involves the evaluation of conditional expectations with respect the Wiener measure $\mathbb{W}$, then it is not needed solving the BSDE \eqref{selfEq1Aux} for the computation of the coefficients $\Sigma_k(q,\bm{\omega})$. Here, we write explicitly the first four orders:
\begin{gather}
\Sigma_1(0,\bm{0})=\mathbb{E}_{\mathbb{W}}\left[Z_1(1,\bm{\omega} ) \, \right]\,,\\
\Sigma_2(0,\bm{0})=\frac{1}{2}\mathbb{E}_{\mathbb{W}}\left[\int_0^1\text{d}q_0\, P\left(q_0\right) \mathbb{E}_{\mathbb{W}}\left[Z_1(1,\bm{\omega} )|\mathcal{F}_{q_0}\right]^2 \, \right]\,,\\
\begin{multlined}
\Sigma_3(0,\bm{0})=\frac{1}{6}\mathbb{E}_{\mathbb{W}}\left[\int_0^1\text{d}q_1 P\left(q_1\right) x\left(q_1\right) \mathbb{E}_{\mathbb{W}}\left[Z_1(1,\bm{\omega} )|\mathcal{F}_{q_1}\right]{}^3 \, \right]\\+\frac{1}{2}\mathbb{E}_{\mathbb{W}}\Bigg[\int_0^1 \text{d}q_1\,P\left(q_1\right) \mathbb{E}_{\mathbb{W}}\left[Z_1(1,\bm{\omega})|\mathcal{F}_{q_1}\right]\int_{q_1}^1 \text{d}q_0\, P\left(q_0\right) \mathbb{E}_{\mathbb{W}}\left[Z_1(1,\bm{\omega} )|\mathcal{F}_{q_0}\right]{}^2 \, \,\Bigg]\,,
\end{multlined}
\end{gather}
and
\begin{multline}
\Sigma_4(0,\bm{0})=\frac{1}{2}\mathbb{E}_{\mathbb{W}}\Bigg[\int _0^1\text{d}q_2\,P\left(q_2\right)\mathbb{E}_{\mathbb{W}}\left[Z_1(1,\bm{\omega} )|\mathcal{F}_{q_2}\right]\\\times\int_{q_2}^1 \text{d}q_1\,P\left(q_1\right) \mathbb{E}_{\mathbb{W}}\left[Z_1(1,\bm{\omega} )|\mathcal{F}_{q_1}\right] \int_{q_1}^1 \text{d}q_0P\left(q_0\right) \mathbb{E}_{\mathbb{W}}\left[Z_1(1,\bm{\omega})|\mathcal{F}_{q_0}\right]^2 \Bigg]\\
+\frac{1}{8}\mathbb{E}_{\mathbb{W}}\left[\int_0^1\text{d}q_2 P\left(q_2\right) \mathbb{E}_{\mathbb{W}}\left[\int_{q_2}^1 \text{d}q_0 P\left(q_0\right) \mathbb{E}_{\mathbb{W}}\left[Z_1(1,\bm{\omega} )|\mathcal{F}_{q_0}\right]{}^2 \, \Bigg|\mathcal{F}_{q_2}\right]^2 \,\right]\\
+\frac{1}{6}\mathbb{E}_{\mathbb{W}}\Bigg[\int _0^1\text{d}q_2\,P\left(q_2\right)\mathbb{E}_{\mathbb{W}}\left[Z_1(1,\bm{\omega} )|\mathcal{F}_{q_2}\right]\int_{q_2}^1 \text{d}q_1\,P\left(q_1\right) x\left(q_1\right) \mathbb{E}_{\mathbb{W}}\left[Z_1(1,\bm{\omega} )|\mathcal{F}_{q_1}\right]^3 \Bigg]\\
\frac{1}{4}\mathbb{E}_{\mathbb{W}}\Bigg[\int_0^1 \text{d}q_2 P\left(q_2\right) x\left(q_2\right) \mathbb{E}_{\mathbb{W}}\left[Z_1(1,\bm{\omega} )|\mathcal{F}_{q_2}\right]^2 \int_{q_2}^1\text{d}q_0\,P\left(q_0\right) \mathbb{E}_{\mathbb{W}}\left[Z_1(1,\bm{\omega})|\mathcal{F}_{q_0}\right]^2 \, \Bigg]\\
+\frac{1}{12}\mathbb{E}_{\mathbb{W}}\left[\int_0^1 \text{d}q_2 P\left(q_2\right) x\left(q_2\right)^2 \mathbb{E}_{\mathbb{W}}\left[Z_1(1,\bm{\omega} )|\mathcal{F}_{q_2}\right]^4 \, \right]\,.
\end{multline}

Theorem \eqref{series_exp_theo}, together with the the next result, provide an important tool to the computation of the RSB expectation.
\begin{theorem}
Under the hypothesis \eqref{bound_cond} the series expansion \eqref{series_expansion} converges uniformly for $|\alpha| \leq 1$ and
\begin{equation}
\label{bound_coeff}
|\Sigma_k(q,\bm{0})|\leq \frac{(1-\epsilon)^k}{k},\quad a.s.\,\forall k\geq 1\,.
\end{equation}
Then
\begin{equation}
\label{remainder}
\sup_{\alpha\in [-1,1]}\left|\Sigma(\Psi(\alpha),x)- \left(\,\alpha\Sigma_1(0,\bm{0})+\cdots+(-1)^{K+1}\alpha^K\Sigma_K(0,\bm{0})\right)\right|\leq \frac{(1-\epsilon)^{K+1}}{K+1},\,\forall k\geq 1.
\end{equation}
\end{theorem}
\begin{proof}
We start by proving \eqref{bound_coeff}. Given two bounded random variables $Z_1$ and $Z_2$ that verifies the condition \eqref{bound_cond} and such that $Z_2$ is strictly positive, let us denote $(\Sigma_k(Z_1))$ and $(\Sigma_k(Z_2))$ the solutions of the corresponding equations \eqref{recursion_Sigma_k_expansion} (assuming $h<K<\infty$). We prove that, if
\begin{equation}
|Z_1(1,\bm{\omega})|\leq Z_2(1,\bm{\omega}) \,\,a.s.
\end{equation}
then, for all $q\in[0,1]$ 
\begin{equation}
\label{trial_prop}
|\Sigma_p(Z_1;q,\bm{\omega})|\leq\Sigma_p(Z_2;q,\bm{\omega})\,\,a.s.\,,\,\,\forall p\geq 1 
\end{equation}
 The above relation holds for $k=1$ as a consequence of the monotonicity of the conditional expectation
\begin{equation}
|\Sigma_1(Z_1,q,\bm{\omega})|\leq \mathbb{E}_{\mathbb{W}}\left[|Z_1(1,\bm{\omega})|\,|\mathcal{F}_q\right] \leq\mathbb{E}_{\mathbb{W}}\left[Z_2(1,\bm{\omega})\,|\mathcal{F}_q\right]\leq \Sigma_1(Z_2;0,\bm{0})\,\,a.s.
\end{equation}
By induction hypothesis, let us assume the relation \eqref{trial_prop} holds for all $p\leq k$. A trivial computation yields:
\begin{equation}
|\Lambda _{p,n}(\Sigma(Z_1);q,\bm{\omega})|\leq \Lambda _{p,n}(\Sigma(Z_2);q,\bm{\omega})\,\,a.s.,\forall 2\leq n\leq p\leq k+1
\end{equation}
so
\begin{equation}
\begin{multlined}
|\Sigma_{k+1}(Z_1;q,\bm{\omega})|\leq \mathbb{E}_{\mathbb{W}}\left[\int^1_q\text{d}q' \,P(q')\sum^{k+1}_{n=2}\frac{x^{n-2}(q)}{n(n-1)}|\Lambda _{k,n}(\Sigma(Z_1);q',\bm{\omega})|\,\Bigg|\mathcal{F}_q\right]\\
\leq \mathbb{E}_{\mathbb{W}}\left[\int^1_q\text{d}q' \,P(q')\sum^{k+1}_{n=2}\frac{x^{n-2}(q)}{n(n-1)}\Lambda _{k,n}(\Sigma(Z_2);q',\bm{\omega})\,\Bigg|\mathcal{F}_q\right]=\Sigma_{k+1}(Z_2;q,\bm{\omega}) \,\,a.s.
\end{multlined}
\end{equation}
that proves the relation \eqref{trial_prop}.

In particular, if $Z_2(1,\bm{\omega})=1-\epsilon$, by Proposition \ref{Non:llinear_prop} we have
\begin{equation}
\Sigma(\,\log(1+\alpha(1-\epsilon)\,),x)=\log(1+\alpha(1-\epsilon)\,),
\end{equation}
that is analytic for $|\alpha|\leq 1$. Then, by theorem \ref{series_exp_theo}
\begin{equation}
\begin{multlined}
\Sigma_k(Z_1,x)=\frac{(-1)^k}{k!}\partial^k_{\alpha}\Sigma(\log(1+\alpha(1-\epsilon)\,),x)|_{\alpha=0}\\=\frac{(-1)^k}{k!}\partial^k_{\alpha}\log(1+\alpha(1-\epsilon)\,),x)|_{\alpha=0}=\frac{(1-\epsilon)^k}{k}\,.
\end{multlined}
\end{equation}
The bound \eqref{remainder} is given by the Lagrange Remainder Theorem.
\end{proof}

By the above theorem, the RSB expectation can be approximated arbitrarily well by a proper truncation of the series expansion \eqref{series_exp_theo} and the error exponentially decreases with the order of the truncation. 

Note that, using the manipulations \eqref{Renorm_1} and \eqref{Renorm_2}, the above results can be extended to all the claims that are surely bounded by a constant. We guess that all the theorems of this chapter can be extended to a wider class of unbounded claims. The proof is beyond the aim of this thesis.

\chapter{The physical variational problem}
\label{C7}
\thispagestyle{empty}
In this chapter, we derive the self-consistency equations for the order parameter. The tools developed in Chapter \ref{C6} will be crucial
\section{The free energy functional}
The true variational free energy function $\Phi$ is derived from \eqref{freeFull1}, by substituting the solution of the equation \eqref{selfEq1Aux} of the edge and vertex contributions.

 In such a way, we get a variational functional depending explicitly on the cavity field functional $h$.

The cavity functional $h$ and the POP $x$ encode the cavity fields distributions inside the pure states entirely, as we discussed in Chapter \ref{C4}, so it constitutes the full$-$RSB order parameters. 

Using the same notation as in Chapter \ref{C5} the free energy functional is given by
\begin{equation}
\label{freeFullLast}
\Phi=\Phi(h,x)=\overline{\Sigma^{(c)}(\Psi^{\text{(v)}}(1,\bm{\omega}),x)}-\frac{c}{2}\overline{\Sigma^{(e)}(\Psi^{\text{(e)}}(1,\bm{\omega}),x)}
\end{equation}
where the symbols $\Sigma^{(c)}$ and $\Sigma^{(2)}$ are the RSB expectations of the claims $\Psi^{\text{(v)}}$ and $\Psi^{\text{(e7v)}}$, defined in \ref{definitions} in Chapter \ref{C5}, defined, respectively, in the probability space of $c$ or $2$-components vectorial Brownian motion, and the bar $\overline{\Cdot}$ is the average over the random couplings.

The equilibrium free energy $F$ is finally derived by the physical variational problem, through the minimization of the functional $\Phi$ with respect to the functional $h$:
\begin{equation}
F=-\beta \min_{h} \Phi(h)\,.
\end{equation}
The minimum is taken on the space of all functional $h:\Omega\to \mathbb{R}$, where $\Omega$ is the space of scalar contnuous function $h:[0,1]\to \mathbb{R}$.

By symmetry arguments, we can also guess that $h$ should be anti-symmetric
\begin{equation}
h(1,\omega)=-h(1,-\omega),\forall \omega \in \Omega
\end{equation}
and bounded. The boundedness criterion arises from the fact that $h$ should represent the cavity field, induced by a finite number of spins.

\section{Self consistency equation}
We obtain the self-consistency equations by imposing the stationary condition with respect to the order parameters, by putting to $0$ the first variation of the functional $\Phi$ with respect to the cavity functional $h$ and the POP $x$. 

For both the edge and the vertex contribution, let us introduce the vertex and edge cavity magnetizations $\bm{m}^{\text{(e)}}(1,\omega_1,\omega_2 )$ and $\bm{m}^{\text{(v)}}(1,\omega_1,\cdots,\omega_c)$, defined by:
\begin{equation}
\bm{m}^{\text{(e)}}(1,\omega_1,\cdots,\omega_c)=\bm{M}^{\text{(e)}}\big(\,h(\omega_1),h(\omega_2) \,\big)
\end{equation}
\begin{equation}
\bm{m}^{\text{(v)}}(1,\omega_1,\cdots,\omega_c)=\bm{M}^{\text{(v)}}\big(\,h(\omega_1),\cdots,,h(\omega_c) \,\big)
\end{equation}
with
\begin{equation}
\bm{M}^{\text{(e/v)}}\big(\,\bm{y}\,\big)=\nabla\log \,\Delta^{\text{(e/v)}} \left(\,\bm{y}\right)\,, \quad \text{with}\,\,y\in \mathbb{R}^{(2/c)},
\end{equation}
where the symbol $\nabla$ denotes the gradient over the variables $\bm{y}$. Note that $\|\bm{M}^{\text{(v)}}\|\leq c $ and  $\|\bm{M}^{\text{(e)}}\|\leq 2$.

The derivative of the free-energy with respect the cavity field functional can be simply obtained by computing the free energy with respect a "perturbed" cavity field functional of the form
\begin{equation}
h^{\alpha}(1,\bm{\omega})=h(1,\bm{\omega})+\alpha \delta(1,\bm{\omega})\,,
\end{equation}
where $\delta$ is any bounded process, and deriving over $\alpha$, at  $\alpha=0$. The derivative is simply give by formula \eqref{first_der} in Section \ref{sec6.1}.

In such a way we have
\begin{equation}
\partial_{\alpha}\Psi^{\text{(v)}}(1,\bm{\omega})|_{\alpha=0}=\sum^c_{i=1}\delta(1,\omega_i)m_i^{\text{(v)}}(1,\omega_1,\cdots,\omega_c)
\end{equation}
and
\begin{equation}
\partial_{\alpha}\Psi^{\text{(e)}}(1,\bm{\omega})|_{\alpha=0}=\sum^2_{i=1}\delta(1,\omega_i)m_i^{\text{(e)}}(1,\omega_1,\omega_2)
\end{equation}
and then
\begin{multline}
\label{Eqh}
\frac{1}{2}\sum^2_{i=1}\overline{\mathbb{E}_{(\mathbb{W}_{\nu})^{\otimes 2}}\left[\,\mathcal{E}\big(x \bm{r}^{\text{(e)}};1,\bm{\omega}\big) m_i^{\text{(e)}}(1,\omega_1,\omega_2)\,\delta(1,\omega_i)\,\right]}\,=\\
\frac{1}{c}\sum^c_{i=1}\overline{\mathbb{E}_{(\mathbb{W}_{\nu})^{\otimes c}}\left[\,\mathcal{E}\big(x \bm{r}^{\text{(v)}};1,\bm{\omega}\big)m_i^{\text{(v)}}(1,\bm{\omega})\,\delta(1,\omega_i) \,\right]}\,,
\end{multline}
for all perturbing bounded functional $\delta:\Omega\to \mathbb{R}$.

Let $\mathcal{F}_1^{(i)}$ be the $\sigma-$algebra generated by the component $\omega_i$ of the Brownian motion, then the above equation yields
\begin{multline}
\label{Eqh2}
\frac{1}{2}\sum^2_{i=1}\overline{\mathbb{E}_{(\mathbb{W}_{\nu})^{\otimes 2}}\left[\,\mathcal{E}\big(x \bm{r}^{\text{(e)}};1,\bm{\omega}\big) m_i^{\text{(e)}}(1,\omega_1,\omega_2)\,\Big|\mathcal{F}_1^{(i)}\right]}\,=\\
\frac{1}{c}\sum^c_{i=1}\overline{\mathbb{E}_{(\mathbb{W}_{\nu})^{\otimes c}}\left[\,\mathcal{E}\big(x \bm{r}^{\text{(v)}};1,\bm{\omega}\big)m_i^{\text{(v)}}(1,\bm{\omega})\,\delta(1,\omega_i) \Big|\mathcal{F}_1^{(i)}\right]}\,.
\end{multline}
The conditional expectation $\mathbb{E}_{(\mathbb{W}_{\nu})^{\otimes (2/c)}}[\Cdot|\mathcal{F}_1^{(i)}]$ means that we take the expectation over all the component of the Brownian motion $\bm{\omega}$ except $\omega_i$, and we impose that  $\omega_i$ is the same in both sides of the above equation. Using a notation with functional Dirac delta, we have
\begin{multline}
\label{Eq3}
\frac{1}{2}\sum^2_{i=1}\overline{\mathbb{E}_{(\mathbb{W}_{\nu})^{\otimes 2}}\left[\,\mathcal{E}\big(x \bm{r}^{\text{(e)}};1,\bm{\omega}\big) m_i^{\text{(e)}}(1,\omega_1,\omega_2)\,\delta_{[0,1]}[\omega_i-u ]\,\right]}\,=\\
\frac{1}{c}\sum^c_{i=1}\overline{\mathbb{E}_{(\mathbb{W}_{\nu})^{\otimes c}}\left[\,\mathcal{E}\big(x \bm{r}^{\text{(v)}};1,\bm{\omega}\big)m_i^{\text{(v)}}(1,\bm{\omega})\,\delta_{[0,1]}[\omega_i-u]\,\right]}\,,
\end{multline}
where the symbol $\delta_{[0,1]}[\Cdot]$ is the functional Dirac delta over the realization of the considered process $\omega_i$ in all the interval $[0,1]$.

Using the relation \eqref{identity_DDE} and exploiting the symmetry over the vertex index, after coupling average, we get:\begin{multline}
\label{Eq4}
\overline{\mathbb{E}_{(\mathbb{W}_{\nu})^{\otimes 2}}\left[\,e^{-\int^1_0 dq P(q)\phi^{(e)}(q,\omega_1,\omega_2)}\tanh(\beta J_{1,2})\tanh(\beta h_2(1,\omega_2))\,\,\Big|\mathcal{F}_1^{(1)}\,\right]}\,=\\
\overline{\mathbb{E}_{(\mathbb{W}_{\nu})^{\otimes c}}\left[\,e^{-\int^1_0 dq P(q)\phi^{(v)}(q,\omega_1,\cdots,\omega_c)}\frac{\sinh\left(\sum^c_{i=2}\beta u(J_{0,i},h_i(1,\bm{\omega}))\,\right)}{\prod^c_{i=2} \cosh(\beta u(J_{0,i},h_i(1,\bm{\omega}))}\Bigg|\mathcal{F}_1^{(1)}\,\right]}\,,
\end{multline}
where the function $u(J,h)$ is defined in \eqref{U} and the symbol $J_{ij}$ are the random couplings.

Finally, form the derivation over the POP, we have
\begin{equation}
\frac{1}{2}\sum^2_{i=1}\overline{\mathbb{E}_{(\mathbb{W}_{\nu})^{\otimes 2}}\left[\,\mathcal{E}\big(x \bm{r}^{\text{(e)}};q,\bm{\omega}\big) \big(\,r_i^{\text{(e)}}(q,\omega_1,\omega_2)\,\big)^2\,\,\right]}\,=\\
\frac{1}{c}\sum^c_{i=1}\overline{\mathbb{E}_{(\mathbb{W}_{\nu})^{\otimes c}}\left[\,\mathcal{E}\big(x \bm{r}^{\text{(v)}};q,\bm{\omega}\big)\big(\,r_i^{\text{(v)}}(q,\bm{\omega})\big)^2\,\right]}\,,
\end{equation}
The equilibrium free energy is finally given by replacing the solution $(h,x)$ in the functional \eqref{freeFullLast}.
\part{Conclusions}
\chapter*{Conclusions and outlooks}
\addcontentsline{toc}{chapter}{Conclusions and outlooks}
\thispagestyle{empty}
\markboth{Conclusions and outlooks}{}

In this thesis, the first non-ambiguous description of the full$-$RSB formalism for the Ising spin glass on random regular graphs is obtained.

The order parameter is a stochastic functional, related to the distribution of the cavity fields populations at each site \eqref{AustinRep}, in contrast with the order parameter in the Mézard-Parisi ansatz which involves a hierarchical tower of distributions of distributions, hard to extend to the full$-$RSB limit.

Since the order parameter is a functional, the extension of the discrete-RSB scheme to the continuous case cannot be achieved simply through a (generalized) Parisi-like partial differential equation.

This problem has been overcome by formalizing the ideas suggested by G. Parisi, in \cite{ParisiMarginal}, in a proper new stochastic variational approach, that provides a powerful mathematical tool to study such class of models.

The approach proposed in this thesis allows deriving a well-defined free energy functional from which the self-consistency mean-field equations can be easily derived.

We get, then, a complete definition of the full$-$RSB problem for the Ising spin glass on random regular graph, with a given free energy, depending on certain order parameters, and the equation for the order parameters.

The mathematical properties of the variational functional are deeply studied.

We guess that the solution of the presented mean-field equations provides the right equilibrium free energy, but a rigorous proof is still missing. It is worth noting, however, that the mathematical formalization of the full$-$RSB scheme, proposed in that work, is a necessary groundwork for a rigorous analysis of the RSB phenomenon-beyond the fully connected models.

Unfortunately, the mean-field equations are very difficult to solve. Since the order parameter is a functional of a Brownian motion, we guess that the self-consistency equation may be solved by a population dynamics over populations of Brownian motion paths.

The quantitative evaluation of the free energy and a deeper analysis of the physics properties will be investigated in next works.

\end{document}